\documentclass[12pt]{article}
\usepackage{amsmath}
\usepackage{graphicx}
\usepackage{enumerate}
\usepackage{url} % not crucial - just used below for the URL 
\usepackage{mathrsfs}
\usepackage{amsfonts}
\usepackage{amssymb}
\usepackage{amsthm}
\usepackage{cases}
\usepackage{indentfirst}
\usepackage{epsfig}
\usepackage {graphicx}
\usepackage{color}
\usepackage{CJK}
\usepackage{hyperref}
\usepackage{setspace}
\usepackage{algorithm}
\usepackage{algorithmicx}
\usepackage{algpseudocode}
\usepackage{natbib}
\usepackage{soul}
\usepackage{xr}
\usepackage{comment}
\usepackage{multirow}
\usepackage{lipsum}

\newcommand{\blind}{0}
\newtheorem{theorem}{Theorem}[section]

\newtheorem{lemma}{Lemma}[section]
\newtheorem{condition}{Condition}[section]

\newtheorem{definition}{Definition}[section]
\numberwithin{equation}{section}
\renewcommand{\theequation}{\thesection.\arabic{equation}}
\begin{document}
\def\htau{\hat{\tau}}
\def\calI{\mathcal{I}}
\def\calT{\mathcal{T}}
\def\calB{\mathcal{B}}
\def\by{\mathbf{y}}
\def\bX{\mathbf{X}}
\def\beps{\boldsymbol{\epsilon}}
\def\btheta{\boldsymbol{\theta}}
\def\bM{\mathbf{M}}
\def\bI{\mathbf{I}}
\def\tbX{\tilde{\mathbf{X}}}
\def\bSigma{\boldsymbol{\Sigma}}
\def\bs{\mathbf{s}}
\def\bS{\mathbf{S}}
\def\tbU{\tilde{\mathbf{U}}}
\def\calG{\mathcal{G}}
\def\tbtheta{\tilde{\boldsymbol{\theta}}}
\def\bw{\mathbf{w}}
\def\bZ{\mathbf{Z}}
\def\bP{\mathbf{P}}
\def\hgamma{\hat{\gamma}}
\def\supp{\operatorname{supp}}
\def\calC{\mathcal{C}}
\def\hbSigma{\widehat{\boldsymbol{\Sigma}}}
\def\hsigma{\hat{\sigma}}
\def\hcalA{\widehat{\mathcal{A}}}
\def\calS{\mathcal{S}}
\def\calH{\mathcal{H}}
\def\hbbeta{\hat{\boldsymbol{\beta}}}
\def\tbbeta{\tilde{\boldsymbol{\beta}}}
\def\htbeta{\hat{\tilde{\boldsymbol{\beta}}}}
\def\tr{\operatorname{tr}}
\def\calA{\mathcal{A}}
\def\hcalS{\hat{\mathcal{S}}}
\def\httheta{\hat{\tilde{\boldsymbol{\theta}}}}
\def\hbtheta{\hat{\boldsymbol{\theta}}}
\def\hpi{\hat{\pi}}
\def\bX{\mathbf{X}}
\def\calZ{\mathcal{Z}}
\def\calY{\mathcal{Y}}
\def\bx{\mathbf{x}}
\def\calE{\mathcal{E}}
\def\calM{\mathcal{M}}
\def\bM{\mathbf{M}}
\def\bm{\mathbf{m}}
\def\bo{\mathbf{o}}
\def\be{\mathbf{e}}
\def\bu{\mathbf{u}}
\def\bv{\mathbf{v}}
\def\bvartheta{\boldsymbol{\vartheta}}
\def\bO{\mathbf{O}}
\def\bE{\mathbf{E}}
\def\bV{\mathbf{V}}
\def\bU{\mathbf{U}}
\def\tbz{\tilde{\mathbf{z}}}
\def\bGamma{\boldsymbol{\Gamma}}
\def\tbeps{\tilde{\boldsymbol{\epsilon}}}
\def\calV{\mathcal{V}}
\def\teps{\tilde{\epsilon}}
\def\calL{\mathcal{L}}
\def\bzeta{\boldsymbol{\zeta}}
\def\bc{\mathbf{c}}
\def\P{\operatorname*{P}}
\def\E{\operatorname*{E}}
\def\Var{\operatorname*{Var}}
\def\Cov{\operatorname*{Cov}}
\def\Cor{\operatorname*{Cor}}
\def\sign{\operatorname*{sign}}
\def\R{\operatorname*{\mathbb{R}}}
\def\F{\operatorname*{\mathcal{F}}}
\def\Rank{\operatorname*{Rank}}
\def\diag{\operatorname*{diag}}
\def\swap{\operatorname*{swap}}
\def\FDR{\operatorname{FDR}}
\def\Power{\operatorname{Power}}
\def\Pr{\operatorname{Pr}}
\def\D{\operatorname*{D}}
\def\I{\operatorname{I}}
\def\FDP{\operatorname{FDP}}
\def\red{\color{red}}
\def\blue{\color{blue}}
\def\cyan{\color{cyan}}
\newcommand{\bbeta}{\boldsymbol{\beta}}
\def\mb{\mathbf{b}}
\def\bD{\mathbf{D}}
\def\bd{\mathbf{d}}
\def\htheta{\hat{\theta}}
\def\ttheta{\tilde{\theta}}
\def\bmu{\boldsymbol{\mu}}
\def\tby{\tilde{\mathbf{y}}}
\def\ba{\mathbf{a}}
\def\bv{\mathbf{v}}
\def\bA{\mathbf{A}}
\def\bT{\mathbf{T}}
\def\tbM{\hat{\mathbf{M}}}
\def\hJ{\hat{J}}
\def\bE{\mathbf{E}}
\def\bW{\mathbf{W}}
\def\bgamma{\boldsymbol{\gamma}}
\def\hbv{\hat{\mathbf{v}}}
\def\bY{\mathbf{Y}}
\def\bz{\mathbf{z}}
\def\hbM{\hat{\mathbf{M}}}
\def\hG{\hat{G}}
\def\hV{\hat{V}}
\def\vec{\operatorname{vec}}
\def\bdel{\boldsymbol{\delta}}
\def\hK{\hat{K}}
\def\tG{\tilde{G}}
\def\PWER{\operatorname{PWER}}
\def\tbc{\tilde{\mathbf{c}}}
\def\bg{\mathbf{g}}
\def\PFER{\operatorname{PFER}}
\def\CUS{\operatorname{CUS}}
\def\FD{\operatorname{FD}}
\def\TD{\operatorname{TD}}
\def\hSigma{\operatorname{\hat{\Sigma}}}
\def\hGamma{\operatorname{\hat{\Gamma}}}
\def\CV{\operatorname{CV}}
\def\vech{\operatorname{vech}}

\def\spacingset#1{\renewcommand{\baselinestretch}%
{#1}\small\normalsize} \spacingset{1}

%%%%%%%%%%%%%%%%%%%%%%%%%%%%%%%%%%%%%%%%%%%%%%%%%%%%%%%%%%%%%%%%%%%%%%%%%%%%%%

\if0\blind
{
  \title{\bf OPTICS: Order-Preserved Test-Inverse Confidence Set for Number of Change-Points}
  \author{Ao Sun\hspace{.2cm}\\
    Data Sciences and Operations Department, University of Southern California\\
    and \\
    Jingyuan Liu \\
     Department of Statistics, Xiamen University}
  \maketitle
} \fi

\if1\blind
{
  \bigskip
  \bigskip
  \bigskip
  \begin{center}
    {\LARGE\bf OPTICS: Order-Preserved Test-Inverse Confidence Set for Number of Change-Points}
\end{center}
  \medskip
} \fi

\bigskip
\begin{abstract}
	Determining the number of change-points is a first-step and fundamental task in change-point detection problems, as it lays the groundwork for subsequent change-point position estimation. While the existing literature offers various methods for consistently estimating the number of change-points, these methods typically yield a single point estimate without any assurance that it recovers the true number of changes in a specific dataset. Moreover, achieving consistency often hinges on very stringent conditions that can be challenging to verify in practice. To address these issues, we introduce a unified test-inverse procedure to construct a confidence set for the number of change-points. The proposed confidence set provides a set of possible values within which the true number of change-points is guaranteed to lie with a specified level of confidence. We further proved that the confidence set is sufficiently narrow to be powerful and informative by deriving the order of its cardinality. Remarkably, this confidence set can be established under  more relaxed conditions than those required by most point estimation techniques. {We also advocate multiple-splitting procedures to enhance stability and extend the proposed method to heavy-tailed and dependent settings.}  As a byproduct, we may also leverage this constructed confidence set to assess the effectiveness of point-estimation algorithms. Through extensive simulation studies, we demonstrate the superior performance of our confidence set approach. Additionally, we apply this method to analyze a  bladder tumor microarray dataset. {Supplementary Material, including proofs of all theoretical results, computer code, the R package, and extended simulation studies, are available online.}
\end{abstract}

\noindent%
\vfill

\newpage
\spacingset{1.75} % DON'T change the spacing!

\section{Introduction}
\label{sec:intro}
\subsection{Motivation and intuition}
Estimating the number of change-points is a fundamental task in change-point detection problems, as a consistent estimation of the number usually leads to consistent estimation of change locations \citep{harchaoui2010multiple, wang2021statistically}. Therefore, the consistency in change-point detection problems can typically be formulated as 
\begin{equation}
	\label{equ:consistency}
	\liminf_{n \to \infty}\Pr\{\hK = K^*\} = 1,
\end{equation}
where $\hK$ is the estimated number of change-points, and $K^*$ is the true number of change-points. Classical methods for obtaining a consistent $\hK$ are mainly based on the Bayesian information criterion (BIC, \cite{schwarz1978estimating}). See for instance \cite{yao1988estimating, bai1998estimation, braun2000multiple,zhang2012model, fryzlewicz2014wild} and \cite{cho2015multiple}. Additionally, some recent approaches,  such as \cite{padilla2019optimal, wang2021statistically, padilla2021optimal},  have embraced a hard threshold approach. Both the BIC-based methods and hard-threshold techniques require parameter tuning, and the consistency of their estimates highly relies on these tuning parameters. To mitigate this dependence on tuning, \cite{zou2020consistent} introduced an Order-Preserved Sample-Splitting Procedure (COPSS) that selects the number of change-points by optimizing out-of-sample prediction, thus is tuning-free.

{ Nevertheless, these methods only provide point estimates for the number of change-points, without any guarantee that the true number $K^*$ can be recovered on a finite-sample basis. As a result, the consistency result may not be sufficient in practice.} To intuitively illustrate, we conduct a simple simulation for COPSS, using Binary Segmentation (BS, \cite{fryzlewicz2014wild}) as the change-position detector. The detailed model settings are provided in Section \ref{sec:simulation}. According to COPSS, the true number $K^*$ can be correctly identified only when $K^*$ is the minimizer of the out-of-sample prediction error; that is, in an ascending order,  $K^*$ should rank the first in terms of prediction error among all the candidate numbers of change-points. Left picture in Figure \ref{fig:illus1} depicts the piechart of such ranks for $K^*$ over 500 simulation runs. We observe that in only 55.8\% runs, $K^*$ can rank the first and hence be successfully recovered. 

In brief, a single-point estimate can be unreliable; a misjudgment in the number of change-points could further result in erroneous estimations of change positions. To this end, inspired by the concept of confidence interval, a more prudent approach is constructing a confidence set, denoted as $\calA$, for $K^*$, with a specified confidence level $1- \alpha$, such that
$$
\liminf_{n \to \infty}\Pr\{K^* \in \calA\} \ge 1 - \alpha.
$$
This task appears challenging since the construction of confidence sets for the number of change-points demands an investigation into the asymptotic behavior of a discrete random variable. However, Right picture in Figure \ref{fig:illus1} offers an insight for this problem. The figure depicts a histogram of differences in prediction errors between the models fitted with the true number of change-points $K^*$ and those fitted with the minimizer $\hK$. We observe that while the models with $K^*$ may not always lead to optimal out-of-sample prediction performance, they typically exhibit relatively small deviations from the optimum. In light of this and motivated by \cite{lei2020cross}, we address the challenge of constructing a confidence set by framing it as a testing problem. In this testing framework, the null hypothesis posits that the selected number is indeed optimal from the perspective of out-of-sample prediction error. Note that in hypothesis testing, the null hypothesis is not rejected unless there exists substantial evidence to the alternative. Therefore, although the true number need not be the minimizer, it will likely not be rejected unless the observed difference is statistically significant. We then collect those numbers that are not rejected to form the confidence set $\calA$. When the significance level is set as $\alpha$, we will show that the true number of change-points lies in the confidence set with probability at lease $1-\alpha$. 
\begin{figure}[h]
	\centering
	\includegraphics[width=4.5cm, height=4.5cm]{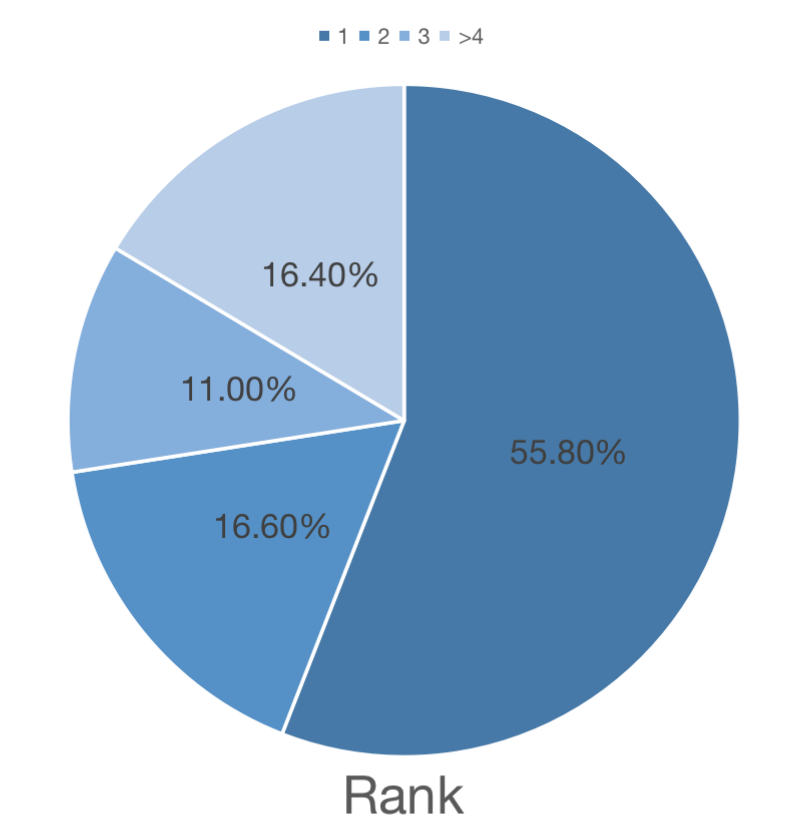}
	\includegraphics[width=6cm, height=4cm]{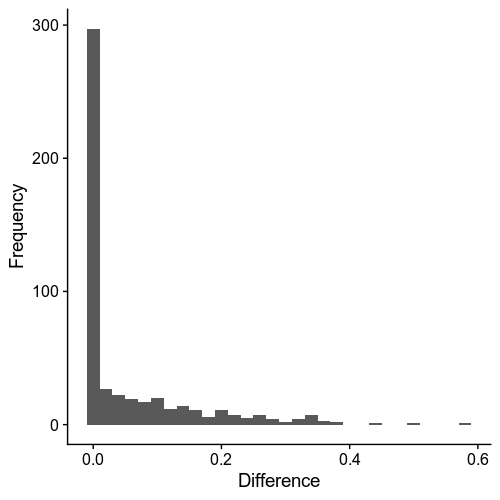}
	\caption{Left figure: Piechart for ranks of $K^*$ in an ascending order according to prediction error. Right figure: Histogram of the difference in the out-of-sample prediction errors between the models fitted with the true number of change-points $K^*$ and those fitted with the minimizer $\hK$. \label{fig:illus1}}
\end{figure}

Apart from the desirable coverage probability, we also delve into the cardinality of the proposed confidence set, which reflects the power of tests from the perspective of statistical inference, or the rate of false negatives in the context of model selection. Therefore, the theory of cardinality holds its own significance. We establish that, with an overwhelming probability, the selected confidence set possesses nontrivial power with a bounded cardinality.

Furthermore, the cardinality of the confidence set can serve as a metric for assessing the efficacy and stability of a change-point detection algorithm in the detection stage. When an inefficient algorithm is employed, it becomes challenging to distinguish the true number of change-points $K^*$ from other potential choices, resulting in the confidence set containing many candidates to achieve the desired coverage rate. Conversely, with an efficient and stable detection algorithm, the true number $K^*$ often exhibits significantly superior performance compared to other choices. Consequently, the null hypothesis will only be retained when $K$ is closely aligned with $K^*$, leading to a small cardinality for the confidence set.

\subsection{Our contributions}

Meanwhile, some works in the change-point literature could potentially facilitate constructing confidence intervals or sets for change locations. For instance, \cite{yao1989least} derived the asymptotic distribution of change locations in the context of one-dimensional mean change problems. Similarly, \cite{bai1998estimating} investigated the asymptotic behavior of change locations in structural breaking models. However, the asymptotic distributions in these studies rely on unknown population-level parameters, such as the means and variances of noises, as well as the true number of change-points. To the best of our knowledge, no research has delved into the asymptotic distributions of the number of change-points. 

In light of this, we propose a test-based framework to circumvent the asymptotic distribution of number of change-points. This framework was inspired by \cite{lei2020cross}, which focused on linear models and constructed the confidence interval for number of active predictors via cross-validation. However, it's crucial to recognize the substantial differences between linear models and change-point detection problems in both methodology and the whole theoretical foundation; even the conventional cross-validation techniques are not applicable in the context of change-point detection. For instance, in change-point problems, both the mean values and variances of individual data points might be non-stationary. Hence, obtaining a consistent estimator of the covariance matrix becomes challenging. This, in turn, causes failure of the Gaussian multiplier bootstrap procedure, which is crucial for constructing confidence intervals. Furthermore, \cite{lei2020cross} did not take into account the length of their proposed confidence intervals, which in turn impacts the power of tests. Consequently, if we include an excessive number of candidate values, the resulting interval could become uninformative even if the coverage rate is guaranteed. 

In this paper, we establish the confidence set for the number of change-points by incorporating order-preserved splitting and multiplier bootstrap, and systematically establish the theoretical properties of the proposed set in the change-point framework. Interestingly, the theory of the proposed confidence set can be built upon weaker conditions than the corresponding point estimates of number of change-points. Additionally, beyond the scope of coverage rate, we derive a sharp upper bound for the cardinality of the confidence set, which plays a crucial role in controlling the false negative rate and ensuring that the associated test has nontrivial power. { The proposed OPTICS method is then extended to handle heavy-tailed and $m$-dependent settings. The effectiveness of these extensions is verified through simulations. As a byproduct, we provide an easy-to-implement R package, \textit{OPTICS}, available on GitHub to implement the proposed method.} \footnote{\url{https://github.com/suntiansheng/OPTICS}}

%First, the random data-splitting used in their algorithm is not suitable for detecting change-points since the order structure is critical information. Secondly, the comparison criterion $\calC_K$ used in their method for linear models is unsuitable for detecting change-points. Therefore, we should take effort to redefine a suitable criteria for change-point detection and derive its asymptotic behavior. Finally,  \cite{lei2020cross} assumed that the data points are independent and identically distributed (i.i.d) in their theory. In contrast, change-points occur at least when the moments of two adjoining data points differ, which makes the i.i.d assumption unrealistic. As a result, we need to prove our results using newly developed technical lemmas.

\subsection{Notations}
We introduce the following notations used throughout this paper. For two sequences $a_n$ and $b_n$, $a_n \lesssim b_n$ ($a_n \gtrsim b_n$) means with probability approaching one, $a_n\le c b_n $ ($b_n\le c a_n $) for some $c > 0$ and sufficiently large $n$; 
{$a_n \gg b_n$ stands for $a_n/b_n\rightarrow\infty$}. 
$a:=b$ represents that $a$ is defined as $b$. Given a vector $\mathbf{x} = [x_1, \ \ldots, \ x_d]^{\top} \in \mathbb{R}^d$, its $\ell_2$-norm is defined as $\|\mathbf{x}\|_2 = \big(\sum_{j=1}^d |x_j|^2\big)^{1/2}$, and its $\ell_\infty$-norm is $\| \bx\|_{\infty}=\max_{j=1,\ldots, d}|x_j|$. For a random variable $X$, its Orlicz norm $\| X\|_{\psi_{\beta}} = \inf\left\{C > 0: \E\left[ \psi_{\beta}\left(|X|/C\right) \right] \le  1\right\}$, where $\psi_{\beta}(x) := \exp(x^{\beta}) - 1$ for $\beta = 1,2$. Let $\calT_K = \{\tau_1, \tau_2, \ldots, \tau_K\}$ be a set of  $K$ change-points for a size-$m$ sequence $\{\bx_i, i=1,2,\ldots, m\}$, where $\tau_0 < \tau_1 < \ldots < \tau_K$. Denote $\bar{\bx}_{\tau_{k},\tau_{k+1}}$ as the average of $\{\bx_i, i = \tau_k+1, \tau_k+2, \ldots, \tau_{k+1}\}$, and set
$
\calS_{\bx}^2(\calT_K) = \sum_{k = 0}^K \sum_{i = \tau_{k}+1}^{\tau_{k+1}} \|\bx_i - \bar{\bx}_{\tau_{k},\tau_{k+1}} \|_2^2
$. 
Moreover, let $\tilde{\calT}_{\tilde{K}}$ be another set of $\tilde{K}$ change-points. Denote $\calS_{\bx}^2(\calT_K \cup \tilde{\calT}_{\tilde{K}})  = \calS_{\bx}^2(\operatorname{sort}(\calT_K \cup \tilde{\calT}_{\tilde{K}})) $, where $\operatorname{sort}(\calA)$ represents the set of sorted elements of $\calA$ in ascending order. { For a matrix $\bX$, $\vech(\bX)$ be the vectorization of the lower half matrix of $\bX$, otherwise, if $\bx$ is a vector, we define $\vech(\bx) = \bx$. We denote $\lambda_{\min}(\bX)$ and $\lambda_{\max}(\bX)$ as the smallest and largest eigenvalues of the matrix $\bX$, respectively.} For a set $\mathcal{A}$, $|\mathcal{A}|$ denotes its cardinality.  $\mathrm{1}_{\mathcal{A}}$ is the indicator function that takes value 1 when $\mathcal{A}$ is true and 0 otherwise. 

\subsection{Organization of the paper}

The rest of the paper is organized as follows. In Section \ref{sec2}, we propose an Order-Preserved Test-Inverse Confidence Set (OPTICS) for the number of change-points -  the methodology, intuition, algorithm and practical guidelines are provided. In Section \ref{sec:3}, we systematically study the theoretical properties of the OPTICS, including but not limited to the coverage rate and the asymptotic bound for the cardinality of OPTICS. The finite-sample performance of OPTICS is empirically verified through several simulation studies in Section \ref{sec:simulation}. In Section \ref{sec:5}, we apply construct the OPTICS for a bladder
tumor microarray dataset. {Section \ref{sec:6}} concludes the paper.  The Supplementary Material contains the proofs of theoretical results, the additional literature review, simulations and real data results of the main paper.

\section{Methodology: OPTICS for number of change-points}\label{sec2}
\subsection{A general change-point model and its score transformation}
Suppose a sequence of independent data observations $\calZ = \{\bz_1, \ldots, \bz_{2n}\}$ are collected from the following multiple change-point model:
\begin{equation}
	\label{sec1:e1}
	\mathbf{z}_i \sim m\left(\cdot \mid \boldsymbol{\beta}^*_k\right), \quad \tau_{k-1}^*<i \leq \tau_{k}^*, k=1, \ldots, K^*+1; \  i=1, \ldots, 2n.
\end{equation}
In model (\ref{sec1:e1}), the sample size is set to be $2n$ for the later notation convenience. $K^*$ is the true number of change-points that is allowed to vary with $n$. $\tau_k^{*'}s$ are the locations of true change-points; for convention, set $\tau_0^*=0$ and $\tau_{K^*+1}^* = 2n$, so that the change-points in $\calT^* = \{\tau_1^*, \ldots, \tau_{K^*}^*\}$ partition the $2n$ sample points into $k+1$ segments. $m(\cdot\mid \bbeta_k^*)$ represents a certain model structure for the $k$th segment, with a $d$-dimensional parameter vector $\bbeta_k^*$, where $\bbeta_{k}^* \ne \bbeta_{k+1}^*$. {We assume $d$ is fixed throughout the manuscript.} This model setting is quite general, including mean changes \citep{hao2013multiple}, variance changes \citep{chen1997testing}, structural changes in regression model \citep{bai1998estimating}, Covariance change-points model \citep{aue2009break} and network change-points model \citep{wang2021optimal, fan2025mosaic}. See more discussio in Section~\ref{sup:loss} of Supplementary Material.
%See more discussion in \cite{zou2020consistent}. 
Our primary goal in this paper is to construct a confidence set for $K^*$. 

Inspired by \cite{zou2020consistent}, we embed the generic model \eqref{sec1:e1} into a multiple mean change-point detection problem via a score-type transformation. Specifically, let $\ell(\bbeta; \bz_i)$ be a plausible loss function for data point $\bz_i$, thus the score function can be defined as its derivative {$ \bs_{\bbeta}(\bz_i) = \vech( \partial \ell(\bbeta; \bz_i)/{\partial \bbeta})$}. Intuitively, if the score is identifiable, for $i \in (\tau_{k-1}^*, \tau_{k}^*]$ and $i^\prime \in (\tau_{k}^*, \tau_{k+1}^*]$, we should have $\E\{\bs_{\bgamma}(\bz_i)\}  = 0$ if and only if $\bgamma = \bbeta^*_{k}$, and $\E\{\bs_{\bgamma}(\bz_{i^\prime})\}  = 0$ if and only if $\bgamma = \bbeta^*_{k+1}$. Hence, given a fixed $d$-dimensional vector $\bgamma$, typically $\E\{\bs_{\bgamma}(\bz_i)\} \ne \E\{\bs_{\bgamma}(\bz_{i^\prime})\}$.  This motivates us to decompose the score into 
\begin{equation}
	\label{model:score}
	\bs_i :=\bs_{\bgamma}(\bz_i) = \bmu_i + \beps_i, i =1,\ldots, 2n,  
\end{equation}
where $\bmu_i = \E[\bs_{\bgamma}(\bz_i)]$, and $\beps_i = \bs_{\bgamma}(\bz_i)  - \E[\bs_{\bgamma}(\bz_i)]$. Denote the covariance matrix $\Cov(\beps_i)  = \bSigma^{(k)}$ when $i \in (\tau_k^*, \tau_{k+1}^*]$. Note that the score-type transformation generally remains invariant regardless of the choice of $\bgamma$. Hence, we can choose any $\bgamma$, such as $\bgamma=0$ or $\bgamma := \arg\min_{\bbeta} \sum_{\bz_i \in \calZ} \ell(\bbeta; \bz_i)$.

\subsection{Construction of OPTICS}

In this subsection, we propose an order-preserved test-inverse confidence set (OPTICS)  for the number of change-points. Let $\calM = \{1, 2, \ldots, K_{\max}\}$ be a tentative candidate set of the number of change-points in model \eqref{sec1:e1}, where $K_{\max} > K^*$ is allowed to increase with the sample size $2n$; a commonly adopted convention for $K_{\max}$ is $K_{\max}=\log(n)$ \citep{zou2020consistent}. Our objective is to identify a small subset of $\calM$ that covers the true number $K^*$ with a specified level of confidence. To begin with, we  introduce a criterion for evaluating the model-fitting capability of each $K\in\calM$, utilizing an order-preserved data-splitting technique \citep{zou2020consistent}. 

Considering the intrinsic order structure of change-point problems, we divide the data into the following ``Odd sample" $\mathcal{Z}_O$  and ``Even sample" $\mathcal{Z}_E$ based on the parity of temporal order: 
$$
\mathcal{Z}_O=\left\{\bz_{2 i-1}, i=1, \ldots, n\right\} \quad \text { and } \quad \mathcal{Z}_E=\left\{\bz_{2 i}, i=1, \ldots, n\right\}.
$$
Given each $K \in \calM$, certain base change-point detection algorithm can be adopted to estimate the change-position set $\calT_{K} = \{\tau_{1}^K, \ldots, \tau_{K}^K\}$ using the odd sample $\mathcal{Z}_O$. Then a natural criterion can be defined, based on the even sample $\mathcal{Z}_E$, as  
$$
\mathcal{C}^\prime\left(\mathcal{T}_K ; \mathcal{Z}_E\right):=\frac1n\mathcal{S}^2_{\bs^E}(\mathcal{T}_K)=\frac{1}{n}\sum_{k=0}^K \sum_{i=\tau_k^K+1}^{\tau_{k+1}^K}\|\bs_i^E - \bar{\bs}_{\tau_k^K,\tau_{k+1}^K}^E\|^2_2, 
$$
where $\bs_i^E$ is the score calculated using $\mathcal{Z}_E$ and $\bar{\bs}_{\tau_k^K,\tau_{k+1}^K}^E$ is the average of $\{\bs_i^E, \ i \in[\tau_{k}^K + 1, \tau_{k+1}^K]\}$. However, $\mathcal{C}^\prime\left(\mathcal{T}_K ; \mathcal{Z}_E\right)$ gauges the predictive performance within the even sample $\mathcal{Z}_E$ itself, wherein over-fitting always appears advantageous. In addition, the term $\bar{\bs}_{\tau_k^K,\tau_{k+1}^K}^E$ in $\mathcal{C}^\prime\left(\mathcal{T}_K ; \mathcal{Z}_E\right)$ renders $\|\bs_i^E - \bar{\bs}_{\tau_k^K,\tau_{k+1}^K}^E\|^2_2$ and $\|\bs_j^E - \bar{\bs}_{\tau_k^K,\tau_{k+1}^K}^E\|^2_2$ dependent for those $i$ and $j$ in the same sub-interval. Hence,  we replace $\bar{\bs}_{\tau_k^K,\tau_{k+1}^K}^E$ with its counterpart in $\mathcal{Z}_O$, denoted by $\bar{\bs}_{\tau_k^K,\tau_{k+1}^K}^O$, and refine the above criterion as
\begin{eqnarray}\label{CK}\mathcal{C}\left(\mathcal{T}_K ; \mathcal{Z}_E\right):=\frac{1}{n}\sum_{k=0}^K \sum_{i=\tau_k^K+1}^{\tau_{k+1}^K}\|\bs_i^E - \bar{\bs}_{\tau_k^K,\tau_{k+1}^K}^O\|^2_2.
\end{eqnarray}
Further let $\bar{\bs}_{K,i}^O = \sum_{k = 0}^{K} \mathrm{1}_{\{\tau_k^K+1, \tau_{k+1}^K\}} \bar{\bs}_{\tau_k^K,\tau_{k+1}^K}^O$, then \eqref{CK} can be rewritten as
$$
\mathcal{C}\left(\mathcal{T}_K; \calZ_E\right)  =n^{-1} \sum_{i =1}^{n} \| \bs_i^E  - \bar{\bs}_{K,i}^O\|_2^2.
$$
As discussed in Section 1, for the true number of change-points $K^*$, while $\mathcal{C}\left(\mathcal{T}_{K^*}; \calZ_E\right)$ might not be the minimum among all candidate models in $\mathcal{M}$, it is often reasonably close to the minimum. Formally, this motivates us to establish the following hypotheses for each candidate $K\in\calM$:
{\small\begin{equation}
		\label{equ:hypo}
		H_{0,K}: E[\mathcal{C}\left(\mathcal{T}_K; \calZ_E\right)] \text{ is the minimum in } \calM \quad \text{v.s.} \quad H_{1,K}: E[\mathcal{C}\left(\mathcal{T}_K; \calZ_E\right)] \text{ is not the minimum in } \calM.
\end{equation}}Following the philosophy of hypothesis testing, we do not reject $H_{0,K}$ unless $\mathcal{C}\left(\mathcal{T}_K; \calZ_E\right)$ significantly departs from the minimum. Then $H_{0,K^*}$, the null hypothesis corresponding to the true number $K^*$, is expected not rejected with an overwhelming probability. Therefore, the confidence set $\calA$ can be naturally defined as the collection of $K\in\calM$ whose $H_{0,K}$'s are not rejected. In other words, let $p_K$ be the associated $p$-value for the testing problem \eqref{equ:hypo}, then for a predetermined significance level $\alpha$, $$\calA := \{ K \in \calM: p_K > \alpha\}.$$  
By this means, the members in $\calA$ are statistically equivalent with respect to the criterion $\calC(\cdot,\mathcal{Z}_E)$; that is, all the number of change-points in this set are highly competitive. Indeed, as to be seen in Section \ref{sec:3}, the selected set $\calA$ covers the true number $K^*$ with probability at least $1-\alpha$: $$\Pr\{K^* \in \calA\} = \Pr\{p_{K^*} > \alpha\} \ge 1- \alpha.$$

The remaining task is to obtain the $p$-value $p_K$ associated with the testing problem \eqref{equ:hypo} for each $K\in\calM$. To accomplish this, for any $J, K\in \calM$, we further define $$\delta_{K, J} = \E[\mathcal{C}\left(\mathcal{T}_K; \calZ_E\right) - \mathcal{C}\left(\mathcal{T}_J; \calZ_E\right)].$$ Hence, the hypotheses in \eqref{equ:hypo} are equivalent to 
\begin{equation}
	\label{sec2:hypo}
	H_{0,K} : \max_{J \in \calM, J \ne K}  \delta_{K,J}\le 0\quad  \text{ v.s. } \quad H_{1,K} : \max_{J \in \calM, J \ne K}\delta_{K,J}> 0.
\end{equation}
One possible point estimate of $\delta_{K,J}$ in \eqref{sec2:hypo} is
\begin{eqnarray}\label{Xi}
	\hat{\delta}_{K, J}  =  \frac{1}{n}\sum_{i =1}^{n} \left(\| \bs_i^E  -\bar{\bs}_{K,i}^O\|^2_2 -  \| \bs_i^E  - \bar{\bs}_{J,i}^O\|^2_2\right) : = \frac{1}{n}\sum_{i =1}^{n} \xi_{K,J}^{(i)},
\end{eqnarray}
where $\xi_{K,J}^{(i)} := \| \bs_i^E  -\bar{\bs}_{K,i}^O\|^2_2 -  \| \bs_i^E  - \bar{\bs}_{J,i}^O\|^2_2, \ i = 1,\ldots, n$, are independent random variables given the odd sample $\mathcal{Z}_O$. Then naturally we can further take
$$
\max_{J \in \calM, J \ne K} \hat{\delta}_{K, J} = \max_{J \in \calM, J \ne K} \frac{1}{n} \sum_{i =1}^{n} \xi_{K,J}^{(i)},
$$
and the test statistic can be set as its studentized version: 
\begin{eqnarray}\label{TK}T_K = \max_{J \ne K}  \frac{\sqrt{n}\hat{\delta}_{K,J}}{\hat{\sigma}_{K,J}}=\max_{K\ne J} \frac{1}{\sqrt{n}} \sum_{i=1}^{n} \frac{\xi_{K,J}^{(i)}}{\hat{\sigma}_{K,J}},
\end{eqnarray}
where $\hat{\sigma}_{K,J}^2 = n^{-1}\sum_{i=1}^n (\xi_{K,J}^{(i)})^2$ is the estimated second moment.

Next, we calculate the $p$-values corresponding to $T_K$ using a Gaussian comparison and bootstrap method analogous to \cite{chernozhukov2013gaussian, chernozhukov2017central}. To be specific, we first generate independent standard Gaussian random variables $\zeta_i, \ i =1,\ldots, n$. Then for $b=1,\ldots, B$, where $B$ is the total number of bootstrap runs, define the $b$th bootstrap statistic as 
\begin{eqnarray}\label{bootT}
	{ T^{\sharp}_{K,b}} = \max_{K\ne J} \frac{1}{\sqrt{n}} \sum_{i=1}^{n} \frac{\xi_{K,J}^{(i)}}{\hat{\sigma}_{K,J}} \zeta_i.
\end{eqnarray}
Then the $p$-value is naturally set to be $\hat{p}_K = B^{-1} \sum_{b=1}^B \I({ T^{\sharp}_{K,b}} > T_K)$. In \eqref{bootT}, a non-centered bootstrap statistic is used. The reason is, as the random variables $\{\xi_{K,J}^{(i)},i=1,\ldots n\}$ may have varying means and variance, it is intricate to verify the sample variance converges to population variance under such a fluctuated scenario. 
However, the regular law of large numbers still implies that the sample second moment is expected to approximate its population counterpart. Notably, since the second moment serves as an upper bound of variance, the non-centered bootstrap statistic tends to be slightly conservative, while ensuring the type-I error.

In sum, the confidence set $\calA$ for the true number of change-points $K^*$ is constructed upon a test-based method with the order-preserved sample splitting technique. Thus, we name the confidence set to be Order-Preserved Test-Inverse Confidence Set (OPTICS). The entire procedure for obtaining OPTICS is summarized in the following steps. 
\begin{enumerate}
	\item \textbf{(Initialization)}. Given a proper { $\bgamma$ in \eqref{model:score}}, calculate the score functions $\bs_i$ for $i=1,\ldots, 2n$.
	\item For each given candidate number of change-points $K \in \calM$: 
	\begin{itemize}
		\item [2.1] \textbf{(Training)}. Obtain the estimated change-position set $\calT_K$ based on the odd sample $\calZ_O$. Compute the piecewise averages $\bar{\bs}_{\tau_k^K,\tau_{k+1}^K}^O$ in \eqref{CK}, or equivalently, $\bar{\bs}_{K,i}^O$. 
		
		\item [2.2] \textbf{(Validation)}. For $i = 1,\ldots, n$ and $K \ne J$, compute $\xi_{K,J}^{(i)}$ in \eqref{Xi} using the even sample $\calZ_E$. Further obtain the test statistic $T_K$ in \eqref{TK}.
		
		\item [2.3] \textbf{(Bootstrapping)}. For $b = 1,\ldots, B$, calculate the Gaussian multiplier bootstrap statistic { $T_{K,b}^{\sharp}$} in \eqref{bootT}, as well as the associated $p$-value $\hat{p}_K$.
	\end{itemize}
	\item \textbf{(OPTICS)}.  Given a significance level $\alpha$, the OPTICS is taken to be
	\begin{eqnarray}\label{set.est}
		\calA = \{ K \in \calM: \hat{p}_K > \alpha\}.
	\end{eqnarray}
\end{enumerate}

\subsection{Practical guidelines of OPTICS}\label{sec2:3}

{ 
\subsubsection{Choice of loss function and computational complexity}
The choice of the loss function and change-point detection methods is crucial in change-point analysis. Table \ref{tab:condition} in the Supplementary Material provides a comprehensive overview of settings for mean, variance, regression coefficient, nonparametric change-point models, covariance change-points model and Network change-points model \citep{fan2025mosaic}.

The overall computational complexity of OPTICS depends on the change-point detection methods and the bootstrap testing procedure. The complexity of change-point detection varies depending on the model and the specific detection method used (see \cite{truong2020selective} for details). For bootstrap testing, if we treat basic mathematical operations such as addition, subtraction, and multiplication as $O(1)$, each computation of \eqref{Xi} has a complexity of $O(n d_p)$, where $d_p$ is the dimension of $\bs_1$. Completing the entire testing procedure requires $O(B K_{\max} n d_p)$ operations, where $B$ is the number of bootstrap samples and $K_{\max}$ is the maximum number of change-points. However, by parallelizing the bootstrap procedure, the computational burden can be significantly reduced.
}

\subsubsection{Reduce a set to a single number of change-points}
The OPTICS produces a set of numbers of change-points $\calA$. It is eligible to guarantee the predetermined confidence and thus is typically more informative in practice. However, in some circumstances, a single estimated number is still desirable, especially when the ultimate goal is to estimate the change-positions.  One suggestion is to adopt the rightmost number $\bar{K}$ in $\calA$, as it is guaranteed to be larger than the true number $K^*$ with probability $1-\alpha$:
$$
\Pr\{\bar{K} \ge K^*\} \ge \Pr\{K^* \in \calA\} \ge 1-\alpha.
$$
$\bar{K}$ tends to be a slight overestimate of $K^*$. Additionally, the leftmost number $\underline{K}$ could also be chosen if the primary goal is to control the FWER, as
$$
\operatorname{FWER} = \Pr\{\underline{K} > K^*\}  \le  \Pr\{K^* \not\in \calA\} \le \alpha.
$$

Another suggestion is the post-hoc strategy, which incorporates the data-driven OPTICS with domain knowledge. That is, pick the member in $\calA$ with the most compelling scientific or industrial interpretation. For instance, in time-series change-point detection, if certain specific positions are known to be true changes, it is advised to adopt the member in $\calA$ with which these change-positions can be successfully detected. This strategy allows us to make informative and interpretable decisions, as well as to precisely estimate change-positions.

{ 
We emphasize that OPTICS returns a set of change-point numbers that are no worse than any other number of change-points in the criterion $\calC$ defined in (2.3). If OPTICS were to return an empty set, it would imply that for any given number of change-points, there exists another number of change-points that significantly outperforms the candidate, leading to a contradiction. Therefore, OPTICS always returns at least one change-point number.
}

{
\subsubsection{Multiple-splitting OPTICS for finite-sample stability}

In practice, the resulting confidence set may be sensitive to the particular split used in the OPTICS procedure, especially when the sample size is moderate or the signal is weak. To enhance finite-sample stability, we can apply OPTICS with multiple order-preserving splits and then combine the resulting evidence.

To be specific, let $L\ge 2$ be a fixed integer, and suppose for simplicity that $n=Lm$ for some integer $m$. For each $r=1,\ldots,L$, define the $r$-th order-preserving subsample by
$$
\calZ^{(r)}
=
\{\bz_{r+Lj}: j=0,\ldots,m-1\}.
$$
Under independent observations, each $\calZ^{(r)}$ remains an independent sample and preserves the original temporal order. We further split each subsample $\calZ^{(r)}$ into an ``odd sample'' $\calZ_O^{(r)}$ and an ``even sample'' $\calZ_E^{(r)}$.

For each split $r$, we apply the ordinary OPTICS to $(\calZ_O^{(r)},\calZ_E^{(r)})$. Specifically, for each candidate $K\in\calM$, we use $\calZ_O^{(r)}$ to fit the $K$-change-point model and $\calZ_E^{(r)}$ to evaluate its out-of-sample performance. Let $\bar{\bs}_{K,i}^{(r,O)}$ denote the analogue of $\bar{\bs}_{K,i}^{O}$ constructed from $\calZ_O^{(r)}$, and let $\bs_i^{(r,E)}$ denote the score vector from $\calZ_E^{(r)}$. For $J\in\calM\backslash\{K\}$, define
$$
\hat{\delta}_{K,J}^{(r)}
=
\frac{1}{n_r}
\sum_{i=1}^{n_r}
\left(
\bigl\|\bs_i^{(r,E)}-\bar{\bs}_{K,i}^{(r,O)}\bigr\|_2^2
-
\bigl\|\bs_i^{(r,E)}-\bar{\bs}_{J,i}^{(r,O)}\bigr\|_2^2
\right),
$$
where $n_r:=|\calZ_E^{(r)}|$. Based on $\hat{\delta}_{K,J}^{(r)}$, we compute the split-specific OPTICS test statistic and the corresponding \(p\)-value \(\hat p_K^{(r)}\) exactly as in the original procedure.

To aggregate the evidence across different order-preserving splits, we adopt the Cauchy combination method \citep{liu2020cauchy}. For each $K\in\calM$, define
$$
T_K
=
\sum_{r=1}^{L}\omega_r
\tan\!\left\{\bigl(0.5-\hat p_K^{(r)}\bigr)\pi\right\},
$$
where the weights $\omega_r$ are nonnegative and satisfy $\sum_{r=1}^{L}\omega_r=1$. The corresponding combined $p$-value is
$$
\hat p_K^{\mathrm{MS}}
=
\frac{1}{2}
-
\frac{1}{\pi}\arctan\!\bigl(T_K\bigr).
$$
We then define the multiple-splitting OPTICS (MS-OPTICS) confidence set by
$$
\calA_{\mathrm{MS}}
=
\{K\in\calM:\hat p_K^{\mathrm{MS}}>\alpha\}.
$$

The multiple-splitting version has the same interpretation as the original OPTICS, but is typically less sensitive to the choice of a particular sample split (See Subsection \ref{app:simu:ms-optics} in Supplementary Material for comparison). It therefore provides a simple and practical refinement when greater numerical stability is desired in moderate samples. In our view, the original single-splitting OPTICS remains the default choice because of its simpler presentation and lower computational cost, whereas the multiple-splitting version is best viewed as a practical enhancement.
}

\section{Theory of OPTICS}
\label{sec:3}
In this section, we systematically study the theoretical properties of OPTICS. We first prove that the Gaussian multiplier bootstrap procedure indeed produces valid $p$-values under the desired null space. The second subsection discusses the coverage rate of OPTICS. The theoretical rate of cardinality of OPTICS, which corresponds to the power of test, is provided in the third subsection. 

\subsection{Validity of the bootstrap procedure}
\label{sec3.1}
To study the theoretical guarantee of the Gaussian multiplier bootstrap procedure in our problem setting, we first impose the following  technical assumptions. 
\begin{condition}[Tails and moments] 
	\label{con1}
	For $i=1,\ldots, 2n$, suppose there exist some positive constants $M_1$ and $M_2$, such that 
	\begin{itemize}
		\item [(i)] for any constant vector $\mathbf{b}$, $\| \mathbf{b}^\top \bs_i \|_{\psi_1} \le M_1 \|\mathbf{b} \|_2$;
		\item [(ii)] for any $j = 1,\ldots, d$, we have $\frac{1}{n} \sum_{i=1}^n \E\left[|s_{ij}|^{3} \right] \le M_1$ and $\frac{1}{n} \sum_{i=1}^n \E\left[|s_{ij}|^{4} \right] \le M_1^2$,  where $s_{ij}$ is the $j$th element in $\bs_i$;
		\item[(iii)]$ \lambda_{\min} (\Var[\bs_i])\ge M_2$ and $ \lambda_{\max} (\E[\bs_i \bs_i^\top])\le M_1$, where $\Var(\bs_i)$ is the covariance matrix of $\bs_i$.
	\end{itemize}
\end{condition}

\begin{condition}[Distance between change-positions]
	\label{con2}
	For any two distinct change positions $\tau_k$ and $\tau_{k^\prime}$ in $\calT_K$,  $|\tau_k - \tau_{k^\prime}| {\gtrsim}  n/\log(n)$.
\end{condition}

Condition \ref{con1} provides a sub-exponential tail bound for the data and the moment conditions. Similar conditions can be found in other change-point literature; see \cite{liu2020unified} and \cite{yu2021finite} for instances. Condition \ref{con2} requires a sufficient distance between any two change-points in $\mathcal{T}_K$, since distinguishing two change-points becomes challenging if they are too close. This condition is fairly mild when the candidate model $K \ll n$ \citep{chen2021data}. This condition is also imposed in several prominent works, such as Wild Binary Segmentation (WBS); see Assumption 3.2 in \cite{fryzlewicz2014wild}, and more recently, the Narrowest-Over-Threshold (NOT) method; see Theorem 1 in \cite{baranowski2019narrowest}.

\begin{theorem}\label{thm1}
Suppose Model \eqref{sec1:e1} holds. Let $\mathcal F_O$ denote the $\sigma$-field generated by the observed sample used to construct $\{\bar{\bs}_{K,i}^O:K\in\calM,\ 1\le i\le n\}$. For $J\in\calM\backslash\{K\}$, define $\delta_{K,J}= \frac{1}{n}\sum_{i=1}^n \E\!\left(\xi_{K,J}^{(i)}\mid \mathcal F_O\right)$, and $\sigma^2_{K,J} = \frac{1}{n}\sum_{i=1}^n \E\left((\xi_{K,J}^{(i)})^2\mid \mathcal F_O\right)$. Assume $K_{\max}\asymp \log(n)$ and Conditions \ref{con1} and \ref{con2} hold. In addition, assume there exists a constant $c_0>0$ such that $\min_{J\in\calM\backslash\{K\}} \sigma^2_{K,J}\ge c_0$ with probability tending to one. Then,
\begin{itemize}
    \item[(1)] If $\max_{J \ne K}\delta_{K,J} /\sigma_{K,J}\le x_n (n \log (n)) ^{-1/2}$ for some $x_n = o(1)$, then {for $n \to \infty$},
    $$
    \Pr\left\{H_{0,K} \text{ is not rejected at level } \alpha \right\} \ge 1- \alpha + o(1).
    $$
    
    \item[(2)] If $\alpha \ge n^{-1}$ and $\max_{J \ne K} \delta_{K,J}/ \sigma_{K,J} \ge c n^{-1/2}\log(n)$ for a sufficiently large constant $c>0$, then {for $n \to \infty$},
    $$
    \Pr\left\{H_{0,K} \text{ is not rejected at level } \alpha \right\} = o(1).
    $$
\end{itemize}
\end{theorem}

Theorem \ref{thm1} (1) depicts an ``approximate null" space, i.e, $
\max_{J \ne K}\delta_{K,J} /\sigma_{K,J}\le x_n (n \log (n)) ^{-1/2}$, where we do not reject $H_0$. It implies that the p-value obtained from the bootstrap is preserves the type I error under this approximate null. Meanwhile, (2) claims that OPTICS would successfully rule out those candidate models with inferior predictive power. 

A direct implication of Theorem \ref{thm1} is that when COPSS in \cite{zou2020consistent} is consistent, OPTICS also covers the true number of change-points $K^*$ with confidence level $1-\alpha$ asymptotically. To intuitively see this, consider the special scenario when $K^*$ is indeed the minimizer among all the candidate models in $\calM$. In this case, COPSS is consistent, since it selects the minimizer to be the estimated number of change-points. On the other hand, we have $\max_{K \ne K^*}\delta_{K^*,K} /\sigma_{K^*,J}\le 0$, thus Theorem \ref{thm1} (1) indicates that $\Pr\{K^* \in \calA\} \ge 1 - \alpha+o(1)$. Nevertheless, as demonstrated in the next section, OPTICS preserved the confidence (coverage rate) under conditions that are less stringent than those for COPSS. 

%\begin{corollary}
%	\label{cor1}
%	If the COPSS is consistent, then the OPTICS covers the true number of change-points, i.e., $\Pr\{K^* \in \calA\} \ge 1 - \alpha$ for sufficient large $n$.
%\end{corollary}

\subsection{Coverage rate of OPTICS}
\label{sec3.2}

In this section, we systematically study the coverage rate of OPTICS. 

Let $\underline{\lambda} = \min_{1 \le k \le K^*} (\tau_{k+1}^* - \tau_{k}^*)$, $\bar{\lambda} = \max_{1 \le k \le K^*} (\tau_{k+1}^* - \tau_{k}^*)$ be the minimum and maximum of distances between adjacent true change-positions respectively. For $k=1,\ldots, K^*$, denote $\Delta_k := \|\bmu_{k+1} -\bmu_{k} \|_2$ as the jump size of the $k$th change in model \eqref{model:score}, and $\Delta_{(k)}$ as the corresponding $k$th order statistic. Without loss of generality, assume the true number of change-points $K^*$ belongs to the candidate set $\calM$, i.e., $K_{\max} \ge K^*$. Note that $\calM$ can always be chosen conservatively to include $K^*$. Let $\calM_l =\{K \in \calM: K < K^*\}$ be the lack-of-fit set, and $\calM_o = \{K \in \calM: K > K^*\}$ be the over-fit set. Hence, the candidate set $\calM$ is naturally partitioned into $\calM = \calM_o \cup \calM_l \cup \{K^*\}$. We impose the following conditions before introducing the theoretical results for coverage rate of OPTICS. 

\begin{condition}[Number of change-points]
	\label{con3}
	\item	[(i)] $K^* = o(\underline{\lambda})$ and $K^*(\log(K^*\vee e))^2 = o(\log\log(\bar{\lambda}))$ for Euler's number $e$;
	\item [(ii)] $(K^*_{\max} \log K_{\max}^*)^{1/2} = o(\log\log\bar{\lambda})$.
\end{condition}

\begin{condition}[Accuracy of estimation]
	\label{con4}
	\item[(i)](Over-fit) For any $K\in \calM_o$, denote the corresponding estimated change-position set as $\calT_{oK}=\{\tau_{o1}^K,\ldots, \tau_{oK}^K\}$. There exists a subset $\{\tau_{ok_s}^K, s = 1, \ldots, K^*\} \subset \calT_{oK}$, such that
	$$
	\Pr\left(\forall \tau_k^*\in\calT^*, \ \left|\tau^K_{ ok_s}-\tau_k^*\right| \leq b_{n}\right)\rightarrow 1,
	$$
	where $b_n>0$ satisfies $K^* \log\log(b_n \vee e) = o(\log\log(\bar{\lambda}))$ and $\sum_{k=1}^{K^*} b_n \Delta_k^2 = o(M_1\log\log(\bar{\lambda}))$.
	\item[(ii)] (Lack-of-fit) For any $K \in \calM_l$, denote the corresponding estimated change-position set as $\calT_{lK}=\{\tau_{l1}^K,\ldots, \tau_{lK}^K\}$. There exists a subset of true change-positions $\calI^{*}_{lK}\subset\calT^*$, such that for any $\tau_k^*\in\calI^*_{lK}$, no estimated change-point lies within interval $\{\tau_k^* - \underline{\lambda}/2+1, \ldots, \tau_k^* + \underline{\lambda}/2\}$. Furthermore, denote $\bar{\boldsymbol{\mu}}_{\tau_{lk}^K, \tau_{l(k+1)}^K}$ to be the average of $\{\boldsymbol\mu_i, \ i=\tau_{lk}^K+1,\tau_{lk}^K+2, \ldots, \tau_{l(k+1)}^K\}$. Then for some constant $M_3>0$,
	\begin{equation}
		\label{con:LB2}
		\Pr\left(\forall \tau_k^* \in \mathcal{I}_{lK}^*, \sum_{i=\tau_k^*-\frac{\underline{\lambda}}{4}+1}^{\tau_k^*+\frac{\underline{\lambda}}{4}}\|\boldsymbol{\mu}_i-\bar{\boldsymbol{\mu}}_{K, i}\|_2^2 \geq M_3 \underline{\lambda} \Delta_k^2\right) \rightarrow 1,
	\end{equation}
	with $\bar{\boldsymbol{\mu}}_{K, i}:=\sum_{k=0}^K \mathrm{1}_{\left\{\tau^K_{lk}+1 \leq i \leq \tau_{l(k+1)}^K\right\}} \bar{\boldsymbol{\mu}}_{\tau_{lk}^K, \tau_{l(k+1)}^K}$.
\end{condition}

\begin{condition}[Minimum Signal] 
	\label{con5}
	Assume the minimum jump size $\Delta_{(1)}$ satisfies
	$$
	\frac{\underline{\lambda} \Delta_{(1)}^2}{K^* M_1(\log \bar{\lambda})^2} \rightarrow \infty,
	$$
	where $M_1$ is defined in Condition \ref{con1}.
\end{condition}

Condition \ref{con3} imposes some standard assumptions on the number of change-points. Condition \ref{con4} (i) states that under the over-fitting setting where $K>K^*$, for each true change $\tau_k^*\in\calT^*$, there must exist an estimated change-point lying within the $c_n$-neighborhood of $\tau^*_k$ asymptotically. This estimation accuracy ensures the reliability of the order-preserved cross-validation criterion. Condition \ref{con4} (ii) claims that for an under-fitted model with $K<K^*$, some undetected true change-positions have to be isolated from all the estimated ones by length $\underline{\lambda}/2$, and the estimation errors of mean scores due to the undetected change-positions are not neglected. Condition \ref{con5} bounds below the minimum jump size of mean scores across true adjacent intervals. Conditions \ref{con4} and \ref{con5} are imposed to guarantee the consistent selection of change-points; see the discussion in \cite{pein2021cross}.

\begin{theorem}
	\label{thm2}
	Suppose Conditions \ref{con1}--\ref{con5} hold. Assume that the maximum discrepancy satisfies $\min_{K \ne K^*} \max_{i=1,\ldots,n}\|\bar{\bs}^O_{K^*,i}- \bar{\bs}_{K,i}^O \|^2_2 \gtrsim \Delta_{(K^*)}^2$. 
	If the maximum jump size $\Delta_{(K^*)}$ satisfies
	\begin{equation}\label{true.signal}
	    	\Delta_{(K^*)}^2 \gtrsim x_n^{-2} (\log\log\bar{\lambda})^2 \frac{\log^2(n)}{n},
	\end{equation}
	then $\Pr\left\{K^* \in \calA \right\} \ge 1- \alpha + o(1)$.
\end{theorem}

% \begin{theorem}
% 	\label{thm2}
% 	{ Suppose Model \eqref{sec1:e1} holds}. Under Conditions \ref{con1}- \ref{con5}, $\min_{K \ne K^*} \max_{i=1,\ldots,n}\|\bar{\bs}^O_{K^*,i}- \bar{\bs}_{K,i}^O \|^2_2 \gtrsim \Delta_{(K^*)}^2$ with { asymptotic probability $1$}, and $\Delta^2_{(K^*)} \gtrsim x_n^{-1} \log\log(\overline{\lambda}) \sqrt{\log^{3}(n)/n}$
% 	% together with
% 	for $x_n = o(1)$, we have { for $n \to \infty$}
% 	\begin{eqnarray}\label{true.signal}
% 		\max_{K \ne K^*}\frac{\delta_{K^*,K} }{\sigma_{K^*,K}}\le x_n \frac{1}{\sqrt{n \log  n}},
% 	\end{eqnarray}
% 	and hence $ \Pr\left\{K^* \in \calA \right\} \ge 1- \alpha + o(1)$.
% \end{theorem}

The inequality \eqref{true.signal} in Theorem \ref{thm2} states that the true number of change-points $K^*$ indeed lies within the ``approximate null" space. Therefore, incorporating Theorem \ref{thm1}, the coverage rate of $\calA$ can be guaranteed. The conditions for Theorem \ref{thm2} relax those in \cite{zou2020consistent} and \cite{pein2021cross}. The latter two papers both additionally impose a divergent lower bound for the ``over-fit" effect; to be specific, they require for some $c_n\to\infty$, 
\begin{eqnarray}\label{overfit.effect}
	S_{\boldsymbol\epsilon^O}(\calT^*) - \max_{K \in \calM_o}S_{\boldsymbol\epsilon^O}(\calT^* \cup \calT_K)\ \ge c_n,
\end{eqnarray}
where $\boldsymbol\epsilon^O=\{\boldsymbol\epsilon_{2i-1}, i=1,\ldots, n\}$ represents the collection of individual errors from the given odd sample, with $\boldsymbol\epsilon_i$ defined in model \eqref{model:score}. 
{We will explore the explicit rate of $c_n$ in Section \ref{sec:cardinality} to enhance the power of the method. As shown in Section \ref{sec:simulation}, the divergence of $c_n$ in (\ref{overfit.effect}) can be fairly stringent in many practical scenarios, potentially leading to the failure of consistency for the respective estimations.}
%\begin{remark}
%Technically, our statement is different with \cite{zou2020consistent} and \cite{pein2021cross} as we give a uniform bound for expectation while they give a probability bound for happenings of events. It is well known that convergence in probability does not imply convergence in moment, hence, we need to take more efforts to prove our results.
%\end{remark}
\subsection{Cardinality of OPTICS}\label{sec:cardinality}

Theorem \ref{thm2} guarantees that OPTICS $\calA$ covers the true number of change-points $K^*$ at the nominal confidence level asymptotically. However, it is always possible to select a sufficiently conservative $\calA$ that encompasses $K^*$, such as the trivial set $\calM$. In this instance, OPTICS is non-informative and has no power. Therefore, in this subsection, we delve into examining the power of OPTICS by analyzing its cardinality. We first define the following two sets to depict the under-fit and over-fit signal-to-noise ratios, respectively. 

\begin{definition}[Signal-to-noise ratio]
	\label{con6} 
	\item[(i)] In a lack-of-fit model $K\in\calM_l$, consider the undetected change-position set $\calI_{lK}^*$ in Condition \ref{con4}. For a sufficiently large $n$, define 
	$$\calB_{1n} := \left\{ K \in \calM_{l}: \sum_{\tau_k^* \in \calI_{lK}^*} \Delta^2_k \gtrsim M_1 \Delta_{(K^*)}^2\sqrt{n\log(n)}/\underline{\lambda}\right\},$$
	where $\Delta_{(K^*)}$ is the maximum jump size between adjacent changes, as defined in Section 3.2.
	\item[(ii)] In an over-fit model $K\in\calM_o$, for a sufficiently large $n$, define 
	{
    $$
    \begin{aligned}
        \calB_{2n} :=& \Big\{K \in \calM_{o}: \E\left\{\calS_{\bs^E}\left(\mathcal{T}_{K}\right)-\calS_{\bs^E}\left(\mathcal{T}_{K} \cup \mathcal{T}^*\right)\right\}  + \left\{\calS_{\boldsymbol\epsilon^O}\left(\mathcal{T}^*\right)-\calS_{\boldsymbol\epsilon^O}\left(\mathcal{T}_{K} \cup \mathcal{T}^*\right)\right\} \\
    &\quad\quad\quad\quad \gtrsim M_1\Delta_{(K^*)}^2 \sqrt{n \log(n)}\Big\}.
    \end{aligned}
    $$}
\end{definition}

Note that set $\calB_{1n}$ in Definition \ref{con6} (i) collects a subset of lack-of-fit models, where the jump sizes of undetected changes are bounded below. Set $\calB_{2n}$ consists of some over-fitted models whose in-sample over-fitting effects are sufficiently large. Based on $\calB_{1n}$ and $\calB_{2n}$, the following theorem provides the asymptotic upper bounds for the cardinality of OPTICS.

\begin{theorem}
	\label{thm3}
	Suppose Model \eqref{sec1:e1} holds. Under Conditions \ref{con1}--\ref{con5} and assuming the maximum discrepancy satisfies $\max_{i=1,\ldots, n}\| \bar{\bs}^O_{K,i}- \bar{\bs}^O_{J,i}\|_2^2 \lesssim \Delta_{(K^*)}^2$ for all pairs $K \ne J$ with asymptotic probability $1$, we have as $n \to \infty$:
	\begin{equation}
		\label{cardinality}
		\Pr\left\{|\calA| \le K_{\max} - |\calB_{1n}|   -|\calB_{2n}|  \right\} \ge 1 - o(1),
	\end{equation}
	and
	\begin{equation}
		\label{power}
		\Pr\left\{ K^* \in \calA, \ |\calA| \le K_{\max} - |\calB_{1n}|  -|\calB_{2n}|  \right\} \ge 1- \alpha+ o(1).
	\end{equation}
\end{theorem}
Theorem \ref{thm3} confines the cardinality of confidence set $\calA$. If for sufficiently large $n$, the minimum jump size on its own satisfies $\Delta_{(1)}^2\geq M_1\Delta_{(K^*)}^2\sqrt{n\log(n)}/\underline{\lambda}$, then all lack-of-fit models belong to $\calB_1$, hence $|\calB_1| = |\calM_{l}|$. Based on \eqref{cardinality} in Theorem \ref{thm3}, the confidence set $\calA$ excludes all lack-of-fit models where $K \in \calM_{l}$ with overwhelming probability. If further all over-fit models belong to $\calB_2$, the cardinality of OPTICS will be 1, and $\Pr\{\calA = K^*\} \to 1$. On the other hand, if unfortunately both under-fitting and over-fitting signals are not sufficiently strong, such that there are no elements in either $\mathcal{B}_1$ or $\mathcal{B}_2$, then \eqref{power} in Theorem \ref{thm3} degenerates to the coverage rate result in Theorem \ref{thm2}.

{
\section{Two extensions of OPTICS}
In this section, we extend OPTICS to the heave-tailed data and m-dependent data.

\subsection{Extension of OPTICS to heavy-tailed data}

We explore the potential extension of the OPTICS method to heavy-tailed data. In this scenario, the tail behavior of the scores $\bs_i$ defined in \eqref{model:score} may exhibit heavy-tailed characteristics, causing the Sub-Gaussian condition in Condition \ref{con1} to no longer hold. This issue could result in the failure of the OPTICS method.

A possible solution is to generalize the criterion in \eqref{CK} from the $\ell_2$ loss to a robust loss function, such as the Huber loss \citep{huber1992robust}, defined as
$$
\ell_{\kappa}(u)=\begin{cases}
			u^2/2, & \text{ if } |u| \le \kappa\\
            \kappa |u| - \kappa^2/2, & \text{ if } |u| > \kappa,
		 \end{cases}
$$
where $\kappa > 0$ is a tuning parameter that balances bias and robustness. Then, we define a robust version of \eqref{CK} as 
\begin{eqnarray}\label{CK_huber}\mathcal{C}_{\kappa}\left(\mathcal{T}_K ; \mathcal{Z}_E\right):=\frac{1}{n}\sum_{k=0}^K \sum_{i=\tau_k^K+1}^{\tau_{k+1}^K}\ell_{\kappa} \left(\bs_i^E - \bar{\bs}_{\tau_k^K,\tau_{k+1}^K}^O\right),
\end{eqnarray}
where $\ell_{\kappa}(\bx) = \sum_{i=1}^p \ell_{\kappa}(x_i)$ for $\bx \in \R^p$. Alternatively, the $\ell_1$ loss function can also be used. Intuitively, the difference $\bs_i^E - \bar{\bs}_{\tau_k^K,\tau{k+1}^K}^O \approx \bs_i^E - \E\left(\bs_i^E\right) = \beps_i^E$. When $\beps_i^E$ exhibits heavy-tailed behavior, the $\ell_2$ loss function becomes unsuitable, as it can be distorted by extreme values. In contrast, the Huber loss is much more robust, as it transitions from the $\ell_2$ loss to the $\ell_1$ loss when $\beps_i^E$ deviates significantly from zero. Then, we can define a robust version of \eqref{Xi} as
\begin{eqnarray}\label{Xi_robust}
	\hat{\delta}_{K, J,\kappa}  =  \frac{1}{n}\sum_{i =1}^{n} \left( \ell_{\kappa}\left(\bs_i^E  -\bar{\bs}_{K,i}^O\right) -  \ell_{\kappa}\left( \bs_i^E  - \bar{\bs}_{J,i}^O\right)\right) : = \frac{1}{n}\sum_{i =1}^{n} \xi_{K,J,\kappa}^{(i)}.
\end{eqnarray}
The corresponding test statistic is
\begin{eqnarray}\label{TK_robust}T_{K,\kappa} = \max_{J \ne K}  \frac{\sqrt{n}\hat{\delta}_{K,J,\kappa}}{\hat{\sigma}_{K,J,\kappa}}=\max_{K\ne J} \frac{1}{\sqrt{n}} \sum_{i=1}^{n} \frac{\xi_{K,J,\kappa}^{(i)}}{\hat{\sigma}_{K,J,\kappa}},
\end{eqnarray}
where $\hat{\sigma}_{K,J,\kappa}^2 = n^{-1}\sum_{i=1}^n (\xi_{K,J,\kappa}^{(i)})^2$ is the estimated second moment. For this robust version of statistic, we can also use the Gaussian multiplier bootstrap statistic similar as \eqref{bootT}, which is
\begin{eqnarray}\label{bootT_robust}
	{ T^{\sharp}_{K,\kappa, b}} = \max_{K\ne J} \frac{1}{\sqrt{n}} \sum_{i=1}^{n} \frac{\xi_{K,J,\kappa}^{(i)}}{\hat{\sigma}_{K,J,\kappa}} \zeta_i.
\end{eqnarray}
Then the $p$-value is naturally set to be $\hat{p}_{K,\kappa} = B^{-1} \sum_{b=1}^B \I({T^{\sharp}_{K,\kappa,b}} > T_{K,\kappa})$. We refer to \cite{chen2020robust} for adaptively selecting the hyperparameter $\kappa$. We name this extension as Huber-OPTICS (H-OPTICS). Subsection \ref{app:simu:heavy} of the Supplementary Material demonstrates the superior performance of our proposed method through simulation studies

\subsection{Extension of OPTICS to $m$-dependent data}
We extend the OPTICS method to handle dependent data, specifically focusing on $m$-dependent data, which is commonly used in econometrics applications \citep{moon2013tests, kokoszka2018dynamic}. We restrict our discussion in multiple mean change-points detection, where
$$
\bz_i = \bbeta_k^* + \beps_i, \quad \tau_{k-1}^*<i \leq \tau_{k}^*, k=1, \ldots, K^*+1; \  i=1, \ldots, (m+1)n,
$$
where $\beps_i \sim (\mathbf{0}, \bSigma)$ for covariance matrix $\bSigma \in \R^{d \times d}$. This model is a special case of \eqref{sec1:e1}.
In our work, we define a stochastic process $\calZ = \{\bz_1, \ldots, \bz_{(m+1)n}\}$ is called $m$-dependent for $m \ge 0$, if
$$
\{\bz_1,\ldots,\bz_{i-1}, \bz_i\} \text{ and } \{\bz_{i+m+1}, \bz_{i+m+2},\ldots, \bz_{(m+1)n}\}
$$
are independent. Here, we assume sample size is $(m+1)n$ to simplify the presentation. When $m = 0$, $m$-dependence reduces to standard independence.

{
For $m$-dependent data, we apply the order-preserving $(m+1)$-splitting
$$
\calZ^{(r)}
=
\{\bz_{r+(m+1)i}: i=0,\ldots,n-1\},\quad 
r=1,\ldots,m+1.
$$
Then each subsequence $\calZ^{(r)}$ consists of independent observations. We may therefore apply the multiple-splitting OPTICS procedure as in Section~\ref{sec2:3}, yielding a split-specific $p$-value for each $K\in\calM$. These $p$-values are then combined to obtain a single combined $p$-value. We refer to this extension as $m$-dependent OPTICS (m-OPTICS). {Subsections~\ref{app:simu:dependent} and \ref{app:simu:vary_dependent} of the Supplementary Material provide thorough simulations to investigate the effects of $m$-dependence. We find that m-OPTICS is robust to various $m$-dependent structures and performs especially well when $m$ is moderate. However, we should emphasize that $m$ cannot be excessively large. A very high $m$ leads to a loss of selection power, as each subsample will contain too few observations to yield a reliable change-point estimation.}
}
}

\section{Simulation studies}
\label{sec:simulation}
In this section, we conduct simulation studies under various change-point model settings to assess the empirical performance of OPTICS. Two base change-point detection algorithms are used in the training stage for constructing OPTICS - the Binary Segmentation (BS, \cite{fryzlewicz2014wild}) and the Segmentation Neighborhood (SN, \cite{auger1989algorithms}). 

For comparison, we consider the consistent estimation method COPSS \citep{zou2020consistent}, the FDR-control method FDRseg \citep{li2016fdr} and the FWER-control method SMUCE \citep{frick2014multiscale}. We also adopt BS and SN as base algorithms for COPSS. Since no state-to-art confidence set construction methods for the number of change-points are available in literature, we artificially create the quasi-confidence set based on each point estimate  from COPSS, FDRseg or SMUCE, respectively, to match the cardinality of OPITCS. This also naturally enhances the coverage rate of these existing methods. Specifically, given the point estimate $\hK$, the associated quasi-confidence set is defined as $\calA^{q} : = \{\hK-q, \hK-q+1, \ldots, \hK, \ldots, \hK+q\}$ with cardinality $2q+1$. We remark that these sets only facilitate a fair numerical comparison with comparable cardinality with OPTICS; they lack theoretical justification of the coverage rate. 

Throughout this section, the candidate set is taken as $\calM = \{1, 2, \ldots, \log (n)\}$ following \cite{zou2020consistent}, and the significance level $\alpha = 0.1$. We conduct $100$ simulation runs for each setting, and report the coverage rate $\sum_{i=1}^{100} \I(K^* \in \calA_i)/100$ and the average cardinality of confidence set $\sum_{i=1}^{100}|\calA_i|/100$, where $\calA_i$ denotes the set obtained in the $i$th simulation run.

\subsection{Multiple mean-change model}
\label{subsec:4.1}
Consider the multiple mean-change model
$$
\by_i = \bmu^*_k + \beps_i, \ \tau_{k-1}^*  < i  \le \tau^*_{k}, \ k = 1, \ldots, K^{*}+1, \ i=1,\ldots,2n,
$$ 
where $\tau_k^*, \ k = 1, \ldots, K^*$ are the true mean change-points, $\bmu^*_{k}$ is the $d$-dimensional mean vector for subject $i$ when $\tau_{k-1}^*  < i  \le \tau^*_{k}$; $\beps_i$ is the independently and identically distributed random error. 

In this example, both the univariate mean with $d=1$ and the multivariate mean vector with $d=5$ are studied. However, as FDRseg and SMUCE only work for univariate mean change-point detection, we only compare OPTICS with COPSS when $d=5$. The sample size is taken to be $2n = 1000$, and the set of true change-points $\calT^* = \{\tau_k^* = 200k, k =1, \ldots, 4\}$, hence $K^*=4$. The $k$th mean vector $\bmu^*_{k} = (-1)^{k-1} A\mathbf{1}_d, \ k = 1, \ldots, 4$, where $A$ is a {scalar}, representing the amplitude, and $\mathbf{1}_d$ is the $d$-dimensional vector with all elements 1's. The amplitude $A$ varies among $\{0.50, 0.625, 0.75, 0.875, 1\}$. Two error distributions are under consideration - the standard normal distribution $N(0,1)$ and a t-distribution $t(10)$. As we observe similar phenomenon for these two types of error distributions, we only exhibit that for the more challenging $t(10)$ case. Refer to Section \ref{sup:simu} of Supplementary Material for the simulation results under $N(0,1)$. 

Table \ref{tab:2} and \ref{tab:4} report the coverage rates when $d=1$ and $d=5$, respectively, under t-distributed errors. The quantities in parentheses are average cardinalities of estimated sets, from which we observe the average cardinality of OPTICS is around 3. Therefore, to match this cardinality, we take $q=1$ for the quasi-confidence sets for COPSS, FDRseg and SMUCE, i.e., $\calA^{1} = \{\hK-1, \hK, \hK+1\}$ with cardinality $3$, and denote the generated sets by $\calA^1_{COPSS}$,  $\calA^1_{FDRseg}$ and $\calA^1_{SMUCE}$, respectively. We also report the coverage rate solely by the point estimates in the last four rows.  

From both tables, the coverage rate of OPTICS gradually meets the nominal level with decreasing average cardinality as the amplitude $A$ increases, especially with the SN detection algorithm. This empirically illustrates the validity of Theorem \ref{thm2} and \ref{thm3}. The BS algorithm is less favorable due to its nature as an approximation algorithm rather than an exact one; thus, the accuracy of detected change positions falls short of meeting the requirements in Condition \ref{con4}. Nevertheless, OPTICS with both SN and BS are superior to the quasi-confidence sets generated from COPSS, FDRreg and SMUCE in terms of coverage probabilities. Recall that these quasi-confidence sets do not possess any theoretical guarantee. On the other hand, as depicted in the last four rows of Table \ref{tab:2}, the point estimates exhibit low probabilities of correctly encompassing the true number of change-points. 

Furthermore, OPTICS demonstrates its superiority compared to SMUCE and FDRseg especially under non-Gaussian errors. This advantage stems from SMUCE and FDRseg utilizing the supremum of Brownian motion to establish cutoffs, which become misaligned under non-Gaussian assumptions. Notably, SMUCE lacks power when dealing with small amplitudes, elucidating the stringency of using the FWER-type criterion. Additionally, it is worth noting that OPTICS serves as a unified method applicable to the multivariate change-point models, while both SMUCE and FDRseg are limited to scenarios where $d = 1$.

\begin{table}[h]
	\caption{The coverage rates in the mean-change model: $d=1$; t-distributed error.\label{tab:2}}
	\centering
	\begin{tabular}{cccccc}
		\hline
		Amplitude  $A$          & 0.50       & 0.625      & 0.75       & 0.875      & 1.00          \\ \hline
		OPTICS(BS)  & 0.72(4.02) & 0.82(4.11) & 0.89(3.48) & 0.86(2.86) & 0.80(2.75) \\
		OPTICS(SN)  & 0.89(4.40) & 0.92(4.22) & 0.97(3.88) & 0.98(2.82) & 0.99(2.80) \\
		$\calA^1_{COPSS}$(BS)& 0.46(3.00)       & 0.58(3.00)       & 0.69(3.00)       & 0.80(3.00)       & 0.88(3.00)       \\
		$\calA^1_{COPSS}$(SN)& 0.38(3.00)      & 0.45(3.00)       & 0.80(3.00)       & 0.88(3.00)       & 0.97(3.00)      \\
		$\calA^1_{FDRseg}$ & 0.69(3.00)      & 0.64(3.00)      & 0.59(3.00)       & 0.58(3.00)      & 0.61(3.00)      \\
		$\calA^1_{SMUCE}$ & 0.28(3.00)      & 0.62(3.00)      & 0.94(3.00)       & 0.97(3.00)      & 0.96(3.00)      \\
		COPSS(BS) & 0.19       & 0.25       & 0.46       & 0.47       & 0.48       \\
		COPSS(SN) & 0.18       & 0.33       & 0.57       & 0.77       & 0.85       \\
		FDRseg    & 0.39       & 0.45       & 0.44       & 0.43       & 0.38       \\	
		SMUCE     & 0.07       & 0.30       & 0.58       & 0.81       & 0.89       \\
		\hline
	\end{tabular}
\end{table}

\begin{table}[h]
	\caption{The coverage rates in the mean-change model: $d=5$; t-distributed error.\label{tab:4}}
	\centering
	\begin{tabular}{cccccc}
		\hline
		Amplitude  $A$        & 0.50       & 0.625      & 0.75       & 0.875      & 1.00          \\ \hline
		OPTICS(BS)  & 0.60(3.37) & 0.54(2.84) & 0.61(2.17) & 0.58(2.29) & 0.70(2.06) \\
		OPTICS(SN)  & 0.73(4.09) & 0.82(3.40) & 0.78(2.33) & 0.91(2.29) & 0.97(2.37) \\
		$\calA^1_{COPSS}$(BS) & 0.51(3.00)       & 0.51(3.00)       & 0.65(3.00)       & 0.77(3.00)       & 0.84(3.00)       \\
		$\calA^1_{COPSS}$(SN) & 0.58(3.00)      & 0.66(3.00)       & 0.85(3.00)       & 0.92(3.00)       & 0.93(3.00)      \\
		COPSS(BS) & 0.20       & 0.23       & 0.38       & 0.40       & 0.54       \\
		COPSS(SN) & 0.24       & 0.39       & 0.68       & 0.84       & 0.89     \\    
		\hline
	\end{tabular}
\end{table}

\subsection{Linear regression model with coefficient structural-breaks}
\label{subsec:4.2}
The change-point detection problem can be naturally extended to the coefficient structural-change detection in linear regression models. In this subsection, we consider the following linear regression model with $K*$ potential coefficient structural-breaks:   
$$
y_i = \bx_i^\top \bbeta^*_{k} +\epsilon_i, \ \tau^*_{k-1} < i \le \tau_k^*, \ k = 1, \ldots, K^*+1, \ i = 1,\ldots, 2n,
$$
where $\bbeta^*_{k} = (-1)^{k-1}A\mathbf{1}_d, \ k = 1, \ldots, K^*+1$. We generate the covariate $\bx_i$ from the multivariate normal distribution $N(\mathbf{0}, \bI_d)$, with $d=5$, and the error term $\epsilon_i$ from $N(0,1)$ and $t(10)$, respectively. The sample size $n = 1000$, the true number of change-points $K^*=4$, and the true change-point set $\calT^* = \{200k, k = 1, \ldots, 4\}$. 
% \begin{table}[h]
% 	\caption{The coverage rates in linear model with coefficient structural-breaks; $t(10)$ errors.\label{tab:7}}
% 	\centering
% 	\begin{tabular}{cccccc}
% 		\hline
% 		Amplitude $A$         & 0.10       & 0.125      & 0.15       & 0.175      & 0.20        \\ \hline
% 		OPTICS(BS)  & 0.72(4.22) & 0.75(4.01) & 0.74(4.02) & 0.69(4.02) & 0.81(4.34) \\
% 		OPTICS(SN)  & 0.97(4.22) & 0.98(4.05) & 0.94(4.28) & 0.92(4.51) & 0.91(4.49) \\
% 		$\calA^1_{COPSS}$(BS) & 0.25(3.00)       & 0.22(3.00)       & 0.31(3.00)       & 0.33(3.00)       & 0.38(3.00)       \\
% 		$\calA^1_{COPSS}$(SN) & 0.32(3.00)      & 0.42(3.00)       & 0.49(3.00)       & 0.53(3.00)       & 0.65(3.00)      \\
% 		COPSS(BS) & 0.11       & 0.08       & 0.08       & 0.08       & 0.11       \\
% 		COPSS(SN) & 0.00       & 0.00       & 0.01       & 0.01       & 0.01       \\
% 		\hline
% 	\end{tabular}
% \end{table}

\begin{table}[h]
	\caption{The coverage rates in linear model with coefficient structural-breaks; $t(10)$ errors.\label{tab:7}}
	\centering
	\begin{tabular}{cccccc}
		\hline
		Amplitude $A$         & 0.10       & 0.125      & 0.15       & 0.175      & 0.20       \\ \hline
		OPTICS(BS)  & 0.76(3.94) & 0.89(3.67) & 0.79(3.26) & 0.94(2.80) & 0.85(2.39) \\
		OPTICS(SN)  & 0.85(3.22) & 0.95(2.65) & 0.94(2.19) & 0.99(1.95) & 0.98(1.62) \\
		$\calA^1_{COPSS}$(BS) & 0.44(3.00) & 0.46(3.00) & 0.52(3.00) & 0.52(3.00) & 0.53(3.00) \\
		$\calA^1_{COPSS}$(SN) & 0.67(3.00) & 0.72(3.00) & 0.78(3.00) & 0.75(3.00) & 0.78(3.00) \\
		COPSS(BS) & 0.14       & 0.24       & 0.21       & 0.19       & 0.28       \\
		COPSS(SN) & 0.46       & 0.46       & 0.58       & 0.59       & 0.66       \\
		\hline
	\end{tabular}
\end{table}

The coverage rates under the $t(10)$ distribution are presented in Table \ref{tab:7}. For the coverage rates under $N(0,1)$, please refer to Table \ref{tab:6} in Section \ref{sup:simu} of Supplementary Material. Since FDRseg and SMUCE are inapplicable for structural-break detection, our comparison focuses solely on OPTICS against COPSS and its related quasi-confidence sets. In Table \ref{tab:7}, a notably superior performance of OPTICS is observed in terms of coverage probability, particularly when SN is used as the base algorithm. Importantly, COPSS fails to achieve consistency under these weak signal settings, and the associated quasi-confidence intervals that match the cardinalities of OPTICS are not adequately expanded to ensure the coverage rate. In contrast, OPTICS adaptively adjusts its cardinality to achieve the desired confidence level, due to the nature of its originated testing framework. 

{
We also include expanded simulation studies covering variance, network, and covariance change-points, and multiple mean changes with heavy-tailed as well as $m$-dependent distributions in Section \ref{sup:simu} of the Supplementary Material
}

{
\section{Real Data Analysis}
\label{sec:5}

%\subsection{Bladder tumor microarray dataset}
In this section, we apply OPTICS to analyze the bladder tumor microarray dataset sourced from the \textit{ecp} R package \citep{james2013ecp}. We extract data from 10 individuals diagnosed with bladder tumors. For each individual, log-intensity-ratio measurements for 2215 distinct genetic loci were collected. Our primary objective is to identify the number of change-points within the genetic loci, allowing us to pinpoint potential influential genes associated with bladder tumors. 

We first demonstrate that the dataset may exhibit sub-Gaussian behavior, thereby satisfying the main assumption of OPTICS. An equivalent condition for sub-Gaussianity is that the $k$th moment is bounded by $c k^{k/2}$ for some universal constant $c > 0$. The left panel in Figure~\ref{newfig:01} presents the values of up to the $100$th moments for 10 individuals, which are bounded by $k^{k/2}$ (black curve). This confirms that, at least for each individual, the data exhibit sub-Gaussian properties.

\begin{figure}[h]
	\centering
	\includegraphics[width=6cm, height=4cm]{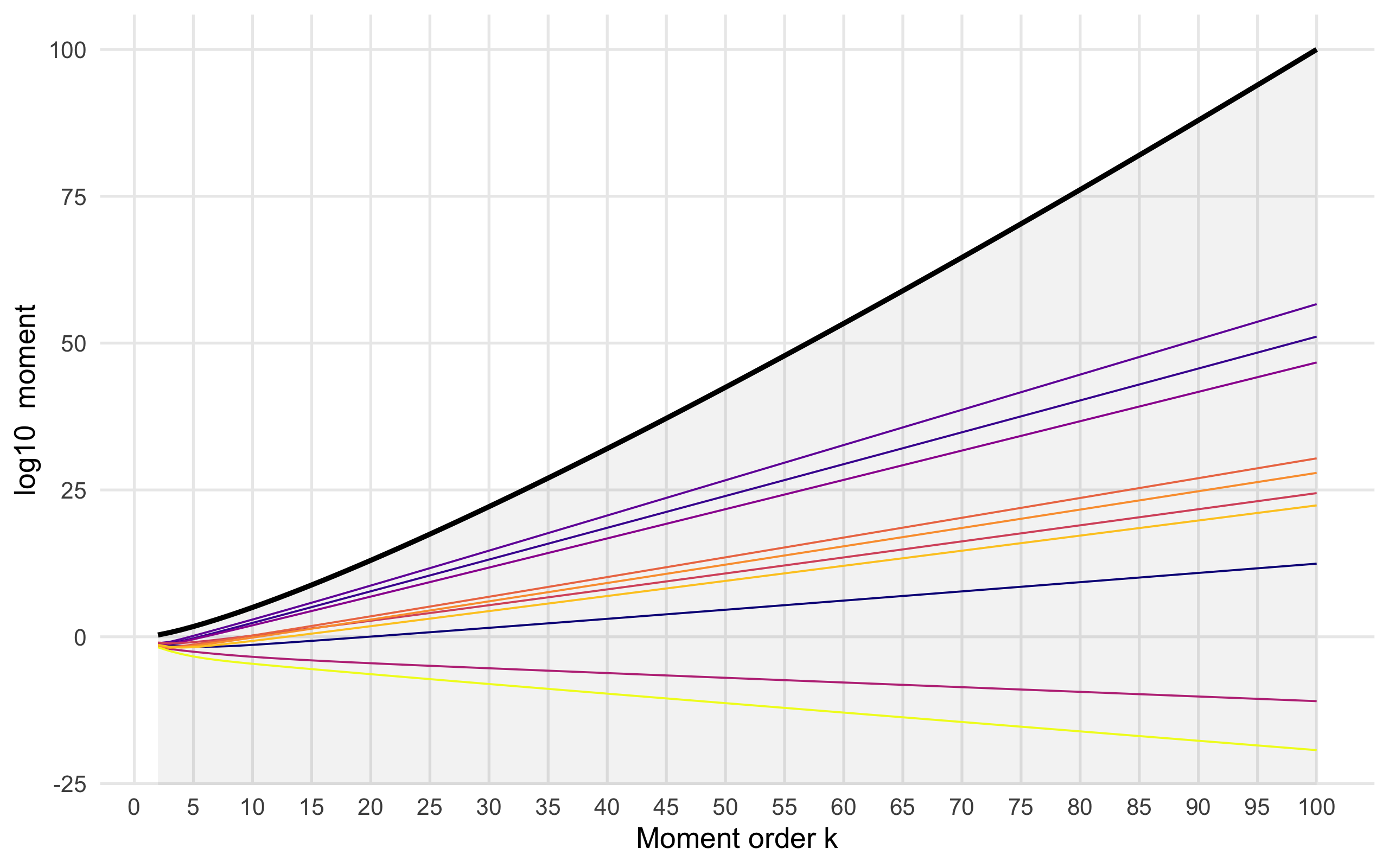}
      \includegraphics[width=5cm, height=3cm]{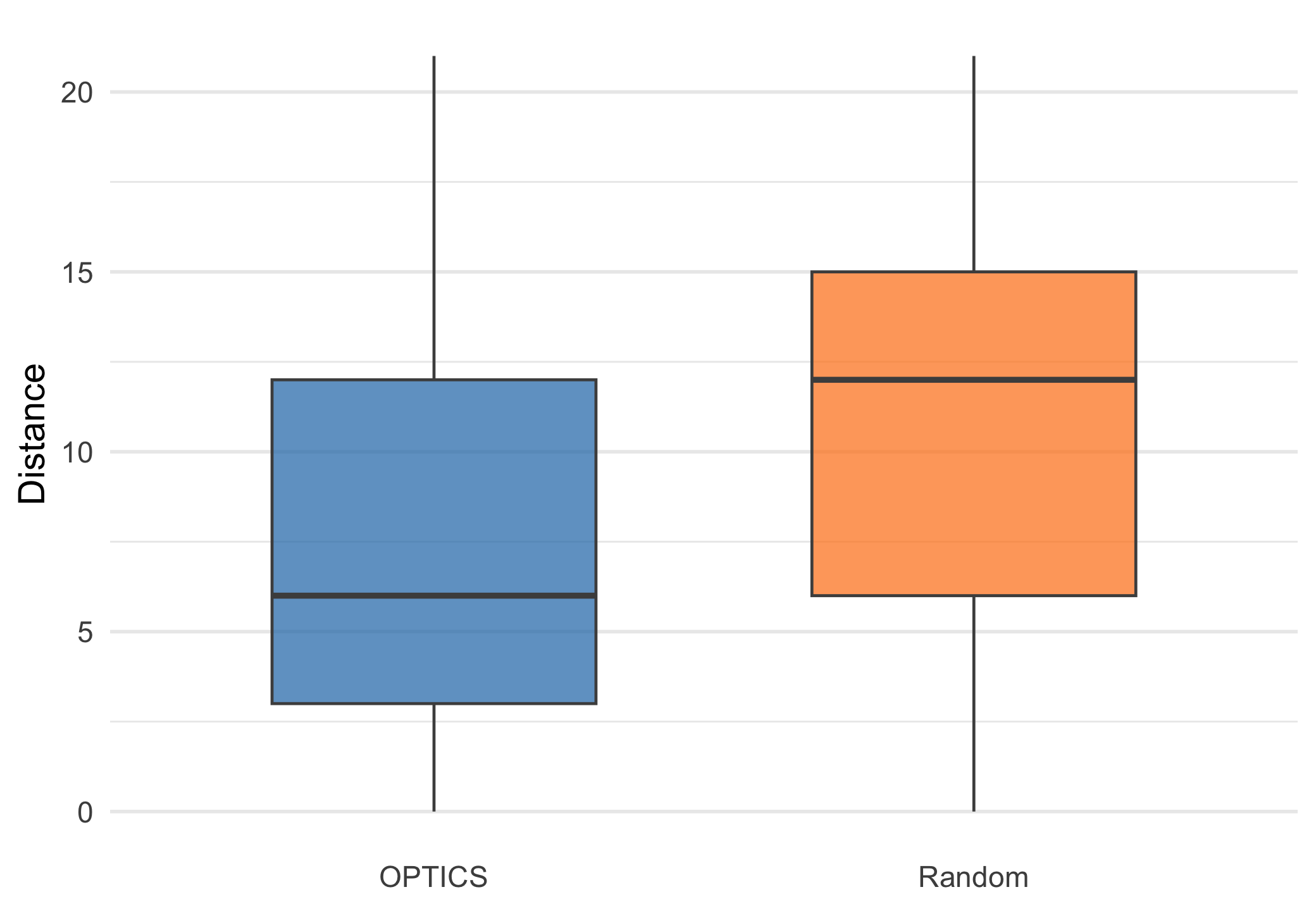}
	\caption{Left Panel: The $k$th moment plots of $10$ individuals, where the black curve represents $y = k^{k/2}$. The y-axis is rescaled using $\log_{10}$.  
Right Panel: Boxplots of Hausdorff distances between the confidence sets obtained by OPTICS and those generated randomly. {\label{newfig:01}}}
\end{figure}

We apply our OPTICS procedure for detecting multiple mean changes to the bladder tumor microarray dataset. This dataset is analyzed from two perspectives. First, we apply OPTICS to each individual and construct the corresponding confidence sets. Since all ten individuals share the same disease, their change-point patterns may exhibit similar structures. We set $K_{\max} = 3\log(10)$ and apply OPTICS to obtain $\calA_i$ for $i \in [10]$, where $\calA_i$ represents the confidence set for the number of change-points for individual $i$. The concrete values within each set can be found in Section~\ref{supp:real} in the Supplementary Material. To verify whether $\calA_i$ for $i \in [10]$ share similar patterns, we compute the pairwise Hausdorff distances between these sets and compare them to randomly generated sets. Specifically, we construct $\calB_i$ for $i \in [10]$, where $\calB_i$ is sampled from $\{1, \ldots, K_{\max}\}$ without replacement and has the same cardinality as $\calA_i$. The right panel of Figure~\ref{newfig:01} presents the boxplots of the distributions of Hausdorff distances for the OPTICS confidence sets and the randomly generated sets. It is evident that the OPTICS confidence sets have smaller Hausdorff distances, indicating greater consistency in detecting shared patterns across individuals. This shows that the OPTICS has power in covering the underlying true number of change-points.

We further detect the change-points using the Binary Segmentation algorithm. To control the family-wise error rate (FWER) (see the discussion in Section \ref{sec2:3}), we select the smallest value in each $\calA_i$ as the optimal number of change-points. The detected change-points, based on the corresponding number of change-points, are presented in Figure~\ref{fig:sep}. From the figure, it is evident that the number of change-points chosen by OPTICS provides a good fit for each individual.
\begin{figure}[H]
	\centering
	\includegraphics[width=16cm, height=15cm]{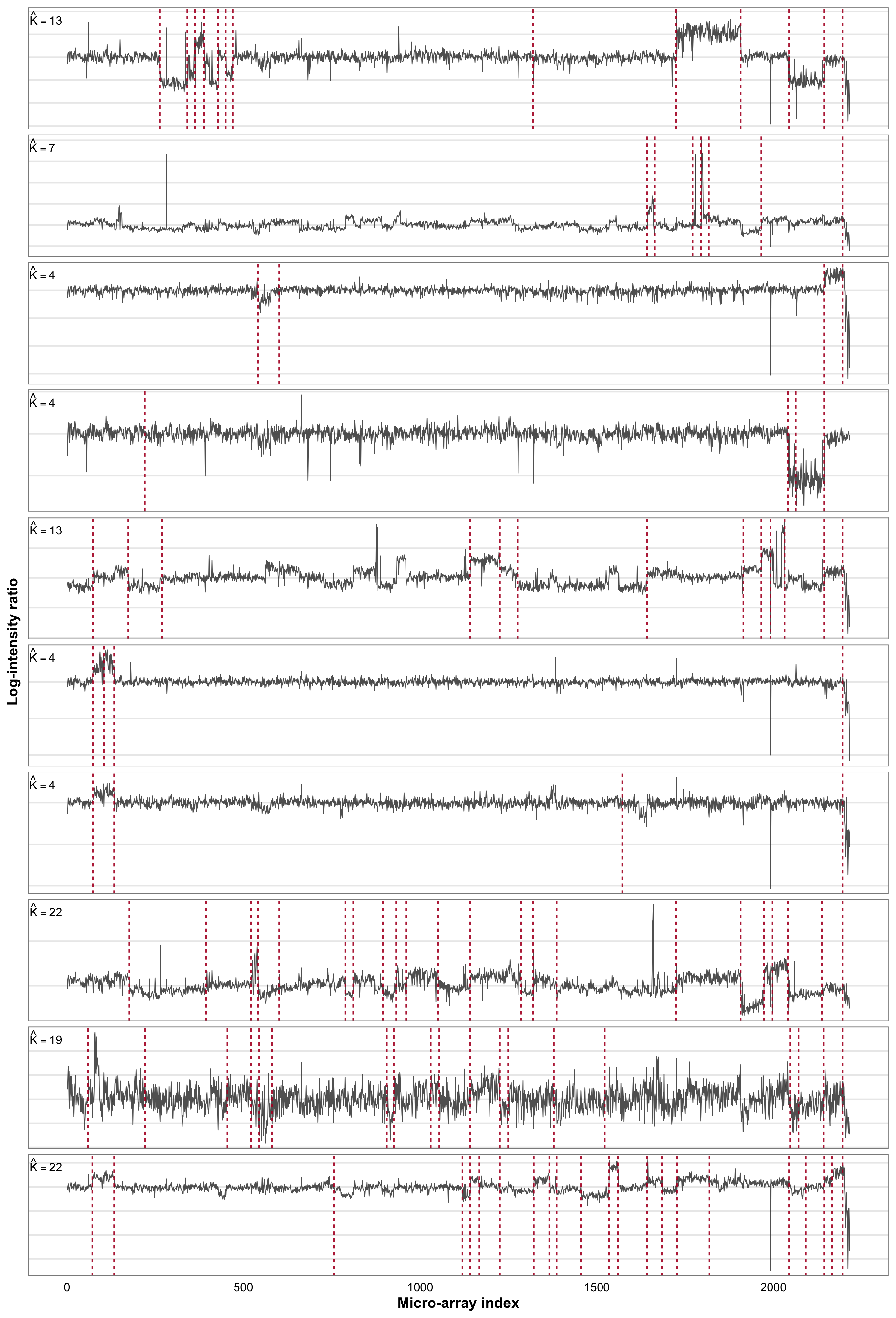}
	\caption{Detected change-points for each individual separately.{\label{fig:sep}}}
\end{figure}

We now perform a joint analysis of the 10 individuals. Using OPTICS for the multi-dimensional mean change-points model, we construct the confidence set for the number of change-points, $\calA_{\textit{joint}} = \{10,13,16, 19, 22\}$. To control the FWER, we select the minimum value within $\calA_{\textit{joint}}$. The change-points are then detected using the Wild Binary Segmentation (WBS) algorithm, which is $\mathcal{S}_{\text{joint}} = \{154, 358, 1140, 1268, 1534, 1724, 1906, 1966, 2052, 2142\}$, as shown in Figure~\ref{fig:joint}. The plot illustrates that each dark red bold vertical line, identified via the OPTICS method, corresponds to a location where a change occurs in at least one coordinate. Conversely, the light red vertical lines, representing the default WBS detections not selected by OPTICS, contain numerous false positives. This confirms that OPTICS-based selection effectively controls the false discovery rate while preserving detection accuracy.

\begin{figure}[H]
	\centering
	\includegraphics[width=16cm, height=10cm]{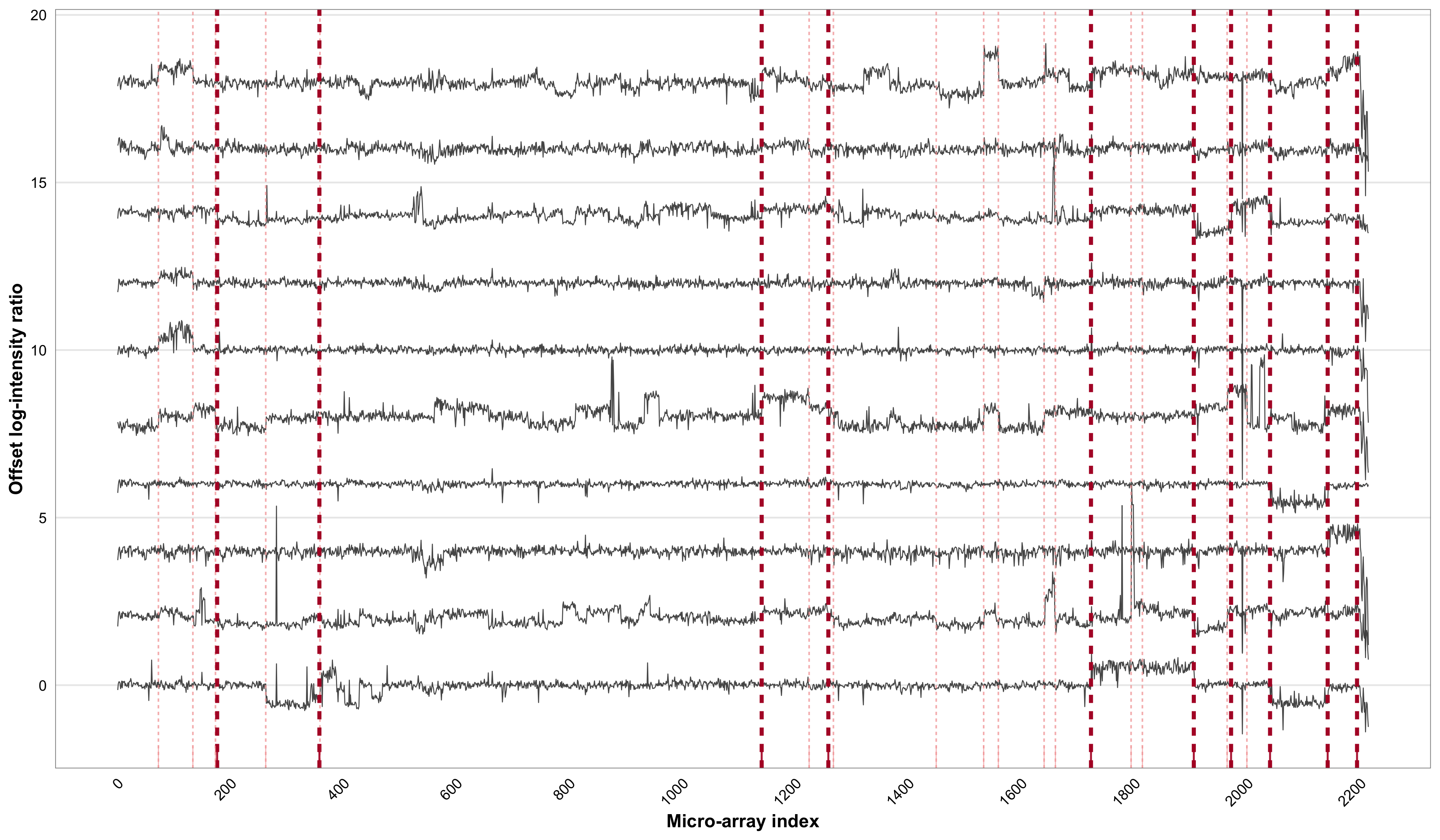}
	\caption{{Detected change-points for 10 individuals analyzed jointly. The bold dashed lines (dark red) represent change-points selected by the OPTICS FWER control criteria, while the solid light lines (light red) indicate additional change-points identified by the default stopping criteria in the \texttt{wbs} R package.}{\label{fig:joint}}}
\end{figure}

}

\section{Conclusion}
\label{sec:6}
Determining the number of change-points is a fundamental problem in literature. Rather than offering a single-point estimate, we propose a testing framework designed to construct a confidence set for the true number of change-points. The proposed method, named OPTICS, rigorously covers the true number with the predetermined confidence level under mild conditions. Additionally, we delve into studying the cardinality of OPTICS to ensure the obtained set is nontrivial and informative. The cardinality and coverage rate of OPTICS can also be utilized to assess the efficacy of base change-point detection algorithms. {Furthermore, we implement a multiple-splitting approach to stabilize OPTICS. We also extend this framework to accommodate high-dimensional datasets and develop a robust version capable of handling $m$-dependent and heavy-tailed distributions.}

{
\section{Supplementary Materials}
The Supplementary Material provides proofs for Theorems \ref{thm1}, \ref{thm2}, and \ref{thm3}, along with an extended literature review on statistical inference for change-point detection. We also include guidelines for selecting loss functions across various models and expanded simulation studies covering variance, network, and covariance change-points, multiple mean changes with heavy-tailed as well as $m$-dependent distributions. Finally, additional real-data results are provided; the dataset used in the real data analysis is available in the \texttt{ecp} R package.
}
%\appendix

\clearpage

	\def\spacingset#1{\renewcommand{\baselinestretch}%
		{#1}\small\normalsize} \spacingset{1}

	%%%%%%%%%%%%%%%%%%%%%%%%%%%%%%%%%%%%%%%%%%%%%%%%%%%%%%%%%%%%%%%%%%%%%%%%%%%%%%
	
	\if1\blind
	{
		\title{\bf Supplement to ``OPTICS: Order-Preserved Test-Inverse Confidence Set for Number of Change-Points"}
		\author{Ao Sun and Jingyuan Liu
			%\thanks{
			%The authors gratefully acknowledge \textit{please remember to list all relevant funding sources in the unblinded version}}\hspace{.2cm}\\
			MOE Key Laboratory of Econometrics, Department of Statistics, \\
			School of Economics, Wang Yanan Institute for Studies in Economics \\
			and Fujian Key Lab of Statistics, Xiamen University, P.R China}
		\maketitle
	} \fi
	
	\if0\blind
	{
		\bigskip
		\bigskip
		\bigskip
		\begin{center}
			{\LARGE\bf Supplement to ``OPTICS: Order-Preserved Test-Inverse Confidence Set for Number of Change-Points"}
		\end{center}
		\medskip
	} \fi
	
	\bigskip
	
	\spacingset{1.9} % DON'T change the spacing!

	% \renewcommand{\thecondition}{S.\arabic{condition}}
	% \renewcommand{\thesubsection}{S.\arabic{subsection}}
	% \renewcommand{\thelemma}{S.\arabic{lemma}}
	% \renewcommand{\thetheorem}{S.\arabic{theorem}}
	% \renewcommand{\theequation}{S.\arabic{equation}}
	% \renewcommand{\thetable}{S.\arabic{table}}
	% \renewcommand{\thefigure}{S.\arabic{figure}}

    % Reset the counters to start at 1 for the supplementary section
\setcounter{section}{0}
\setcounter{equation}{0}
\setcounter{table}{0}
\setcounter{figure}{0}
% (Add \setcounter{lemma}{0}, \setcounter{theorem}{0}, etc. if they don't reset automatically)

% Redefine the labeling formats
\renewcommand{\thesection}{S.\arabic{section}}
\renewcommand{\thesubsection}{\thesection.\arabic{subsection}} % Outputs S.1.1, S.1.2, etc.
\renewcommand{\thecondition}{S.\arabic{condition}}
\renewcommand{\thelemma}{S.\arabic{lemma}}
\renewcommand{\thetheorem}{S.\arabic{theorem}}
\renewcommand{\theequation}{S.\arabic{equation}}
\renewcommand{\thetable}{S.\arabic{table}}
\renewcommand{\thefigure}{S.\arabic{figure}}

This supplementary material contains the all technique proofs, additional literature review, additional simulations and real data results of the main paper.

\section{Proofs for Theorems in Section \ref{sec3.1}}
In the following proofs, we denote $c$ and $C$ as generic constants which may differ line to line. 
\subsection{Auxiliary lemmas}

The first lemma gives a uniform concentration of sample means $\bar{\bs}_{K,i}^O$ for all $i = 1,\ldots, n$ and $K \in \calM$.  Let $\| \bx \|_1 =  \sum_{j=1}^d |x_j|$ for a vector $\bx \in \mathbb{R}^d$.
\begin{lemma}\label{lemma1}
Under Conditions \ref{con1} and \ref{con2}, there exists a positive constant $c$ such that
$$
\Pr\left\{
\max_{K \in \calM}\max_{1 \le i \le n}
\bigl\|\bar{\bs}_{K,i}^O - \E(\bar{\bs}_{K,i}^O)\bigr\|_1
\gtrsim n^{-1/2}\log(n)
\right\}
\lesssim n^{-c}.
$$
\end{lemma}

\begin{proof}
Recall that $
\bs_i=\E(\bs_i)+\beps_i$,$
\beps_i:=\bs_i-\E(\bs_i)
=(\epsilon_{i1},\ldots,\epsilon_{id})^\top$ for $i=1,\ldots,n$.
By Condition \ref{con1}(i), for each $j=1,\ldots,d$, $
\|\epsilon_{ij}\|_{\psi_1}\le M_1,
\qquad i=1,\ldots,n$.
Since we do not consider the high-dimensional setting here, $d$ is fixed.

For each candidate segmentation $T_K=\{\tau_1^K,\ldots,\tau_K^K\}$, the quantity
$\bar{\bs}_{K,i}^O$ is the sample average of $\{\bs_l\}$ over the segment containing $i$.
Hence,
$$
\bar{\bs}_{K,i}^O-\E(\bar{\bs}_{K,i}^O)
=
\frac{1}{b-a+1}\sum_{l=a}^b \beps_l
$$
for some subset $\{a,\ldots,b\}\subset\{1,\ldots,n\}$ determined by $K$ and $i$. By Condition
\ref{con2}, every such interval has length at least of order $n/\log(n)$. Therefore, if we define
$$
\Theta_n
:=
\left\{
(a,b): 1\le a\le b\le n,\ \ b-a+1 \gtrsim \frac{n}{\log(n)}
\right\},
$$
then
\begin{equation}\label{eq:lemma1-step1}
\Pr\left\{
\max_{K\in\calM}\max_{1\le i\le n}
\|\bar{\bs}_{K,i}^O-\E(\bar{\bs}_{K,i}^O)\|_1
\ge x_n
\right\}
\le
\Pr\left\{
\max_{(a,b)\in\Theta_n}
\left\|
\frac{1}{b-a+1}\sum_{l=a}^b \beps_l
\right\|_1
\ge x_n
\right\}.
\end{equation}
Now fix $(a,b)\in\Theta_n$ and let $m=b-a+1$. Since
$$
\left\|
\frac{1}{m}\sum_{l=a}^b \beps_l
\right\|_1
\le
\sum_{j=1}^d
\left|
\frac{1}{m}\sum_{l=a}^b \epsilon_{l j}
\right|,
$$
we have, by the union bound,
\begin{align}
&\Pr\left\{
\max_{(a,b)\in\Theta_n}
\left\|
\frac{1}{b-a+1}\sum_{l=a}^b \beps_l
\right\|_1
\ge x_n
\right\}
\nonumber\\
&\qquad\le
\sum_{(a,b)\in\Theta_n}
\Pr\left\{
\left\|
\frac{1}{b-a+1}\sum_{l=a}^b \beps_l
\right\|_1
\ge x_n
\right\}
\nonumber\\
&\qquad\le
\sum_{(a,b)\in\Theta_n}\sum_{j=1}^d
\Pr\left\{
\left|
\frac{1}{b-a+1}\sum_{l=a}^b \epsilon_{l j}
\right|
\ge \frac{x_n}{d}
\right\}.
\label{eq:lemma1-step2}
\end{align}

For each fixed $(a,b)\in\Theta_n$ and $j\in\{1,\ldots,d\}$, the random variables
$\epsilon_{aj},\ldots,\epsilon_{bj}$ are independent, mean zero, and uniformly sub-exponential.
Hence, by Bernstein's inequality for sub-exponential random variables, there exist positive
constants $c_1$ and $c_2$, depending only on $M_1$, such that
$$
\Pr\left\{
\left|
\frac{1}{m}\sum_{l=a}^b \epsilon_{l j}
\right|
\ge t
\right\}
\le
2\exp\left[
-c_1 m \min\{t^2,t\}
\right],
\qquad t>0.
$$
Applying this bound with $t=x_n/d$ and using $m\gtrsim n/\log(n)$ for $(a,b)\in\Theta_n$,
we obtain
\begin{equation}\label{eq:lemma1-bernstein}
\Pr\left\{
\left|
\frac{1}{b-a+1}\sum_{l=a}^b \epsilon_{l j}
\right|
\ge \frac{x_n}{d}
\right\}
\lesssim
\exp\left\{
-c_2 \frac{n}{\log(n)} \min\{x_n^2,x_n\}
\right\}.
\end{equation}
Since $|\Theta_n|\le n^2$ and $d$ is fixed, combining
\eqref{eq:lemma1-step1}--\eqref{eq:lemma1-bernstein} yields
$$
\Pr\left\{
\max_{K\in\calM}\max_{1\le i\le n}
\|\bar{\bs}_{K,i}^O-\E(\bar{\bs}_{K,i}^O)\|_1
\ge x_n
\right\}
\lesssim
n^2
\exp\left\{
-c_2 \frac{n}{\log(n)} \min\{x_n^2,x_n\}
\right\}.
$$

Now choose $
x_n=A n^{-1/2}\log(n)$,
for a sufficiently large constant $A>0$. Then $x_n\to 0$, so for all sufficiently large $n$, $
\min\{x_n^2,x_n\}=x_n^2=A^2 n^{-1}\log^2(n)$,
and therefore
$$
\frac{n}{\log(n)}\min\{x_n^2,x_n\}
=
A^2\log(n).
$$
It follows that
$$
\Pr\left\{
\max_{K\in\calM}\max_{1\le i\le n}
\|\bar{\bs}_{K,i}^O-\E(\bar{\bs}_{K,i}^O)\|_1
\ge A n^{-1/2}\log(n)
\right\}
\lesssim
n^{\,2-c_2A^2}.
$$
By taking $A$ sufficiently large, we obtain
$$
\Pr\left\{
\max_{K\in\calM}\max_{1\le i\le n}
\|\bar{\bs}_{K,i}^O-\E(\bar{\bs}_{K,i}^O)\|_1
\gtrsim n^{-1/2}\log(n)
\right\}
\lesssim n^{-c}
$$
for some positive constant $c$. This completes the proof.
\end{proof}

Next lemma shows $(\xi_{K,J}^{(i)}  - \E[\xi_{K,J}^{(i)}])/ \sigma_{K,J}$ has sub-exponential tail under $H_{0,K}$ to fulfill our main theorem. 

\begin{lemma}\label{lemma2}
Let $
\xi_{K,J}^{(i)}
:=
\|\bs_i^E-\bar{\bs}_{K,i}^O\|_2^2-\|\bs_i^E-\bar{\bs}_{J,i}^O\|_2^2$,
and let $\mathcal F_O$ denote the $\sigma$-field generated by the observed sample used to construct
$\{\bar{\bs}_{K,i}^O: K\in\calM,\ 1\le i\le n\}$. Define the conditional variance scale $
\sigma_{K,J}^2
:=
\frac{1}{n}\sum_{i=1}^n \Var\!\left(\xi_{K,J}^{(i)}\mid \mathcal F_O\right)$.
Assume Conditions \ref{con1} and \ref{con2} hold. In addition, assume that for the pair $(K,J)$ under consideration,
the detected change-point sets $\calT_K$ and $\calT_J$ are nested. Then, under $H_{0,K}$, on the event
\begin{equation}\label{event:E}
 E_n
:=
\left\{
\max_{K\in\calM}\max_{1\le i\le n}
\|\bar{\bs}_{K,i}^O-\E(\bar{\bs}_{K,i}^O)\|_1
\le c n^{-1/2}\log(n)
\right\},   
\end{equation}
there exists a constant $C>0$ such that
$$
\max_{1\le i\le n}
\left\|
\frac{\xi_{K,J}^{(i)}-\E(\xi_{K,J}^{(i)}\mid \mathcal F_O)}{\sigma_{K,J}}
\right\|_{\psi_1\mid \mathcal F_O}
\le
C\sqrt{\log(n)}.
$$
Consequently,
$$
\Pr\left\{
\max_{1\le i\le n}
\left\|
\frac{\xi_{K,J}^{(i)}-\E(\xi_{K,J}^{(i)}\mid \mathcal F_O)}{\sigma_{K,J}}
\right\|_{\psi_1\mid \mathcal F_O}
\lesssim \sqrt{\log(n)}
\right\}
\to 1.
$$
\end{lemma}

\begin{proof}
By Lemma \ref{lemma1}, $
\Pr(E_n)\to 1$. We work on the event $E_n$ throughout the proof. For fixed $K,J\in\calM$ and $1\le i\le n$, expand
$$
\begin{aligned}
\xi_{K,J}^{(i)}
&=
\|\bs_i^E-\bar{\bs}_{K,i}^O\|_2^2-\|\bs_i^E-\bar{\bs}_{J,i}^O\|_2^2 \\
&=
\|\bar{\bs}_{K,i}^O\|_2^2-\|\bar{\bs}_{J,i}^O\|_2^2
-2(\bs_i^E)^\top(\bar{\bs}_{K,i}^O-\bar{\bs}_{J,i}^O) \\
&=
(\bar{\bs}_{K,i}^O+\bar{\bs}_{J,i}^O)^\top \Delta_{K,J,i}
-2(\bs_i^E)^\top \Delta_{K,J,i},
\end{aligned}
$$ 
where $
\Delta_{K,J,i}:=\bar{\bs}_{K,i}^O-\bar{\bs}_{J,i}^O$. Conditioning on $\mathcal F_O$, the vectors $\bar{\bs}_{K,i}^O$, $\bar{\bs}_{J,i}^O$, and hence $\Delta_{K,J,i}$ are deterministic. Therefore,
$$
\xi_{K,J}^{(i)}-\E\!\left(\xi_{K,J}^{(i)}\mid \mathcal F_O\right)
=
-2\Bigl\{(\bs_i^E)^\top \Delta_{K,J,i}
-\E\!\left((\bs_i^E)^\top \Delta_{K,J,i}\mid \mathcal F_O\right)\Bigr\}.
$$
Since $\bar{\bs}_{K,i}^O$ and $\bar{\bs}_{J,i}^O$ are $\mathcal F_O$-measurable, it follows from Condition \ref{con1}(i) and the basic properties of the Orlicz norm that
\begin{equation}\label{eq:lemma2-psi1-upper}
\left\|
\xi_{K,J}^{(i)}-\E(\xi_{K,J}^{(i)}\mid \mathcal F_O)
\right\|_{\psi_1\mid \mathcal F_O}
\lesssim
\|\Delta_{K,J,i}\|_2.
\end{equation}
Indeed, writing $\Delta_{K,J,i}=(\delta_{i1},\ldots,\delta_{id})^\top$ and using that $d$ is fixed,
$$
\left\|(\bs_i^E)^\top\Delta_{K,J,i}\right\|_{\psi_1\mid \mathcal F_O}
\le
\sum_{j=1}^d |\delta_{ij}|\, \|s_{ij}^E\|_{\psi_1}
\lesssim
\|\Delta_{K,J,i}\|_1
\lesssim
\|\Delta_{K,J,i}\|_2.
$$

Next, we study the denominator. Since
$$
\xi_{K,J}^{(i)}-\E(\xi_{K,J}^{(i)}\mid \mathcal F_O)
=
-2\Bigl\{(\bs_i^E)^\top \Delta_{K,J,i}
-\E\!\left((\bs_i^E)^\top \Delta_{K,J,i}\mid \mathcal F_O\right)\Bigr\},
$$
we have
$$
\Var\!\left(\xi_{K,J}^{(i)}\mid \mathcal F_O\right)
=
4\,\Delta_{K,J,i}^\top \Var(\bs_i^E)\Delta_{K,J,i}.
$$
By Condition \ref{con1}(iii), there exists a constant $c_0>0$ such that
\begin{equation}\label{eq:lemma2-var-lower}
\Var\!\left(\xi_{K,J}^{(i)}\mid \mathcal F_O\right)
\ge
c_0 \|\Delta_{K,J,i}\|_2^2.
\end{equation}
Hence,
$$
\sigma_{K,J}^2
=
\frac{1}{n}\sum_{i=1}^n \Var\!\left(\xi_{K,J}^{(i)}\mid \mathcal F_O\right)
\gtrsim
\frac{1}{n}\sum_{i=1}^n \|\Delta_{K,J,i}\|_2^2.
$$

Now use the additional nesting assumption on $\calT_K$ and $\calT_J$. Since the two segmentations are nested,
the vector sequence $\{\Delta_{K,J,i}\}_{i=1}^n$ is piecewise constant on the finer partition, and each constant block has length at least of order $n/\log(n)$ by Condition \ref{con2}. Therefore, if
$$
M_{K,J}:=\max_{1\le i\le n}\|\Delta_{K,J,i}\|_2,
$$
then there are at least $c_1 n/\log(n)$ indices $i$ such that $\|\Delta_{K,J,i}\|_2=M_{K,J}$ for some constant $c_1>0$. It follows that
\begin{equation}\label{eq:lemma2-avg-max}
\frac{1}{n}\sum_{i=1}^n \|\Delta_{K,J,i}\|_2^2
\gtrsim
\frac{1}{\log(n)} M_{K,J}^2.
\end{equation}
Combining this with \eqref{eq:lemma2-var-lower}, we obtain
\begin{equation}\label{eq:lemma2-sigma-lower}
\sigma_{K,J}
\gtrsim
\frac{1}{\sqrt{\log(n)}}\, M_{K,J}.
\end{equation}

Finally, combining \eqref{eq:lemma2-psi1-upper} and \eqref{eq:lemma2-sigma-lower} yields, for every $1\le i\le n$,
$$
\left\|
\frac{\xi_{K,J}^{(i)}-\E(\xi_{K,J}^{(i)}\mid \mathcal F_O)}{\sigma_{K,J}}
\right\|_{\psi_1\mid \mathcal F_O}
\lesssim
\frac{\|\Delta_{K,J,i}\|_2}{M_{K,J}/\sqrt{\log(n)}}
\le
C\sqrt{\log(n)}
$$
for some constant $C>0$. Taking the maximum over $1\le i\le n$ gives
$$
\max_{1\le i\le n}
\left\|
\frac{\xi_{K,J}^{(i)}-\E(\xi_{K,J}^{(i)}\mid \mathcal F_O)}{\sigma_{K,J}}
\right\|_{\psi_1\mid \mathcal F_O}
\le
C\sqrt{\log(n)}
$$
on $E_n$. Since $\Pr(E_n)\to 1$, the conclusion follows.
\end{proof}

\begin{lemma}\label{lemma:mean}
Let $
\hat{\delta}_{K,J}
:=
\frac{1}{n}\sum_{i=1}^n \xi_{K,J}^{(i)}$, $\delta_{K,J}
:=
\frac{1}{n}\sum_{i=1}^n \E\!\left(\xi_{K,J}^{(i)} \mid \mathcal F_O\right)$,
where $\mathcal F_O$ denotes the $\sigma$-field generated by the observed sample used to construct
$\{\bar{\bs}_{K,i}^O: K\in\calM,\ 1\le i\le n\}$. Under Conditions \ref{con1} and \ref{con2}, and assume that for the pair $(K,J)$ under consideration,
the detected change-point sets $\calT_K$ and $\calT_J$ are nested, there exists a positive constant $c$ such that
$$
\Pr\left\{
\max_{J\ne K}
\left|
\frac{\hat{\delta}_{K,J}-\delta_{K,J}}{\sigma_{K,J}}
\right|
\gtrsim
n^{-1/2}\log(n)
\right\}
\lesssim n^{-c},
$$
where $
\sigma_{K,J}^2
:=
\frac{1}{n}\sum_{i=1}^n
\Var\!\left(\xi_{K,J}^{(i)} \mid \mathcal F_O\right)$ for  $J\in \calM\backslash\{K\}$.
\end{lemma}

\begin{proof}
Recall that by Lemma \ref{lemma1}, $
\Pr(E_n^c)\lesssim n^{-c_1}$
for some positive constant $c_1$, where $E_n$ is defined in \eqref{event:E}. Fix $J\neq K$, and define
$$
Z_{i,J}
:=
\frac{\xi_{K,J}^{(i)}-\E(\xi_{K,J}^{(i)}\mid \mathcal F_O)}{\sigma_{K,J}},
\qquad i=1,\ldots,n.
$$
Then
$$
\frac{\hat{\delta}_{K,J}-\delta_{K,J}}{\sigma_{K,J}}
=
\frac{1}{n}\sum_{i=1}^n Z_{i,J}.
$$
Conditional on $\mathcal F_O$, the random variables $Z_{1,J},\ldots,Z_{n,J}$ are independent and satisfy $
\E(Z_{i,J}\mid \mathcal F_O)=0$ for $i=1,\ldots,n$.
Moreover, by Lemma \ref{lemma2}, on the event $E_n$, $
\max_{1\le i\le n}\|Z_{i,J}\|_{\psi_1\mid \mathcal F_O}
\lesssim \sqrt{\log(n)}$.

Therefore, conditional on $\mathcal F_O$ and on the event $E_n$, Bernstein's inequality for independent centered sub-exponential random variables yields
$$
\Pr\left\{
\left|
\frac{1}{n}\sum_{i=1}^n Z_{i,J}
\right|
\ge x_n
\ \middle|\ \mathcal F_O
\right\}
\le
2\exp\left[
-c_2 n
\min\left\{
\frac{x_n^2}{\log(n)},
\frac{x_n}{\sqrt{\log(n)}}
\right\}
\right]
$$
for some positive constant $c_2$. Hence, again on $E_n$,
$$
\Pr\left\{
\left|
\frac{\hat{\delta}_{K,J}-\delta_{K,J}}{\sigma_{K,J}}
\right|
\ge x_n
\ \middle|\ \mathcal F_O
\right\}
\le
2\exp\left[
-c_2 n
\min\left\{
\frac{x_n^2}{\log(n)},
\frac{x_n}{\sqrt{\log(n)}}
\right\}
\right].
$$
Applying the union bound over $J\in\calM\backslash\{K\}$, we obtain on $E_n$,
$$
\Pr\left\{
\max_{J\ne K}
\left|
\frac{\hat{\delta}_{K,J}-\delta_{K,J}}{\sigma_{K,J}}
\right|
\ge x_n
\ \middle|\ \mathcal F_O
\right\}
\le
2|\calM|
\exp\left[
-c_2 n
\min\left\{
\frac{x_n^2}{\log(n)},
\frac{x_n}{\sqrt{\log(n)}}
\right\}
\right].
$$
Since $|\calM|\le K_{\max}$, it follows that
$$
\Pr\left\{
\max_{J\ne K}
\left|
\frac{\hat{\delta}_{K,J}-\delta_{K,J}}{\sigma_{K,J}}
\right|
\ge x_n,
\ E_n
\right\}
\le
2K_{\max}
\exp\left[
-c_2 n
\min\left\{
\frac{x_n^2}{\log(n)},
\frac{x_n}{\sqrt{\log(n)}}
\right\}
\right].
$$

Now choose $
x_n=A n^{-1/2}\log(n)$, for a sufficiently large constant $A>0$. Then
$$
n\frac{x_n^2}{\log(n)}=A^2\log(n),
\qquad
n\frac{x_n}{\sqrt{\log(n)}}=A\sqrt{n\log(n)}.
$$
Hence, for all sufficiently large $n$,
$$
n
\min\left\{
\frac{x_n^2}{\log(n)},
\frac{x_n}{\sqrt{\log(n)}}
\right\}
=
A^2\log(n),
$$
and therefore
$$
\Pr\left\{
\max_{J\ne K}
\left|
\frac{\hat{\delta}_{K,J}-\delta_{K,J}}{\sigma_{K,J}}
\right|
\ge A n^{-1/2}\log(n),
\ E_n
\right\}
\le
2K_{\max} n^{-c_2A^2}.
$$
If $K_{\max}$ is fixed or grows at most polynomially in $n$, then by taking $A$ sufficiently large,
$2K_{\max} n^{-c_2A^2}\lesssim n^{-c_3}$ for some positive constant $c_3$. Finally,
$$
\begin{aligned}
 &\Pr\left\{
\max_{J\ne K}
\left|
\frac{\hat{\delta}_{K,J}-\delta_{K,J}}{\sigma_{K,J}}
\right|
\ge A n^{-1/2}\log(n)
\right\}\\
\le&
\Pr(E_n^c)
+
\Pr\left\{
\max_{J\ne K}
\left|
\frac{\hat{\delta}_{K,J}-\delta_{K,J}}{\sigma_{K,J}}
\right|
\ge A n^{-1/2}\log(n),
\ E_n
\right\}
\lesssim n^{-c}   
\end{aligned}
$$
for some positive constant $c$. This completes the proof.
\end{proof}

Let $\sigma^2_{K,J} = \sum_{i=1}^n\E[(\xi_{K,J}^{(i)} )^2]/n$ and $ \hat{\delta}_{K,J} = \sum_{i=1}^n \xi_{K,J}^{(i)} / n$ for $J \in \calM \backslash \{K\}$. The sample second moment matrix is defined as
$
\widehat{\Gamma}^{(K)} = \left(\frac{1}{n}  \sum_{i=1}^n  \xi_{K,J}^{(i)}  \xi_{K,J^\prime}^{(i)}\right)_{J, J^\prime} \in \mathbb{R}^{(|\calM| - 1) \times (|\calM| - 1)}$ with $J, J^\prime \in \calM \backslash \{K\} $ and its population analog $\Gamma^{(K)} = \left(\frac{1}{n} \sum_{i=1}^n\E(\xi_{K,J}^{(i)} \xi_{K,J^\prime}^{(i)}\mid \mathcal{F}_O) \right)_{J, J^\prime}$, where $\mathcal F_O$ denotes the $\sigma$-field generated by the observed sample used to construct
$\{\bar{\bs}_{K,i}^O: K\in\calM,\ 1\le i\le n\}$. Besides, let $D^{(K)} = \left(\diag(\Gamma^{(K)})\right)^{1/2} $ and $\widehat{D}^{(K)} = \left(\diag(\widehat{\Gamma}^{(K)})\right)^{1/2} $.
Furthermore, let $\widehat{H}^{(K)} = (\widehat{D}^{(K)})^{-1} \widehat{\Gamma}^{(K)} (\widehat{D}^{(K)})^{-1}$ and $H^{(K)} = (D^{(K)})^{-1}\Gamma^{(K)} (D^{(K)})^{-1} $.

\begin{lemma}\label{lemma3}
Assume Conditions \ref{con1} and \ref{con2} hold. In addition, assume the nesting condition in Lemma \ref{lemma2}, and that there exists a constant $c_0>0$ such that $
\min_{J\in\calM\backslash\{K\}} D^{(K)}_{JJ}\ge c_0$ with probability tending to one. If $K_{\max} \asymp \log(n)$, then there exists a positive constant $c$ such that
$$
\Pr\left\{
\max_{J,J^\prime\in\calM\backslash\{K\}}
\left|
\bigl(\widehat H^{(K)}-H^{(K)}\bigr)_{J,J^\prime}
\right|
\gtrsim
n^{-1/2}\log^{3/2}(n)
\right\}
\lesssim n^{-c}.
$$
\end{lemma}

\begin{proof}
Recall that by Lemma \ref{lemma1}, $
\Pr(E_n^c)\lesssim n^{-c_1}$
for some $c_1>0$, where $E_n$ is defined in \ref{event:E}. We work on the event $E_n$ throughout the proof. For $J\in\calM\backslash\{K\}$, define
$$
Z_{i,J}
:=
\frac{\xi_{K,J}^{(i)}-\E(\xi_{K,J}^{(i)}\mid \mathcal F_O)}{D^{(K)}_{JJ}},
\qquad i=1,\ldots,n.
$$
By Lemma \ref{lemma2}, on $E_n$, $
\max_{J\in\calM\backslash\{K\}}\max_{1\le i\le n}
\|Z_{i,J}\|_{\psi_1\mid\mathcal F_O}
\lesssim \sqrt{\log(n)}$. Hence, for each pair $(J,J^\prime)$,
$$
U_{i,JJ^\prime}
:=
Z_{i,J}Z_{i,J^\prime}
-
\E\!\left(Z_{i,J}Z_{i,J^\prime}\mid\mathcal F_O\right)
$$
is conditionally centered and satisfies $
\max_{J,J^\prime}\max_{1\le i\le n}
\|U_{i,JJ^\prime}\|_{\psi_{1/2}\mid\mathcal F_O}
\lesssim \log(n)$,
because the product of two sub-exponential random variables is sub-Weibull of order $1/2$.

Now note that
$$
\begin{aligned}
\frac{\widehat{\Gamma}^{(K)}_{JJ^\prime}-\Gamma^{(K)}_{JJ^\prime}}
{D^{(K)}_{JJ}D^{(K)}_{J^\prime J^\prime}}
&=
\frac{1}{n}\sum_{i=1}^n
\left[
\frac{\xi_{K,J}^{(i)}\xi_{K,J^\prime}^{(i)}}{D^{(K)}_{JJ}D^{(K)}_{J^\prime J^\prime}}
-
\E\!\left(
\frac{\xi_{K,J}^{(i)}\xi_{K,J^\prime}^{(i)}}{D^{(K)}_{JJ}D^{(K)}_{J^\prime J^\prime}}
\middle|\mathcal F_O
\right)
\right] \\
&=
\frac{1}{n}\sum_{i=1}^n U_{i,JJ^\prime}
+
R_{JJ^\prime},
\end{aligned}
$$
where
$$
R_{JJ^\prime}
=
\frac{1}{n}\sum_{i=1}^n
\frac{
\E(\xi_{K,J}^{(i)}\mid\mathcal F_O)\E(\xi_{K,J^\prime}^{(i)}\mid\mathcal F_O)
}{
D^{(K)}_{JJ}D^{(K)}_{J^\prime J^\prime}
}.
$$
By the proof of Lemma \ref{lemma2}, on $E_n$,
$$
\max_{J}\max_{1\le i\le n}
\left|
\frac{\E(\xi_{K,J}^{(i)}\mid\mathcal F_O)}{D^{(K)}_{JJ}}
\right|
\lesssim \sqrt{\log(n)},
$$
so $|R_{JJ^\prime}|\lesssim \log(n)$ uniformly in $(J,J^\prime)$. Since $R_{JJ^\prime}$ appears in both
$\widehat{\Gamma}^{(K)}_{JJ^\prime}$ and $\Gamma^{(K)}_{JJ^\prime}$, it cancels in the centered difference above. Therefore it suffices to control the average of $\{U_{i,JJ^\prime}\}_{i=1}^n$.

Conditional on $\mathcal F_O$, the variables $U_{1,JJ^\prime},\ldots,U_{n,JJ^\prime},$ are independent, centered, and uniformly sub-Weibull of order $1/2$ with parameter of order $\log(n)$. Hence, Bernstein's inequality for sub-Weibull random variables yields
$$
\Pr\left\{
\left|
\frac{1}{n}\sum_{i=1}^n U_{i,JJ^\prime}
\right|
\ge t
\ \middle|\ \mathcal F_O
\right\}
\le
2\exp\left[
-c_2
\min\left\{
\frac{nt^2}{\log^2(n)},
\left(\frac{nt}{\log(n)}\right)^{1/2}
\right\}
\right]
$$
for some $c_2>0$. Applying the union bound over all $(J,J^\prime)\in(\calM\backslash\{K\})^2$, we obtain on $E_n$,
\begin{equation}\label{eq:lemma3-gamma-bound}
\Pr\left\{
\max_{J,J^\prime}
\left|
\frac{\widehat{\Gamma}^{(K)}_{JJ^\prime}-\Gamma^{(K)}_{JJ^\prime}}
{D^{(K)}_{JJ}D^{(K)}_{J^\prime J^\prime}}
\right|
\ge t
\ \middle|\ \mathcal F_O
\right\}
\le
2(K_{\max}-1)^2
\exp\left[
-c_2
\min\left\{
\frac{nt^2}{\log^2(n)},
\left(\frac{nt}{\log(n)}\right)^{1/2}
\right\}
\right].
\end{equation}

Choose $t=A n^{-1/2}\log^{3/2}(n)$,
for a sufficiently large constant $A>0$. Then
$$
\frac{nt^2}{\log^2(n)}=A^2\log(n),
\qquad
\left(\frac{nt}{\log(n)}\right)^{1/2}
=
A^{1/2} n^{1/4}\log^{1/4}(n),
$$
so the first term determines the rate for all sufficiently large $n$. Since $K_{\max}$ is fixed or grows at most polynomially in $n$, \eqref{eq:lemma3-gamma-bound} implies that
\begin{equation}\label{eq:lemma3-middle}
\Pr\left\{
\max_{J,J^\prime}
\left|
\frac{\widehat{\Gamma}^{(K)}_{JJ^\prime}-\Gamma^{(K)}_{JJ^\prime}}
{D^{(K)}_{JJ}D^{(K)}_{J^\prime J^\prime}}
\right|
\gtrsim
n^{-1/2}\log^{3/2}(n),
\ E_n
\right\}
\lesssim n^{-c_3}
\end{equation}
for some $c_3>0$.

Next, for each $J$, $
\widehat D^{(K)}_{JJ}
=
\bigl(\widehat\Gamma^{(K)}_{JJ}\bigr)^{1/2}$, and $
D^{(K)}_{JJ}
=
\bigl(\Gamma^{(K)}_{JJ}\bigr)^{1/2}$.
By \eqref{eq:lemma3-middle} with $J=J^\prime$,
\[
\max_{J}
\left|
\widehat\Gamma^{(K)}_{JJ}-\Gamma^{(K)}_{JJ}
\right|
\lesssim
n^{-1/2}\log^{3/2}(n)
\]
with probability at least $1-O(n^{-c_3})$ on $E_n$. Since $\min_J D^{(K)}_{JJ}\ge c_0>0$ with probability tending to one, the map $x\mapsto x^{1/2}$ is Lipschitz on a neighborhood of $\{\Gamma^{(K)}_{JJ}\}_J$, and therefore
\[
\max_J
\left|
\widehat D^{(K)}_{JJ}-D^{(K)}_{JJ}
\right|
\lesssim
n^{-1/2}\log^{3/2}(n)
\]
with probability at least $1-O(n^{-c_4})$ on $E_n$. Consequently,
$$
\max_J
\left|
(\widehat D^{(K)})^{-1}_{JJ}-(D^{(K)})^{-1}_{JJ}
\right|
\lesssim
n^{-1/2}\log^{3/2}(n)
$$
with probability at least $1-O(n^{-c_4})$ on $E_n$.

Finally, write
\begin{align*}
\widehat H^{(K)}-H^{(K)}
&=
(\widehat D^{(K)})^{-1}\widehat\Gamma^{(K)}(\widehat D^{(K)})^{-1}
-
(D^{(K)})^{-1}\Gamma^{(K)}(D^{(K)})^{-1} \\
&=
\bigl[(\widehat D^{(K)})^{-1}-(D^{(K)})^{-1}\bigr]\widehat\Gamma^{(K)}(\widehat D^{(K)})^{-1} \\
&\quad
+(D^{(K)})^{-1}\bigl[\widehat\Gamma^{(K)}-\Gamma^{(K)}\bigr](\widehat D^{(K)})^{-1} \\
&\quad
+(D^{(K)})^{-1}\Gamma^{(K)}\bigl[(\widehat D^{(K)})^{-1}-(D^{(K)})^{-1}\bigr].
\end{align*}
Because $(K_{\max}-1)$ is finite or polynomially growing, and all diagonal entries of $D^{(K)}$ are bounded away from zero with probability tending to one, each term on the right-hand side is of order
$$
O_{P}\!\left(n^{-1/2}\log^{3/2}(n)\right)
$$
in the elementwise maximum norm. Therefore,
$$
\Pr\left\{
\max_{J,J^\prime}
\left|
\bigl(\widehat H^{(K)}-H^{(K)}\bigr)_{J,J^\prime}
\right|
\gtrsim
n^{-1/2}\log^{3/2}(n)
\right\}
\lesssim
\Pr(E_n^c)+n^{-c_5}
\lesssim n^{-c}
$$
for some positive constant $c$. This completes the proof.
\end{proof}

\subsection{Proof of Theorem \ref{thm1}}
We ignore the Monte Carlo variability from bootstrap resampling and regard $\hat p_K$ as the limiting bootstrap $p$-value when the bootstrap sample size $B\to\infty$. Let $
\gamma_n=n^{-1/2}\log^{3/2}(n)$ and $\eta_n=n^{-1/2}\log(n)$.
Define the event
$$
F_n = \left\{ \|\widehat H^{(K)}-H^{(K)}\|_\infty \le C_1\gamma_n,\ \max_{J\in\calM\backslash\{K\}} \left| \frac{\hat\delta_{K,J}-\delta_{K,J}}{D^{(K)}_{JJ}} \right| \le C_1\eta_n \right\}.
$$
By Lemma \ref{lemma:mean} and Lemma \ref{lemma3}, there exists a constant $c_1>0$ such that
$$
\Pr(F_n^c)\lesssim n^{-c_1}.
$$

Moreover, by the diagonal part of Lemma \ref{lemma3} and the assumption $\min_{J\in\calM\backslash\{K\}} D^{(K)}_{JJ}\ge c_0>0$ with probability tending to one, we have on $F_n$,
$$
\max_{J\in\calM\backslash\{K\}} \left| \frac{\widehat D^{(K)}_{JJ}}{D^{(K)}_{JJ}}-1 \right| \lesssim \gamma_n, \qquad \max_{J\in\calM\backslash\{K\}} \left| \frac{D^{(K)}_{JJ}}{\widehat D^{(K)}_{JJ}}-1 \right| \lesssim \gamma_n.
$$

Let $G^{(K)} = \left(\frac{1}{n} \sum_{i=1}^n \Cov\bigl(\xi_{K,J}^{(i)}, \xi_{K,J^\prime}^{(i)} \mid \mathcal{F}_O\bigr)\right)_{J,J^\prime \in \calM\backslash\{K\}}$ denote the conditional covariance matrix, and define its normalized counterpart $F^{(K)} = (D^{(K)})^{-1} G^{(K)} (D^{(K)})^{-1}$. Let $T^*(A)$ denote a centered Gaussian vector with covariance matrix $A$, and let $z_{\max}(\alpha,A)$ be the upper $\alpha$-quantile of $\|T^*(A)\|_\infty$.

We first prove part (1). On $F_n$,
$$
\begin{aligned}
&\sqrt n \max_{J\in\calM\backslash\{K\}} \frac{\hat\delta_{K,J}}{\widehat D^{(K)}_{JJ}} \\
\le& \sqrt n \max_{J\in\calM\backslash\{K\}} \frac{\hat\delta_{K,J}-\delta_{K,J}}{D^{(K)}_{JJ}} \left(\frac{D^{(K)}_{JJ}}{\widehat D^{(K)}_{JJ}}\right) + \sqrt n \max_{J\in\calM\backslash\{K\}} \frac{\delta_{K,J}}{\widehat D^{(K)}_{JJ}} \\ 
\le& \sqrt n \max_{J\in\calM\backslash\{K\}} \frac{\hat\delta_{K,J}-\delta_{K,J}}{D^{(K)}_{JJ}} + C_2\sqrt n\,\eta_n\gamma_n + \sqrt n \max_{J\in\calM\backslash\{K\}} \frac{\delta_{K,J}}{D^{(K)}_{JJ}}(1+C_2\gamma_n)
\end{aligned}
$$
for some constant $C_2>0$.

Under the approximate null assumption, $\max_{J\in\calM\backslash\{K\}} \frac{\delta_{K,J}}{D^{(K)}_{JJ}} \le x_n (n\log n)^{-1/2}$ where $x_n=o(1)$, and therefore
$$
\sqrt n \max_{J\in\calM\backslash\{K\}} \frac{\delta_{K,J}}{D^{(K)}_{JJ}} = o\!\left((\log n)^{-1/2}\right) = o(1).
$$
Also,
$$
\sqrt n\,\eta_n\gamma_n = \sqrt n\cdot n^{-1/2}\log(n)\cdot n^{-1/2}\log^{3/2}(n) = n^{-1/2}\log^{5/2}(n) = o(1).
$$
Hence the previous display implies
$$
\sqrt n \max_{J\in\calM\backslash\{K\}} \frac{\hat\delta_{K,J}}{\widehat D^{(K)}_{JJ}} \le \sqrt n \max_{J\in\calM\backslash\{K\}} \frac{\hat\delta_{K,J}-\delta_{K,J}}{D^{(K)}_{JJ}} +o(1)
$$
uniformly on $F_n$.

Now define the centered normalized sum vector
$$
S_K = n^{-1/2} \left( \sum_{i=1}^n \frac{\xi_{K,J}^{(i)}-\E(\xi_{K,J}^{(i)}\mid \mathcal F_O)}{D^{(K)}_{JJ}} \right)_{J\in\calM\backslash\{K\}}.
$$
Conditional on $\mathcal F_O$, Lemma \ref{lemma2} provides the required tail control, and a Gaussian approximation for maxima of sums of independent random vectors yields
$$
\sup_{t\in\mathbb R} \left| \Pr\bigl(\|S_K\|_\infty\le t \mid \mathcal F_O\bigr) - \Pr\bigl(\|T^*(F^{(K)})\|_\infty\le t \mid \mathcal F_O\bigr) \right| = o(1).
$$
Consequently, for some sequence $\varpi_n\to 0$,
$$
\Pr\left\{ \sqrt n \max_{J\in\calM\backslash\{K\}} \frac{\hat\delta_{K,J}-\delta_{K,J}}{D^{(K)}_{JJ}} \ge z_{\max}(\alpha+\varpi_n,F^{(K)}) \right\} \le \alpha+\varpi_n+o(1).
$$

Next, note that $\Gamma^{(K)} = G^{(K)} + \mu^{(K)}(\mu^{(K)})^\top$, where $\mu^{(K)} = \left( \frac{\delta_{K,J}}{D^{(K)}_{JJ}} \right)_{J\in\calM\backslash\{K\}}$, so
$$
H^{(K)}-F^{(K)} = \mu^{(K)}(\mu^{(K)})^\top.
$$
Under the approximate null assumption, $\|\mu^{(K)}\|_\infty \le x_n(n\log n)^{-1/2} = o(1)$. Therefore, $\|H^{(K)}-F^{(K)}\|_\infty = o(1)$. By the continuity of the Gaussian max quantiles with respect to the covariance matrix entries, this implies
$$
z_{\max}(\alpha+\varpi_n,H^{(K)}) \ge z_{\max}(\alpha+\varpi_n,F^{(K)})-o(1).
$$
Furthermore, on $F_n$, we have $\|\widehat H^{(K)}-H^{(K)}\|_\infty \le C_1\gamma_n = o(1)$, and therefore
$$
z_{\max}(\alpha,\widehat H^{(K)}) \ge z_{\max}(\alpha+\varpi_n,H^{(K)})-o(1)
$$
for a possibly different sequence $\varpi_n\to 0$.

Combining the previous displays, we obtain
$$
\Pr\{\hat p_K\le \alpha\} = \Pr\left\{ \sqrt n \max_{J\in\calM\backslash\{K\}} \frac{\hat\delta_{K,J}}{\widehat D^{(K)}_{JJ}} \ge z_{\max}(\alpha,\widehat H^{(K)}) \right\} \le \alpha+o(1).
$$
Equivalently,
$$
\Pr\left\{H_{0,K}\text{ is not rejected at level }\alpha\right\} \ge 1-\alpha+o(1).
$$

We now prove part (2). Suppose there exists some $J^\ast\in\calM\backslash\{K\}$ such that
$$
\frac{\delta_{K,J^\ast}}{D^{(K)}_{J^\ast J^\ast}} \ge c\, n^{-1/2}\log(n)
$$
for a sufficiently large constant $c>0$. On $F_n$,
$$
\frac{\hat\delta_{K,J^\ast}}{\widehat D^{(K)}_{J^\ast J^\ast}} = \left( \frac{\delta_{K,J^\ast}}{D^{(K)}_{J^\ast J^\ast}} + \frac{\hat\delta_{K,J^\ast}-\delta_{K,J^\ast}}{D^{(K)}_{J^\ast J^\ast}} \right) \frac{D^{(K)}_{J^\ast J^\ast}}{\widehat D^{(K)}_{J^\ast J^\ast}}.
$$
Hence, on $F_n$,
$$
\frac{\hat\delta_{K,J^\ast}}{\widehat D^{(K)}_{J^\ast J^\ast}} \ge \left( c\,n^{-1/2}\log(n)-C_1 n^{-1/2}\log(n) \right)(1-C_2\gamma_n).
$$
By choosing $c$ sufficiently large, we obtain for all sufficiently large $n$,
$$
\frac{\hat\delta_{K,J^\ast}}{\widehat D^{(K)}_{J^\ast J^\ast}} \ge c_2 n^{-1/2}\log(n)
$$
for some constant $c_2>0$. Therefore,
$$
\sqrt n \max_{J\in\calM\backslash\{K\}} \frac{\hat\delta_{K,J}}{\widehat D^{(K)}_{JJ}} \ge c_2 \log(n)
$$
on $F_n$.

On the other hand, since $\alpha\ge n^{-1}$ and $K_{\max}\asymp \log(n)$, a union bound together with Mills' inequality implies
$$
z_{\max}(\alpha,\widehat H^{(K)}) \lesssim \sqrt{\log(K_{\max})+\log(\alpha^{-1})} \lesssim \sqrt{\log(n)}
$$
with probability tending to one. Consequently, on $F_n$ and for all sufficiently large $n$,
$$
\sqrt n \max_{J\in\calM\backslash\{K\}} \frac{\hat\delta_{K,J}}{\widehat D^{(K)}_{JJ}} > z_{\max}(\alpha,\widehat H^{(K)}).
$$
Thus,
$$
\Pr\{\hat p_K\le \alpha\} \ge \Pr(F_n)-o(1) = 1-o(1),
$$
which is equivalent to
$$
\Pr\left\{H_{0,K}\text{ is not rejected at level }\alpha\right\} = o(1).
$$
This completes the proof.

\section{Proofs for Theorems in Section \ref{sec3.2}}

\begin{lemma}
	\label{tech_lemma1}
	Suppose that Conditions \ref{con1}-\ref{con5} hold. Then, \\
	(i) $\max _{K \in \calM}\E\left\{\calS_{\epsilon^E}\left(\mathcal{T}_K\right)-\calS_{\epsilon^E}\left(\mathcal{T}_K \cup \mathcal{T}^{*}\right)\right\}=o\left(M_1 \log \log \bar{\lambda}\right)$.\\
	(ii) $\max _{K \in \calM}\E\left\{\calS_{\epsilon^E}\left(\mathcal{T}^*\right)-\calS_{\epsilon^E}\left(\mathcal{T}_K \cup \mathcal{T}^*\right)\right\}=o\left(M_1 \log \log \bar{\lambda}\right).$
\end{lemma}

\begin{proof}
	For a fixed $K$, let $m_k$ be the number of true change-points strictly between $\tau_{k}^K$ and $\tau_{k+1}^K$. We denote the merged partition points within this interval as $\tau_{k,0}^K := \tau_{k}^K < \tau_{k, 1}^K < \ldots < \tau_{k, m_k}^K < \tau_{k, m_k+1}^K := \tau_{k+1}^K$. Given $\calZ_O$, $\calT_K$ is a fixed change-point set. According to Lemma 14 in \cite{pein2021cross}, we have
	$$
	\E\left\{\calS_{\epsilon^E}\left(\calT_K\right)-\calS_{\epsilon^E}\left(\mathcal{T}_K \cup \mathcal{T}^*\right) \right\} \leq 2 \sum_{k=0}^K \sum_{l=0}^{m_k}\left(\tau_{k, l+1}^K-\tau_{k,l}^K\right)\E\left(\bar{\epsilon}_{\tau_{k, l}^K: \tau_{k, l+1}^K}\right)^2.
	$$
	Based on the independence of the errors $\epsilon_i$, for each term we have
	$$
	\left(\tau_{k, l+1}^K-\tau_{k,l }^K\right)\E\left(\bar{\epsilon}_{\tau_{k, l}^K: \tau_{k, l+1}^K}\right)^2 =\left(\tau_{k, l+1}^K-\tau_{k,l }^K\right)^{-1}\sum_{i =\tau_{k,l }^K+ 1 }^{\tau_{k, l+1}^K}\E[\epsilon_i^2] \le M_1.
	$$
	Note that the total number of segments in the merged partition $\calT_K \cup \calT^*$ is exactly $\sum_{k = 0}^K (m_k + 1) \le K + K^* + 1 \le 2K_{\max} + 1$. Hence, we have
	$$
	\E\left\{\calS_{\epsilon^E}\left(\mathcal{T}_K\right)-\calS_{\epsilon^E}\left(\mathcal{T}_K \cup \mathcal{T}^*\right) \right\} \leq 2M_1 (2K_{\max} + 1) \lesssim M_1 K_{\max}.
	$$
	Therefore,
	$
	\max_{K \in \calM} \E\left\{\calS_{\epsilon^E}\left(\mathcal{T}_K\right)-\calS_{\epsilon^E}\left(\mathcal{T}_K \cup \mathcal{T}^*\right) \right\} \lesssim M_1 K_{\max}
	$. By assuming $K_{\max} = o(\log\log\bar{\lambda})$ in Condition \ref{con3} (iii), we obtain $M_1 K_{\max} = o(M_1\log\log\bar{\lambda})$, which implies statement (i). The exact same deduction applies to the true change-point set $\calT^*$, yielding statement (ii), which we omit here.
\end{proof}

Let $\calT^*$ be the true change positions, $\calT_{K^*} $ be the estimated change positions when the number of change-points is correctly specified.

Let $\calT^*$ be the true change positions, and $\calT_{K^*} $ be the estimated change positions when the number of change-points is correctly specified.

\begin{lemma}
	\label{tech_lemma2}
	Suppose that Conditions \ref{con1}-\ref{con5} hold. Then\\
	(i) $\min _{K \in \calM_l}\E\left\{\calS_{s^E}\left(\mathcal{T}_K\right)-\calS_{s^E}\left(\mathcal{T}^*\right)\right\} \geq \underline{\lambda}M_3\min _{K \in \calM_l} \left(\sum_{\tau_k^* \in \mathcal{I}^*_{lK}} \Delta_k^2\right).$\\
	(ii) $\max _{K \in \calM}\E\left\{\calS_{s^E}\left(\mathcal{T}^*\right)-\calS_{s^E}\left(\mathcal{T}_K \cup \mathcal{T}^*\right)\right\}=o\left(M_1\log \log \bar{\lambda}\right).$\\
	(iii) $\E\left\{\calS_{s^E}\left(\mathcal{T}_{K^*}\right)-\calS_{s^E}\left(\mathcal{T}^*\right)\right\}=o\left(M_1 \log \log \bar{\lambda}\right).$
\end{lemma}

\begin{proof}
	The proof of this lemma is inspired by Lemma 17 in \cite{pein2021cross}.
	For statement (ii), we observe that $\calS_{s^E}\left(\mathcal{T}^*\right)-\calS_{s^E}\left(\mathcal{T}_K \cup \mathcal{T}^*\right)=\calS_{\epsilon^E}\left(\mathcal{T}^*\right)-\calS_{\epsilon^E}\left(\mathcal{T}_K \cup \mathcal{T}^*\right)$, since the true signal $\mu_i$ is constant on both partitions and cancels out. Hence, statement (ii) directly follows from Lemma \ref{tech_lemma1} (ii).
	
	Now, we show (iii). By definition,
	\begin{equation}
		\label{app:ref}
		\begin{aligned}
			\calS_{s^E}\left(\mathcal{T}_{K^*}\right) =& \sum_{k=0}^{K^*} \sum_{i=\tau_{k}^{K^*}+1}^{\tau^{K^*}_{ k+1}}\left(s_i^E-\bar{s}^E_{\tau_{k}^{K^*}: \tau_{ k+1}^{K^*}}\right)^2 =\sum_{k=0}^{K^*} \sum_{i=\tau_{k}^{K^*}+1}^{\tau^{K^*}_{ k+1}}\left(\mu_i + \epsilon^E_i-\bar{\mu}_{\tau_{k}^{K^*}: \tau_{ k+1}^{K^*}} - \bar{\epsilon}^E_{\tau_{k}^{K^*}: \tau_{ k+1}^{K^*}}\right)^2  \\
			=&\sum_{k=0}^{K^*} \sum_{i=\tau_{k}^{K^*}+1}^{\tau^{K^*}_{ k+1}}\left(\mu_i -\bar{\mu}_{\tau_{k}^{K^*}: \tau_{ k+1}^{K^*}} \right)^2 +2 \sum_{k=0}^{K^*} \sum_{i=\tau_{k}^{K^*}+1}^{\tau^{K^*}_{ k+1}} \epsilon_i^E\left( \mu_i - \bar{\mu}_{\tau_{k}^{K^*}: \tau_{ k+1}^{K^*}}\right)  + \sum_{i=1}^n (\epsilon_i^E)^2\\
			& - \sum_{k=0}^{K^*}\left(\tau_{ k+1}^{K^*}- \tau_{ k}^{K^*}\right)(\bar{\epsilon}_{\tau_{ k}^{K^*}: \tau_{ k+1}^{K^*}}^E)^2. 
		\end{aligned}
	\end{equation}
	Similarly, for $\calS_{s^E}\left(\mathcal{T}^*\right)$,
	$$
	\begin{aligned}
		\calS_{s^E}\left(\mathcal{T}^*\right) =& \sum_{k=0}^{K^*} \sum_{i=\tau_{k}^{*}+1}^{\tau^{*}_{ k+1}}\left(s_i^E-\bar{s}^E_{\tau_{k}^{*}: \tau_{ k+1}^{*}}\right)^2 = \sum_{k=0}^{K^*} \sum_{i=\tau_{k}^{*}+1}^{\tau^{*}_{ k+1}}\left(\epsilon_i^E-\bar{\epsilon}^E_{\tau_{k}^{*}: \tau_{ k+1}^{*}}\right)^2\\
		=& \sum_{i=1}^n (\epsilon_i^E)^2- \sum_{k=0}^{K^*}\left(\tau_{ k+1}^{*}-\tau_{ k}^{*}\right)(\bar{\epsilon}_{\tau_{ k}^{*}: \tau_{ k+1}^{*}}^E)^2.
	\end{aligned}
	$$ 
	The expectation of their difference is
	$$
	\begin{aligned}
		\E\left\{\calS_{s^E}\left(\mathcal{T}_{K^*}\right)-\calS_{s^E}\left(\mathcal{T}^*\right)\right\} 
		=&\sum_{k=0}^{K^*} \sum_{i=\tau_{k}^{K^*}+1}^{\tau^{K^*}_{ k+1}}\left(\mu_i-\bar{\mu}_{\tau_{k}^{K^*}: \tau_{ k+1}^{K^*}}\right)^2
		-\sum_{k=0}^{K^*}\left(\tau_{ k+1}^{K^*}-\tau_{ k}^{K^*}\right)\E\left(\bar{\epsilon}_{\tau_{ k}^{K^*}: \tau_{ k+1}^{K^*}}^E\right)^2\\
		&+\sum_{k=0}^{K^*}\left(\tau_{ k+1}^*-\tau_{ k}^*\right)\E\left(\bar{\epsilon}_{\tau_{ k}^*: \tau_{ k+1}^*}^E\right)^2 \\
		:=& A_1 - A_2 + A_3.
	\end{aligned}
	$$
	For $A_1$, according to the deduction in Lemma 17 in \cite{pein2021cross}, the first term is $A_1 = O\big( \sum_{k=1}^{K^*} b_n \Delta_k^2 \big)$. 
	For $A_2$, we notice that by Condition \ref{con1},
	$$
	A_2=\sum_{k=0}^{K^*} \left(\tau_{ k+1}^{K^*}-\tau_{ k}^{K^*}\right)^{-1} \sum_{i = \tau_{ k}^{K^*}+1}^{\tau_{ k+1}^{K^*}}\E(\epsilon_i^2) \lesssim M_1K^*.
	$$
	The exact same reasoning gives $A_3 \lesssim M_1 K^* $. Hence, 
	$$
	\E\left\{\calS_{s^E}\left(\mathcal{T}_{K^*}\right)-\calS_{s^E}\left(\mathcal{T}^*\right)\right\} \lesssim \sum_{k=1}^{K^*} b_n \Delta_k^2 + M_1 K^*.
	$$
	Statement (iii) then follows from Condition \ref{con4} (i), which guarantees $\sum_{k=1}^{K^*} b_n \Delta_k^2 = o(M_1 \log\log\bar{\lambda})$, and Condition \ref{con3} (i), which ensures $K^* = o(\log\log\bar{\lambda})$.
	
	We now show statement (i). Based on a similar deduction as \eqref{app:ref}, we can decompose the difference directly as:
	$$
	\E\left\{\calS_{s^E}\left(\mathcal{T}_K\right)-\calS_{s^E}\left(\mathcal{T}^*\right)\right\} = \sum_{k=0}^K \sum_{i=\tau_{k}^K+1}^{\tau_{ k+1}^K}\left(\mu_i-\bar{\mu}_{\tau_{k}^K: \tau_{k+1}^K}\right)^2 + \E\left\{\calS_{\epsilon^E}\left(\mathcal{T}_K\right)-\calS_{\epsilon^E}\left(\mathcal{T}^*\right)\right\}.
	$$
	For the noise component, Lemma \ref{tech_lemma1} yields
	\begin{equation}
		\label{tech_eq1}
		\begin{aligned}
			\max_{K \in \calM_l} \left|\E\left\{\calS_{\epsilon^E}\left(\mathcal{T}_K\right)-\calS_{\epsilon^E}\left(\mathcal{T}^*\right)\right\} \right| &\leq \max _{K \in \calM_l}\E\left\{\calS_{\epsilon^E}\left(\mathcal{T}_K\right)-\calS_{\epsilon^E}\left(\mathcal{T}_K \cup \mathcal{T}^*\right)\right\}\\
			&\quad +\max _{K \in \calM_l}\E\left\{\calS_{\epsilon^E}\left(\mathcal{T}^*\right)-\calS_{\epsilon^E}\left(\mathcal{T}_K \cup \mathcal{T}^*\right)\right\} \\
			&= o(M_1 \log \log \bar{\lambda}).
		\end{aligned}
	\end{equation}
	For the signal component, since $K \in \calM_l$ is a lack-of-fit model, Condition \ref{con4} (ii) guarantees there exists a subset of change-points $\tau_{k}^* \in \calI_{lK}^*$ such that there is no estimated change-point within $[\tau_{k}^* - \underline{\lambda}/4, \tau_{k}^* + \underline{\lambda}/4]$. Hence, using Condition \ref{con:LB2},
	$$
    \begin{aligned}
      \min_{K \in \calM_l}\sum_{k=0}^K \sum_{i=\tau_{k}^K+1}^{\tau_{ k+1}^K}\left(\mu_i-\bar{\mu}_{\tau_{k}^K: \tau_{k+1}^K}\right)^2  \ge& \min_{K \in \calM_l} \sum_{\tau_k^* \in \calI^*_{lK}}\sum_{i=\tau_k^*-\frac{\underline{\lambda}}{4}+1}^{\tau_k^*+\frac{\underline{\lambda}}{4}}\left(\mu_i-\bar{\mu}_{\tau_{k}^K: \tau_{k+1}^K}\right)^2 \\
      \ge& \min_{K \in \calM_l} M_3 \underline{\lambda} \sum_{\tau_k^* \in \calI^*_{lK}} \Delta_k^2.  
    \end{aligned}
	$$
	Summarizing the above results, we have
	$$
	\min_{K \in \calM_l} \E\left\{\calS_{s^E}\left(\mathcal{T}_K\right)-\calS_{s^E}\left(\mathcal{T}^*\right)\right\} \ge \min_{K \in \calM_l} M_3 \underline{\lambda} \sum_{\tau_k^* \in \calI^*_{lK}} \Delta_k^2 - o(M_1 \log \log \bar{\lambda}).
	$$
	By Condition \ref{con5}, the right-hand side is strictly dominated by the positive first term. Hence, statement (i) concludes.
\end{proof}

\subsection{Proof of Theorem \ref{thm2}}

Now, we focus on the main theorem. Our strategy is to apply Theorem \ref{thm1}, which requires showing that
$$
\max_{K \in \calM\backslash \{K^*\}} \frac{\delta_{K^*,K}}{\sigma_{K^*,K}} \le x_n \sqrt{\frac{1}{n \log(n)}},
$$
for sufficiently large $n$. Recall that $\sigma^2_{K^*,K}$ can be lower bounded as follows:
$$
\begin{aligned}
	\sigma^2_{K^*,K} \ge &\frac{1}{n} \sum_{i=1}^n \Var\left( \|\bs_i^E - \bar{\bs}_{K^*,i}^O\|_2^2 - \|\bs_i^E - \bar{\bs}_{K,i}^O \|_2^2\right)\\
	=& \frac{1}{n} \sum_{i=1}^n \Var\left(\|\bar{\bs}^O_{K^*,i} \|_2^2 - \|\bar{\bs}^O_{K,i}\|^2_2 - 2 (\bs_i^E)^{\top} (\bar{\bs}^O_{K^*,i}  - \bar{\bs}^O_{K,i})\right)\\
	=& \frac{4}{n} \sum_{i=1}^n (\bar{\bs}^O_{K^*,i}  - \bar{\bs}^O_{K,i})^\top \Var\left(  \bs_i^E\right)(\bar{\bs}^O_{K^*,i}  - \bar{\bs}^O_{K,i}).
\end{aligned}
$$
Hence, for any $K \ne K^*$, we have $\sigma_{K^*,K}^2 \ge  \frac{4 M_2}{n} \sum_{i=1}^n \|\bar{\bs}^O_{K^*,i}- \bar{\bs}_{K,i}^O \|^2_2$. By the assumptions stated in the theorem, $\min_{K \ne K^*} \max_{i=1,\ldots,n}\|\bar{\bs}^O_{K^*,i}- \bar{\bs}_{K,i}^O \|^2_2 \gtrsim \Delta_{(K^*)}^2$ with asymptotic probability $1$. Furthermore, because the estimated change-points set is nested and the distance between adjoining change-points is at least $cn/\log(n)$, summing this maximum discrepancy over the fraction of points $n/\log(n)$ divided by the total $n$ yields a factor of $1/\log(n)$. Thus, $\sigma_{K^*,K}^2 \gtrsim M_2 \log^{-1}(n)\Delta_{(K^*)}^2$, which strictly implies:
$$
\sigma_{K^*,K} \gtrsim \sqrt{M_2} \Delta_{(K^*)} \frac{1}{\sqrt{\log(n)}},
$$
for sufficiently large $n$. 

Because $\delta_{K^*,K} = -\delta_{K,K^*}$, obtaining an upper bound for the ratio $\delta_{K^*,K} / \sigma_{K^*,K}$ is equivalent to bounding $\delta_{K,K^*}$ from below. Therefore, using our bound on $\sigma_{K^*,K}$, it suffices to show that
\begin{equation}
	\label{equ:thm2goal}
	\min_{K \in \calM\backslash \{K^*\}} {\delta_{K,K^*}} \gtrsim - \sqrt{M_2} x_n \Delta_{(K^*)} \frac{1}{\sqrt{n\log^2(n)}},
\end{equation}
with asymptotic probability $1$. Let $\calS_{x, y}\left(\mathcal{T}_K\right) = \sum_{k=0}^K \sum_{i=\tau_k +1}^{\tau_{k+1}} (x_i - \bar{x}_{\tau_k, \tau_{k+1}})(y_i - \bar{y}_{\tau_k, \tau_{k+1}})$. Taking the expectation $\E_E[\cdot]$ with respect to the even sample (conditional on the odd sample so that $\mathcal{T}_K$ is fixed), algebraic expansion yields:
$$
\begin{aligned}
	n\delta_{K, K^*} =& \E_E\left\{\calS_{s^E}\left(\mathcal{T}_K\right)-\calS_{s^E}\left(\mathcal{T}_{K^*}\right)\right\}-\left\{\calS_{\epsilon^O}\left(\mathcal{T}_K\right)-\calS_{\epsilon^O}\left(\mathcal{T}_{K^*}\right)\right\}-\E_E\left\{\calS_{\epsilon^E}\left(\mathcal{T}_K\right)-\calS_{\epsilon^E}\left(\mathcal{T}_{K^*}\right)\right\} \\
	& +2\E_E\left\{\calS_{\epsilon^O, \epsilon^E}\left(\mathcal{T}_K\right)-\calS_{\epsilon^O, \epsilon^E}\left(\mathcal{T}_{K^*}\right)\right\} \\
	:=&  A_K - B_K - C_K + 2D_K.
\end{aligned}
$$
In the following, we aim to bound $A_K$, $B_K$, $C_K$, and $D_K$ such that \eqref{equ:thm2goal} holds. We consider two cases: lack-of-fit $K \in \calM_l$ (controlled primarily via the minimum jump $\Delta_{(1)}$) and over-fit $K \in \calM_o$ (controlled via the maximum jump $\Delta_{(K^*)}$). 

\textbf{Step 1:} (Lack-of-fit scenario).
First, we focus on $K \in \calM_l$. For the signal term $A_K$, we decompose it as:
$$
\min _{K\in \calM_l} A_K=\min _{K \in \calM_l} \E_E\left\{\calS_{s^E}\left(\mathcal{T}_K\right)-\calS_{s^E}\left(\mathcal{T}^*\right)\right\} - \E_E\left\{\calS_{s^E}\left(\mathcal{T}_{K^*}\right)-\calS_{s^E}\left(\mathcal{T}^*\right)\right\}.
$$
Based on Lemma \ref{tech_lemma2} (i) and Condition \ref{con5} (which ensures sufficient signal strength via the minimum jump size $\Delta_{(1)}$),
$$
\min _{K \in \calM_l}\E_E\left\{\calS_{s^E}\left(\mathcal{T}_K\right)-\calS_{s^E}\left(\mathcal{T}^{*}\right)\right\} \geq \underline{\lambda}M_3\min _{K \in \calM_l} \left(\sum_{k \in \mathcal{I}^*_{lK}} \Delta_k^2\right).
$$
The second term is bounded by $o(M_1 \log\log\bar{\lambda})$ according to Lemma \ref{tech_lemma2} (iii). Hence,
$$
\min _{K \in \calM_l} A_K \ge \underline{\lambda}M_3\min _{K\in \calM_l} \left(\sum_{k \in \mathcal{I}^*_{lK}} \Delta_k^2\right) - o(M_1\log\log\bar{\lambda}).
$$
To lower bound $-B_K$, we bound the maximum absolute deviation of $B_K$ using the triangle inequality over the merged partition $\mathcal{T}_K \cup \calT^*$:
$$
\begin{aligned}
	\max_{K \in \calM_l} |B_K| \leq& \max_{K \in \calM_l} \left|\calS_{\epsilon^O}\left(\mathcal{T}_K\right)-\calS_{\epsilon^O}\left(\mathcal{T}_{K} \cup \calT^*\right)\right| + \max_{K \in \calM_l} \left|\calS_{\epsilon^O}\left(\mathcal{T}^*\right)-\calS_{\epsilon^O}\left(\mathcal{T}_{K} \cup \calT^*\right)\right|   \\
	&+ \left|\calS_{\epsilon^O}\left(\mathcal{T}^*\right)-\calS_{\epsilon^O}\left(\mathcal{T}_{K^*} \cup \calT^* \right)\right|  +  \left|\calS_{\epsilon^O}\left(\mathcal{T}_{K^*}\right)-\calS_{\epsilon^O}\left(\mathcal{T}_{K^*} \cup \calT^*\right)\right|. 
\end{aligned}
$$
It follows from Lemma 18 in \cite{pein2021cross} that $\max_{K \in \calM_l} |B_K| = O_P(M_1 K^* (\log\bar{\lambda})^2)$. For the third term, Lemma \ref{tech_eq1} yields $\max_{K \in \calM_l} |C_K|  = o(M_1\log\log\bar{\lambda})$. For the cross term $D_K$, applying a similar triangle inequality expansion and Lemma 15 in \cite{pein2021cross}, we have $\max_{K \in \calM_l} |D_K| \le o(M_1 \log\log\bar{\lambda}) + O_P(K^* M_1 (\log\bar{\lambda})^2)$. 

Noting that $M_1 \log\log\bar{\lambda} \lesssim  M_1 K^* (\log\bar{\lambda})^2$, we combine the above bounds to obtain:
\begin{equation}
	\label{thm2:e1}
	n\delta_{K, K^*} \gtrsim \underline{\lambda}M_3\min _{K\in \calM_l} \left(\sum_{k \in \mathcal{I}^*_{lK}} \Delta_k^2\right) - M_1 K^* (\log\bar{\lambda})^2,
\end{equation}
with asymptotic probability $1$ as $n \to \infty$. According to Condition \ref{con5}, the signal term strictly dominates the noise term, guaranteeing that $n\delta_{K, K^*} > 0$. Because $\delta_{K,K^*} > 0$ implies $\delta_{K^*,K} < 0$, the ratio $\delta_{K^*,K} / \sigma_{K^*,K}$ is strictly negative. Thus, under-fit models trivially satisfy the approximate-null upper bound requirement of Theorem \ref{thm1}.

\textbf{Step 2:} (Over-fit scenario). Next, we consider $K \in \calM_o$. For the signal term $A_K$,
$$
\begin{aligned}
	\min _{K \in \calM_o}A_K=& \min _{K \in \calM_o}\E_E\left\{\calS_{s^E}\left(\mathcal{T}_K\right)-\calS_{s^E}\left(\mathcal{T}_{K^*}\right)\right\} \\
	\geq & \min _{K \in \calM_o}\E_E\left\{\calS_{s^E}\left(\mathcal{T}_{K}\right)-\calS_{s^E}\left(\mathcal{T}_{K} \cup \mathcal{T}^*\right)\right\} \\
	& -\max _{K \in \calM_o}\E_E\left\{\calS_{s^E}\left(\mathcal{T}^*\right)-\calS_{s^E}\left(\mathcal{T}_{K} \cup \mathcal{T}^*\right)\right\} - \left|\E_E\left\{\calS_{s^E}\left(\mathcal{T}_{K^*}\right)-\calS_{s^E}\left(\mathcal{T}^*\right)\right\}\right|.
\end{aligned}
$$
Based on Lemma \ref{tech_lemma2} (ii) and (iii), the last two terms are both $o(M_1\log\log\bar{\lambda})$. Hence, $\min _{K \in \calM_o} A_K \ge  \min _{K \in \calM_o}\E_E\left\{\calS_{s^E}\left(\mathcal{T}_{K}\right)-\calS_{s^E}\left(\mathcal{T}_{K} \cup \mathcal{T}^*\right)\right\}  - o(M_1 \log\log\bar{\lambda})$.

For the second term $B_K$, applying the decomposition through $\calT_K \cup \calT^*$ and using Lemma 18 in \cite{pein2021cross}, we obtain:
$$
\min_{K\in \calM_o} B_K \ge \min _{K \in \calM_o}\left\{\calS_{\epsilon^O}\left(\mathcal{T}^*\right)-\calS_{\epsilon^O}\left(\mathcal{T}_{K} \cup \mathcal{T}^*\right)\right\} - o_P\left(M_1 \log \log \bar{\lambda}\right).
$$
For the third term $C_K$, a similar decomposition using Lemma \ref{tech_lemma1} gives $\max_{K \in \calM_o} |C_K| = O\left(M_1 \log \log \bar{\lambda}\right)$. 
For the last term $D_K$, according to Lemma 19 in \cite{pein2021cross}, we have
$$
\min_{K \in \calM_o} D_K \ge -o\left(\left|\min _{K \in \calM_o}\left\{\calS_{\epsilon^O}\left(\mathcal{T}^*\right)-\calS_{\epsilon^O}\left(\mathcal{T}_{K} \cup \mathcal{T}^*\right)\right\}\right|\right) - o_P(M_1 \log\log\bar{\lambda}).
$$
Combining these elements for sufficiently large $n$, we have:
\begin{equation}
	\begin{aligned}
		\label{thm2:e2}
		\min_{K \in \calM_o} n\delta_{K, K^*} \gtrsim& \min _{K \in \calM_o}\E_E\left\{\calS_{s^E}\left(\mathcal{T}_{K}\right)-\calS_{s^E}\left(\mathcal{T}_{K} \cup \mathcal{T}^*\right)\right\} \\
		&+\min _{K \in \calM_o}\left\{\calS_{\epsilon^O}\left(\mathcal{T}^*\right)-\calS_{\epsilon^O}\left(\mathcal{T}_{K} \cup \mathcal{T}^*\right)\right\}-M_1 \log \log \bar{\lambda} \\
		\gtrsim& -M_1\log \log \bar{\lambda}.
	\end{aligned}
\end{equation}
The final inequality holds with asymptotic probability $1$ because the terms $\E_E\{\calS_{s^E}(\mathcal{T}_{K}) - \calS_{s^E}(\mathcal{T}_{K} \cup \mathcal{T}^*)\}$ and $\{\calS_{\epsilon^O}(\mathcal{T}^*) - \calS_{\epsilon^O}(\mathcal{T}_{K} \cup \mathcal{T}^*)\}$ represent the deterministic reduction in signal sum of squares and the realized noise-part sum of squares improvement from partition refinement, respectively, both of which can be bounded below by zero. Dividing by $n$, we obtain $\min_{K \in \calM_o} \delta_{K,K^*}\gtrsim -M_1\log \log \bar{\lambda} / n$. 

To satisfy condition \eqref{equ:thm2goal}, we require:
$$
\frac{M_1\log \log \bar{\lambda}}{n} \le \sqrt{M_2} x_n  \Delta_{(K^*)} \frac{1}{\sqrt{n\log^2(n)}},
$$
which simplifies to the requirement that $\Delta_{(K^*)} \gtrsim x_n^{-1} \log\log(\bar{\lambda}) \frac{\log(n)}{\sqrt{n}}$. Squaring this yields the necessary lower bound on the maximum jump size given in the theorem statement, implying $\Pr\left\{K^* \in \calA \right\} \ge 1- \alpha + o(1)$.

\subsection{Proof of Theorem \ref{thm3}}

Let $\calB_{1n}$ and $\calB_{2n}$ be defined as in Definition \ref{con6}. First, we establish an upper bound for the standard deviation $\sigma_{K, K^*}$. Recalling the variance expansion, we have:
$$
\sigma^2_{K, K^*} \lesssim  \frac{1}{n} \sum_{i=1}^n  \| \bar{\bs}^O_{K,i}- \bar{\bs}^O_{K^*,i}\|_2^2.
$$
Under the theorem assumption that $\max_{i=1,\ldots, n}\| \bar{\bs}^O_{K,i}- \bar{\bs}^O_{K^*,i}\|_2^2 \lesssim \Delta_{(K^*)}^2$ for all $K$, this yields $\sigma^2_{K, K^*} \lesssim \Delta_{(K^*)}^2$. Therefore, taking the square root gives $\sigma_{K, K*} \lesssim \Delta_{(K^*)}$ with asymptotic probability $1$.

Next, we establish the behavior of $\delta_{K,K^*}$ for models in $\calB_{1n}$ and $\calB_{2n}$ separately:

\textbf{Case 1: $K \in \calB_{1n} \subset \calM_l$.} 
From the lack-of-fit bound established in equation \eqref{thm2:e1}, we have:
$$
n\delta_{K, K^*} \gtrsim \underline{\lambda}M_3 \left(\sum_{k \in \mathcal{I}^*_{lK}} \Delta_k^2\right) - M_1 K^* (\log\bar{\lambda})^2.
$$
For $K \in \calB_{1n}$, the definition of the set implies that the signal term $\underline{\lambda} \sum_{k \in \calI^*_{lK}} \Delta_k^2$ dominates the $M_1 K^* (\log\bar{\lambda})^2$ noise remainder, yielding $n\delta_{K,K^*} \gtrsim M_1 \Delta_{(K^*)}^2 \sqrt{n\log(n)}$. Dividing by $n$, we obtain $\delta_{K,K^*} \gtrsim M_1 \Delta_{(K^*)}^2 \sqrt{\frac{\log(n)}{n}}$.

\textbf{Case 2: $K \in \calB_{2n} \subset \calM_o$.}
From the over-fit bound established in equation \eqref{thm2:e2}, we have:
$$
n\delta_{K, K^*} \gtrsim \E_E\left\{\calS_{s^E}\left(\mathcal{T}_{K}\right)-\calS_{s^E}\left(\mathcal{T}_{K} \cup \mathcal{T}^*\right)\right\} +\left\{\calS_{\boldsymbol\epsilon^O}\left(\mathcal{T}^*\right)-\calS_{\boldsymbol\epsilon^O}\left(\mathcal{T}_{K} \cup \mathcal{T}^*\right)\right\} - M_1 \log \log \bar{\lambda}.
$$
For $K \in \calB_{2n}$, the definition of the set ensures that the first two terms exceed the stochastic remainder by a sufficient margin, yielding $n\delta_{K,K^*} \gtrsim M_1 \Delta_{(K^*)}^2 \sqrt{n\log(n)}$. Dividing by $n$, we again obtain $\delta_{K,K^*} \gtrsim M_1 \Delta_{(K^*)}^2 \sqrt{\frac{\log(n)}{n}}$.

Therefore, for any $K \in \calB_{1n} \cup \calB_{2n}$, combining the lower bound on $\delta_{K,K^*}$ with the upper bound on $\sigma_{K,K^*}$ yields:
$$
\max_{J \ne K} \frac{\delta_{K,J}}{\sigma_{K,J}} \ge \frac{\delta_{K,K^*}}{\sigma_{K,K^*}} \gtrsim \frac{M_1 \Delta_{(K^*)}^2 \sqrt{\frac{\log(n)}{n}}}{\Delta_{(K^*)}} = M_1 \Delta_{(K^*)} \sqrt{\frac{\log(n)}{n}}.
$$
Because $\Delta_{(K^*)} \sqrt{\frac{\log(n)}{n}} \gg \frac{x_n}{\sqrt{n \log(n)}}$ for any $x_n = o(1)$, the ratio strictly violates the threshold required for inclusion in $\calA$. Consequently, every model $K \in \calB_{1n} \cup \calB_{2n}$ is excluded from $\calA$ with probability tending to $1$. It immediately follows that at least $|\calB_{1n}| + |\calB_{2n}|$ candidate models are absent from the selected set, establishing $\Pr\{|\calA| \le K_{\max} - |\calB_{1n} | - |\calB_{2n} |\} \ge 1 - o(1)$.
% Let $\calB_1 := \left\{ K \in \calM_{l}: \sum_{k \in \calI^*_{lK}} \Delta^2_k \ge M_1 \Delta_{(K^*)}^2\sqrt{n\log(n)}/\underline{\lambda}\right\}$, and \\
% {\footnotesize $\calB_2 := \left\{K \in \calM_{o}:   \E\left\{\calS_{s^E}\left(\mathcal{T}_{K}\right)-\calS_{s^E}\left(\mathcal{T}_{K} \cup \mathcal{T}^*\right)\right\}  + \left\{\calS_{\varepsilon^O}\left(\mathcal{T}^*\right)-\calS_{\varepsilon}^O\left(\mathcal{T}_{K} \cup \mathcal{T}^*\right)\right\}  \ge M_1 \Delta_{(K^*)}^2 \sqrt{n \log(n)}\right\}$}. \\ By inequality \eqref{equ:upperxi} and $\max_{i=1,\ldots, n}\| \bar{\bs}^O_{K,i}- \bar{\bs}^O_{J,i}\|_2^2 \lesssim \Delta_{(1)}^2$ for all pair $K \ne J$ with probability goes to 1,
% $$
% \sigma_{K, K^*} \lesssim  n^{-1} \sum_{i=1}^n  \| \bar{\bs}^O_{K,i}- \bar{\bs}^O_{J,i}\|_2^2 +  \Var\left\{(\bs_i^E)^\top \left(\bar{\bs}^O_{K,i} - \bar{\bs}^O_{J,i} \right)\right\} \lesssim M_1 \Delta_{(1)}^2,
% $$
% with {\blue asymptotic probability $1$}. According to equation \eqref{thm2:e1} and \eqref{thm2:e2} and direct calculations, we can show that for $K \in \calB_1 \cup \calB_2$
% $$
% \delta_{K,K^*} \ge  M_1 \Delta_{(1)}^2 \sqrt{\frac{\log(n)}{n}}.
% $$
% Therefore,
% $$
% \begin{aligned}
% 	\max_{J \ne K} \frac{\delta_{K,J}}{\sigma_{K,J}} \ge \frac{\delta_{K,K^*}}{\sigma_{K,K^*}}\ge \frac{\delta_{K,K^*}}{M_1 \Delta_{(1)}^2} \ge c \sqrt{\frac{\log(n)}{n}}.
% \end{aligned}
% $$
% Hence, $\Pr\{|\calA| \le K_{\max} - |\calB_1 | - |\calB_2 |\} \ge 1 - o(1)$  based on Theorem \ref{thm1} (ii).

\section{Additional related literature}
\label{sup:extra_lit}
In the field of change-point detection literature, while there is no shortage of research on statistical inference related to the number of change-points, the majority, if not all, of these works have focused on controlling either the Familywise Error Rate (FWER) or the False Discovery Rate (FDR). \cite{frick2014multiscale} first introduced a novel method called SMUCE to control FWER for data from one-dimensional exponential family.  \cite{pein2017heterogeneous} extended this approach to the heterogeneous change-point detection problem. The FWER in change-point detection context is typically defined as the probability that the estimated number of change-points $\hK$ strictly exceeds the true number $K^*$, i.e., $\Pr\{\hK > K^*\}$. By specifying a small predetermined level $\alpha$ for FWER, one can have certain confidence that $\hK$ does not overestimate the true number of change-points.

On the flip side, however, the SMUCE-type procedures can be overly stringent, resulting in low statistical power and a tendency to significantly underestimate the true number of change-points in practical applications \citep{li2016fdr, chen2021data}. Instead of FWER, \cite{li2016fdr} suggested controlling a less stringent criterion FDR. FDR is defined as the proportion of false discoveries among the selected change-points, i.e., $\E[\hK^F/\hK]$, where $\hK^F$ represents the number of false discoveries. See \cite{li2016fdr} for a more comprehensive definition of FDR in change-point problems.  In the realm of FDR-based methods, \cite{hao2013multiple} introduced a screening and ranking algorithm (SaRa) for detecting one-dimensional normal mean change-points. \cite{li2016fdr} proposed a multiscale change-point segmentation approach (FDRseg) based on the same model setup. \cite{cheng2020multiple} advocated a differential smoothing and testing of maxima/minima algorithm (dSTEM) for continuous time series. \cite{chen2021data} developed a mirror with order-preserved splitting procedure called MOPS,  to address a broader range of change-point models, including structural changes and variance changes. \cite{liu2022generalized} extended the knockoff framework to control FDR in structural change detection, which in turn can be modified for change-point detection. \cite{sun2025synthetic} provided a data-splitting approach for FDR control. While FDR control may provide greater power than FWER, it does not guarantee either consistent estimation of true number of change-points, or the finite-sample confidence of recovering the true number. At its core, FDR only concerns the expectation of false discoveries rather than probability. For instance, consider the scenario of overfitting. Note that $\hK = K^* + \hK^F$, and FDR control aims to ensure that $\E[\hK^F/\hK]\leq\alpha$. Under some mild conditions, it implies that $\E[\hK^F]\leq(\alpha/(1-\alpha)) K^*$ approximately, which further means that $\E[\hK^F]$ can increase as the true number $K^*$ increases. Consequently, if $K^*$ is sufficiently large, $\E[\hK^F]$ may also become substantial. As a result, the point estimate $\hK$ may deviate significantly from the true value in such cases.

The concept of confidence set studied in this paper is fundamentally distinct with FWER and FDR, as it revolves around the probability of recovering the true number of change-points, and the proposed confidence does not rely on the true number $K^*$. Furthermore, letting $\alpha$ approaches zero, a well-designed detection algorithm should lead to the cardinality of the constructed confidence set equal to one, thus further result in consistency, highlighting the robustness and reliability of the algorithm in accurately identifying the true number of change-points.

{
\section{Choice of loss function and discussion of related conditions}
\label{sup:loss}
Table 1 in \cite{zou2020consistent} have discussed several important settings. Here, we discuss two extra important settings here, as detailed below.
\begin{itemize}
\item \textbf{Covariance change-points model.}
Consider the setting $\bz_i \in \R^p \sim (\mathbf{0}, \bSigma_k^*)$ where $\tau_{k-1}^* < i \le \tau_k^*$ for $k = 1, \ldots, K^*+1$ and $i = 1, \ldots, 2n$. For this setup, we can choose the loss function $l(\bbeta; \bz_i) = \|\bz_i \bz_i^\top - \bbeta\|_F^2$ for $\bbeta \in \R^{p \times p}$. The gradient of the loss function is $\partial l(\bbeta; \bz_i)/\partial \bbeta = \bz_i \bz_i^\top$, and with $\bbeta = \bgamma = 0$, we define the score $\bs_i = \operatorname{vech}(\partial l(\bbeta; \bz_i)/\partial \bbeta) = \operatorname{vech}(\bz_i \bz_i^\top)$. The method from \cite{aue2009break} can then be applied to detect change-points in this setting.

    \item \textbf{Network change-points model.}
In the network change-point setting, we assume $\bz_i \in \R^{p \times p} \sim \operatorname{Bern}(\boldsymbol{\Theta}^*_k)$, where $\tau_{k-1}^* < i \le \tau_k^*$ for $k = 1, \ldots, K^*+1$ and $i = 1, \ldots, 2n$. Here, $\bz \sim \operatorname{Bern}(\boldsymbol{\Theta})$ means that $\bz_{i,j} \stackrel{i.i.d.}{\sim} \operatorname{Bern}(\boldsymbol{\theta}_{i,j})$. For this model, we may choose the loss function $l(\bbeta; \bz_i) = \|\bz_i - \bbeta\|_F^2$ for $\bbeta \in \R^{p \times p}$. The gradient is $\partial l(\bbeta; \bz_i)/\partial \bbeta = \bz_i$, and with $\bbeta = \bgamma = 0$, the score becomes $\bs_i = \operatorname{vech}(\partial l(\bbeta; \bz_i)/\partial \bbeta) = \operatorname{vech}(\bz_i)$. The method from \cite{wang2021optimal} can be used to detect the change-points in this context.
\end{itemize}
Let $\boldsymbol{\eta}_i \in \R^p \overset{i.i.d.}{\sim} (0, \mathbf{I})$, and denote it as $\eta_i$ when it reduces to the scalar case (i.e., $p=1$). Model (2.1) can be specified into the following different models. Let $\boldsymbol{\Sigma}$ be a fixed covariance matrix. For the network model, $\mathbf{Z}_i = (z_{i,j,l})_{1 \le j,l \le n} \sim \operatorname{Bern}(\boldsymbol{\Theta}_i)$ represents $z_{i,j,l} \overset{i.i.d.}{\sim} \operatorname{Bern}(\beta_{i,j,l})$, and let $\mathbf{W}_i = \mathbf{Z}_i - \boldsymbol{\Theta}_i$.

\begin{table}[h]
	\caption{Detailed choice of score functions for OPTICS under different models.\label{tab:condition}}
	\centering
	\begin{tabular}{cccc}
		\hline
		Name & Formula  &$l(\bbeta; \bz_i)$     & $\bs_i$         \\ \hline
           Mean  &$\bz_i = \bbeta^*_k + \boldsymbol{\Sigma}^{1/2} \boldsymbol{\eta}_i$ &$\|\bz_i -\bbeta \|_2^2$ & $\bs_i = \bz_i$  \\
           Variance &$z_i = \beta^*_k \eta_i$ &$\|\log(z^2_i) - \beta \|_2^2$  &$s_i = 2\log(z_i)$  \\
           Regression &$z_i = \bx_i^\top \bbeta^*_k + \eta_i$ &$\|z_i -\bbeta^\top \bx_i \|_2^2$ & $\bs_i = z_i(1, \bx_i^\top)^\top$ \\
           Covariance &$\bz_i = (\boldsymbol{\Theta}^*_k)^{1/2} \boldsymbol{\eta}_i, \bbeta^*_k \in \R^{p\times p}$ &$\|\bz_i\bz^\top_i -\boldsymbol{\Theta} \|_F^2$  & $\bs_i = \operatorname{vech}(\bz_i \bz^\top_i)$  \\
           Network &$\mathbf{Z}_i =  \boldsymbol{\Theta}^*_k+\mathbf{W}_i, \boldsymbol{\Theta}^*_k \in \R^{p\times p}$ &$\|\mathbf{Z}_i - \boldsymbol{\Theta}\|_F^2$ & $\bs_i = \operatorname{vech}(\mathbf{Z}_i)$  \\
		\hline
	\end{tabular}
\end{table}

Based on this table, we can clearly identify the sufficient conditions for $\bz_i$ to satisfy Condition 3.1 in the main content. For instance, in the case of the multiple mean change-point model, the original data $\bz_i$ should be a sub-Gaussian vector. Similarly, for multiple variance change-point models, it is required that a transformation of $z_i$, specifically $\log(z_i)$, is sub-Gaussian.
}

\section{More simulation results}

\label{sup:simu}
This section devotes to more simulation results. Table \ref{tab:1} and \ref{tab:3} report the coverage rates when $d=1$ and $d=5$, respectively, when the error term follows the standard normal distribution in Section \ref{subsec:4.1}.  Table \ref{tab:6} provides the coverage rates in Section \ref{subsec:4.2}, with standard normal errors. 

\begin{table}[h]
	\caption{The coverage rates in the mean-change model: $d=1$; normal error.\label{tab:1}}
	\centering
	\begin{tabular}{cccccc}
		\hline
		Amplitude  $A$    & 0.50       & 0.625      & 0.75       & 0.875      & 1.00          \\ \hline
		OPTICS(BS)  & 0.77(4.10) & 0.84(3.75) & 0.87(3.10) & 0.84(2.66) & 0.84(2.68) \\
		OPTICS(SN)  &0.82(4.00) & 0.95(3.42) & 0.98(2.56) & 0.99(2.39) & 1.00(2.43) \\
		$\calA^1_{COPSS}$(BS) & 0.45(3.00)       & 0.74(3.00)       & 0.85(3.00)       & 0.89(3.00)       & 0.97(3.00)       \\
		$\calA^1_{COPSS}$(SN) & 0.57(3.00)      & 0.91(3.00)       & 0.92(3.00)       & 0.97(3.00)       & 0.98(3.00)      \\
		$\calA^1_{FDRseg}$ & 0.83(3.00)      & 0.99(3.00)      & 0.98(3.00)       & 0.96(3.00)      & 0.96(3.00)      \\
		$\calA^1_{SMUCE}$ & 0.24(3.00)      & 0.86(3.00)      & 0.98(3.00)       & 1.00(3.00)      & 1.00(3.00)      \\
		COPSS(BS) & 0.23       & 0.38       & 0.52       & 0.60       & 0.52       \\
		COPSS(SN) & 0.30       & 0.61       & 0.84       & 0.89       & 0.89       \\
		FDRseg    & 0.53       & 0.85       & 0.89       & 0.82       & 0.82       \\
		SMUCE     & 0.01       & 0.36       & 0.92       & 1          & 1          \\
		\hline
	\end{tabular}
\end{table}

\begin{table}[h]
	\caption{The coverage rates in the mean-change model: $d=5$; normal error.\label{tab:3}}
	\centering
	\begin{tabular}{cccccc}
		\hline
		Amplitude  $A$       & 0.50       & 0.625      & 0.75       & 0.875      & 1.00          \\ \hline
		OPTICS(BS)  & 0.73(3.27) & 0.55(2.34) & 0.60(2.19) & 0.56(2.04) & 0.67(2.20) \\
		OPTICS(SN)  & 0.83(3.17) & 0.83(2.04) & 0.94(1.91) & 0.92(1.85) & 0.93(2.04) \\
		$\calA^1_{COPSS}$(BS) & 0.65(3.00)       & 0.70(3.00)       & 0.76(3.00)       & 0.80(3.00)       & 0.88(3.00)       \\
		$\calA^1_{COPSS}$(SN) & 0.77(3.00)      & 0.83(3.00)       & 0.96(3.00)       & 0.93(3.00)       & 0.95(3.00)      \\
		COPSS(BS) & 0.35       & 0.33       & 0.36       & 0.42       & 0.43       \\
		COPSS(SN) & 0.48       & 0.75       & 0.80       & 0.80       & 0.69       \\
		\hline
	\end{tabular}
\end{table}

\begin{table}[h]
	\caption{The coverage rates in linear model with coefficient structural-breaks; $N(0,1)$ errors.\label{tab:6}}
	\centering
	\begin{tabular}{cccccc}
		\hline
				Amplitude $A$   & 0.10       & 0.125      & 0.15       & 0.175      & 0.20        \\ \hline
		OPTICS(BS)  & 0.83(3.52) & 0.90(3.13) & 0.89(2.68) & 0.80(2.23) & 0.78(1.97) \\
		OPTICS(SN)  & 0.92(2.76) & 0.97(2.28) & 0.98(1.90) & 0.97(1.61) & 1.00(1.51) \\
		$\calA^1_{COPSS}$(BS) & 0.50(3.00) & 0.59(3.00) & 0.60(3.00) & 0.59(3.00) & 0.61(3.00) \\
		$\calA^1_{COPSS}$(SN) & 0.72(3.00) & 0.76(3.00) & 0.72(3.00) & 0.74(3.00) & 0.77(3.00) \\
		COPSS(BS) & 0.22       & 0.31       & 0.26       & 0.33       & 0.32       \\
		COPSS(SN) & 0.51       & 0.57       & 0.56       & 0.62       & 0.58       \\
		\hline
	\end{tabular}
\end{table}

\newpage

\subsection{Variance change-point model}

In this subsection, we consider the variance change-point model  
$$
y_i = \sigma_k \epsilon^*_i, \ \tau^*_{k-1} < i \le \tau_k^*, \ k=1,\ldots, K^*+1, \  i =1,\ldots, 2n,
$$
where $\sigma^*_{k+1}/ \sigma^*_{k} = A^{(-1)^{k-1}}, \ k=1,\ldots, K^*$ and $\sigma_1=1$. The error terms $\epsilon_i \sim N(0,0.25)$ following \cite{chen1997testing}. In this model, we set the sample size $n=1000$,  the true change-point set $\calT^* = \{200k, k = 1, \ldots, 4\}$ and the change amplitude $A = \{2, 3, 4, 5, 6 \}$. 
\begin{table}[h]
	\caption{The coverage rates in the variance change-point model\label{tab:5}.}
	\centering
	\begin{tabular}{cccccc}
		\hline
		Amplitude $A$   & 2       & 3      & 4       & 5      & 6       \\ \hline
		OPTICS(BS)  & 0.86(3.99) & 0.85(2.73) & 0.93(2.63) & 0.88(2.71) & 0.88(2.71) \\
		OPTICS(SN)  & 0.83(4.55) & 0.98(2.84) & 0.99(2.66) & 1.00(2.83) & 0.99(3.02) \\
		$\calA^1_{COPSS}$(BS) & 0.51(3.00)       & 0.77(3.00)       & 0.88(3.00)       & 0.87(3.00)       & 0.94(3.00)       \\
		$\calA^1_{COPSS}$(SN) & 0.53(3.00)      & 0.83(3.00)       & 0.94(3.00)       & 0.96(3.00)       & 0.97(3.00)      \\
		$\calA^1_{FDRseg}$ & 0(3.00)      & 0(3.00)      & 0(3.00)       & 0(3.00)      & 0.01(3.00)      \\
		$\calA^1_{SMUCE}$ & 0.24(3.00)      & 0.13(3.00)      & 0.13(3.00)       & 0.15(3.00)      & 0.19(3.00)      \\
		COPSS(BS) & 0.28       & 0.53       & 0.63       & 0.60       & 0.55       \\
		COPSS(SN) & 0.18       & 0.74       & 0.87       & 0.89       & 0.97      \\ 
		FDRseg &0& 0& 0& 0& 0 \\
		SMUCE &0.11& 0.06& 0.06& 0.04& 0.04  \\
		\hline
	\end{tabular}
\end{table}

Table \ref{tab:5} compares the coverage rates of OPTICS with the state-to-art methods. Similar as in the multiple mean-change model, OPTICS with SN and BS achieve the nominal confidence level. Their coverage rates consistently surpass those of quasi-confidence sets created from COPSS. This implies OPTICS is desirable from both empirical and theoretical standpoints. Additionally, as expected, FDRseg and SMUCE, along with their respective quasi-confidence sets, lack power as they are tailored for detecting mean changes, rendering them insensitive to variance changes.

{
\subsection{Network change-points model}
Consider the network change-points model
$$
\mathbf{Z}_i = \boldsymbol{\Theta}^*_k + \mathbf{W}_i, \ \tau_{k-1}^*  < i  \le \tau^*_{k}, \ k = 1, \ldots, K^{*}+1, \ i=1,\ldots,2n,
$$ 
where $\tau_k^*, \ k = 1, \ldots, K^*$ are the true network change-points, $\boldsymbol{\Theta}^*_k$ is the $d \times d$-dimensional network mean matrix for subject $i$ when $\tau_{k-1}^*  < i  \le \tau^*_{k}$; $\mathbf{W}_i$ is the independently distributed Bernoulli error. The sample size is taken to be $2n = 1000$, and the set of true change-points $\calT^* = \{\tau_k^* = 200k, k =1, \ldots, 4\}$, hence $K^*=4$. The $k$th mean vector $\boldsymbol{\Theta}^*_k$ is generated from stochastic block model \citep{wang2021optimal} with connectivity matrices $\mathbf{Q}_k = \mathbf{Q}_l$ with $l = \mod(k/2)+1$, and
$$
\mathbf{Q}_1 = A \times \begin{bmatrix}
    0.6, &1,&0.6\\
    1, &0.6, &0.5\\
    0.6, &0.5,&0.6\\
\end{bmatrix} \text{ and } \mathbf{Q}_2 = A \times \begin{bmatrix}
    0.6, &0.5,&0.6\\
    0.5, &0.6,&1\\
    0.6, &1,&0.6\\
\end{bmatrix}
$$
where $A$ is a scaler, varying among $\{0.50, 0.60, 0.70, 0.80, 0.90\}$. Each network is generated from a $3$-community stochastic block model and node size $d=5$. At the change points, membership are reshuffled randomly. This simulation setting mimics the situation in \cite{wang2021optimal}, and we choose the Network Binary Segmentation (NBS) proposed therein as our change-point detection procedure.

We compare OPTICS with COPSS and its naive confidence set $\calA^1_{COPSS}$. Table \ref{tab:sup:3} report the coverage rates. The quantities in parentheses are average cardinalities of estimated sets. We take $q=1$ for the quasi-confidence sets for COPSS, i.e., $\calA^{1} = \{\hK-1, \hK, \hK+1\}$ with cardinality $3$, and denote the generated sets by $\calA^1_{COPSS}$. We can see that OPTICS is the only method closely reaches the specified confidence level.

\begin{table}[h]
	\caption{The coverage rates in the network change-points model.\label{tab:sup:3}}
	\centering
	\begin{tabular}{ccccccc}
		\hline
		Amplitude  $A$        & 0.50       & 0.60       & 0.70       & 0.80       & 0.90       & 1.00       \\ \hline
		OPTICS(BS)  & 0.79(2.88) & 0.75(2.57) & 0.79(2.59) & 0.70(1.83) & 0.65(1.79) & 0.75(1.85) \\
		$\calA^1_{COPSS}$(BS) & 0.75(3.00) & 0.84(3.00) & 0.89(3.00) & 0.91(3.00) & 0.86(3.00) & 0.87(3.00) \\
		COPSS(BS) & 0.37       & 0.39       & 0.47       & 0.55       & 0.55       & 0.62       \\
		\hline
	\end{tabular}
\end{table}

\subsection{Covariance change-point model}
Consider the covariance change-points model 
%$\bz_i = (\boldsymbol{\Theta}^*_k)^{1/2} \boldsymbol{\eta}_i, \bbeta^*_k \in \R^{p\times p}$
$$
\mathbf{z}_i = (\boldsymbol{\Theta}^*_k)^{1/2} \boldsymbol{\eta}_i, \tau_{k-1}^*  < i  \le \tau^*_{k}, \ k = 1, \ldots, K^{*}+1, \ i=1,\ldots,2n,
$$ 
where $\boldsymbol{\Theta}^*_k$ is the $d \times d$-dimensional covariance matrix for subject $i$ when $\tau_{k-1}^*  < i  \le \tau^*_{k}$; $\boldsymbol{\eta}_i$ is the independently and identically distributed standard Gaussian vector. The sample size is taken to be $2n = 1000$, and the set of true change-points $\calT^* = \{\tau_k^* = 200k, k =1, \ldots, 4\}$, hence $K^*=4$. The $k$th mean vector $\boldsymbol{\Theta}^*_k = \boldsymbol{\Theta}^*_l$ with $l=\mod(k/2)+1$, and
$$
\boldsymbol{\Theta}^*_1 = \mathbf{I}_d \text{ and }  \boldsymbol{\Theta}^*_1 = \{A^{|i-j|}\}_{1 \le i, j \le n},
$$
where $A$ is a scaler, varying among $\{0.10, 0.20, 0.30, 0.40, 0.50\}$. We choose the Wild Binary Segmentation (WBS) proposed in \cite{fryzlewicz2014wild} as our change-point detection procedure.

We compare OPTICS with COPSS and its naive confidence set $\calA^1_{COPSS}$. Table \ref{tab:sup:4} reports the coverage rates, with the quantities in parentheses representing the average cardinalities of the estimated sets. From the results, we observe that OPTICS performs best among all methods, although it still does not achieve the specified confidence level. This may be due to the Frobenius measure not being an ideal choice. We leave this issue for future research.

\begin{table}[h]
	\caption{The coverage rates in the network change-points model.\label{tab:sup:4}}
	\centering
	\begin{tabular}{cccccc}
		\hline
		Amplitude  $A$        & 0.30       & 0.35       & 0.40       & 0.45       & 0.50       \\ \hline
		OPTICS(BS)  & 0.57(4.52) & 0.63(5.17) & 0.70(5.33) & 0.68(5.18) & 0.64(4.29) \\
		$\calA^1_{COPSS}$(BS) & 0.42(3.00) & 0.44(3.00) & 0.22(3.00) & 0.17(3.00) & 0.29(3.00) \\
		COPSS(BS) & 0.09       & 0.05       & 0.06       & 0.03       & 0.08       \\
		\hline
	\end{tabular}
\end{table}

\subsection{Multiple mean-change with heavy-tail error}
\label{app:simu:heavy}
The data-generating process for the heavy-tailed case is nearly identical to the multiple mean-change model with a $t$-distribution described in Section \ref{subsec:4.1}, with the only difference being that we set the degrees of freedom to $df=1$ in the $t$-distribution. From the Table~\ref{tab:coverage_ro}, we observe that Huber-OPTICS (H-OPTICS) with $\kappa=1.5$ and OPTICS consistently achieve higher coverage rates compared to existing methods. Furthermore, H-OPTICS produces narrower confidence sets than OPTICS, demonstrating greater power under heavy-tailed data.

\begin{table}[h]
    \caption{Coverage rates in the mean-change model ($d=1$) with heavy-tail errors.\label{tab:coverage_ro}}
    \centering
    \begin{tabular}{cccccc}
        \hline
        Amplitude $A$      & 0.50       & 0.625      & 0.75       & 0.875      & 1.00          \\ \hline
        H-OPTICS(BS)     & 0.58(3.01) & 0.57(3.07) & 0.54(3.13) & 0.58(3.27) & 0.63(3.00) \\
        H-OPTICS(SN)     & 0.97(2.99) & 0.92(2.95) & 0.93(3.03) & 0.97(3.28) & 0.95(3.05) \\
        OPTICS(BS)         & 0.82(4.21) & 0.73(4.23) & 0.74(4.15) & 0.79(4.25) & 0.80(4.09) \\
        OPTICS(SN)         & 1.00(4.32) & 1.00(4.32) & 1.00(4.54) & 0.99(4.31) & 1.00(4.32) \\
        $\calA^1_{COPSS}$(BS) & 0.29(3.00) & 0.34(3.00) & 0.34(3.00) & 0.31(3.00) & 0.36(3.00) \\
        $\calA^1_{COPSS}$(SN) & 0.60(3.00) & 0.51(3.00) & 0.52(3.00) & 0.45(3.00) & 0.45(3.00) \\
        COPSS(BS)          & 0.02       & 0.10       & 0.04       & 0.05       & 0.05       \\
        COPSS(SN)          & 0.04       & 0.01       & 0.03       & 0.08       & 0.00       \\
        \hline
    \end{tabular}
\end{table}

{
\subsection{Multiple mean-change with $m$-dependent errors}
\label{app:simu:dependent}

The data-generating process for this setting is nearly identical to the multiple mean-change model with a normal distribution described in Section \ref{subsec:4.1}, except that the error terms $\beps_t$ are $m$-dependent. Specifically, they are defined as 
$$
\beps_i = \sum_{l=1}^{m} \phi_l \boldsymbol{\eta}_{t+l}, \quad i = 1, \ldots, 2n, \quad \text{with } \boldsymbol{\eta}_t \sim N(0,\mathbf{I}_d), \, t = 1, \ldots, 2n+m,
$$
and $\phi_l = \sqrt{1/M}$. It follows that the $\beps_i$'s exhibit an $m$-dependent structure. We evaluate the performance of the methods under both a single change-point setting ($d=1$) and a multi-dimensional mean-change setting ($d=5$).

From Table~\ref{tab:m_optics} ($d=1$), we observe that $m$-dependent OPTICS (M-OPTICS) achieves significantly higher coverage rates compared to standard OPTICS and existing competing methods, which struggle heavily with the dependent errors. M-OPTICS(SN), in particular, maintains near-perfect coverage. This phenomenon demonstrates that M-OPTICS can effectively adapt to $m$-dependent data while maintaining the pre-specified coverage rate, albeit with slightly wider confidence sets.

Table~\ref{tab:m_optics_d5} extends this analysis to the more complex multi-dimensional case ($d=5$). As the number of true change-points increases, the estimation problem becomes substantially harder. Consequently, standard OPTICS and all competing baseline methods fail almost completely, yielding coverage rates near zero. While the coverage of M-OPTICS(BS) decreases in this challenging setting, M-OPTICS(SN) exhibits remarkable robustness. By combining the multiple-splitting procedure with self-normalization, M-OPTICS(SN) successfully preserves high coverage rates (ranging from 0.88 to 0.98) across all tested amplitudes, highlighting its distinct advantage in handling multi-dimensional, dependent data structures.

\begin{table}[h]
  \caption{Coverage rates in the mean–change model ($d=1$) with $m$–dependent errors.\label{tab:m_optics}}
  \centering
  \begin{tabular}{cccccc}
    \hline
    Amplitude $A$                & 0.50       & 0.625      & 0.75       & 0.875      & 1.00       \\ \hline
    M-OPTICS(BS)                 & 0.87(4.11) & 0.79(2.58) & 0.85(2.51) & 0.81(2.75) & 0.78(2.55) \\
    M-OPTICS(SN)                 & 0.93(3.23) & 0.96(2.27) & 1.00(1.99) & 1.00(2.23) & 1.00(2.27) \\
    OPTICS(BS)                   & 0.22(2.61) & 0.25(2.72) & 0.29(2.71) & 0.23(2.57) & 0.21(2.59) \\
    OPTICS(SN)                   & 0.41(2.92) & 0.46(3.08) & 0.39(2.91) & 0.31(2.74) & 0.43(3.00) \\
    $\calA^1_{\text{COPSS}}$(BS) & 0.05(3.00) & 0.08(3.00) & 0.06(3.00) & 0.06(3.00) & 0.02(3.00) \\
    $\calA^1_{\text{COPSS}}$(SN) & 0.04(3.00) & 0.06(3.00) & 0.05(3.00) & 0.05(3.00) & 0.06(3.00) \\
    $\calA^1_{\text{FDRseg}}$    & 0.00(3.00) & 0.00(3.00) & 0.00(3.00) & 0.00(3.00) & 0.00(3.00) \\
    $\calA^1_{\text{SMUCE}}$     & 0.00(3.00) & 0.00(3.00) & 0.00(3.00) & 0.00(3.00) & 0.00(3.00) \\
    COPSS(BS)                    & 0.00       & 0.05       & 0.03       & 0.02       & 0.01       \\
    COPSS(SN)                    & 0.01       & 0.01       & 0.01       & 0.03       & 0.04       \\
    FDRseg                       & 0.00       & 0.00       & 0.00       & 0.00       & 0.00       \\
    SMUCE                        & 0.00       & 0.00       & 0.00       & 0.00       & 0.00       \\
    \hline
  \end{tabular}
\end{table}

\begin{table}[h]
  \caption{Coverage rates in the mean–change model ($d=5$) with $m$–dependent errors.\label{tab:m_optics_d5}}
  \centering
  \begin{tabular}{cccccc}
    \hline
    Amplitude $A$                & 0.50       & 0.625      & 0.75       & 0.875      & 1.00       \\ \hline
    M-OPTICS(BS)                 & 0.48(2.16) & 0.48(1.97) & 0.49(1.93) & 0.53(2.06) & 0.66(2.12) \\
    M-OPTICS(SN)                 & 0.88(2.17) & 0.94(1.93) & 0.94(1.92) & 0.97(1.94) & 0.98(1.93) \\
    OPTICS(BS)                   & 0.05(1.88) & 0.04(1.89) & 0.11(2.03) & 0.10(2.03) & 0.07(1.90) \\
    OPTICS(SN)                   & 0.16(2.32) & 0.15(2.27) & 0.20(2.32) & 0.15(2.22) & 0.11(2.22) \\
    $\calA^1_{\text{COPSS}}$(BS) & 0.00       & 0.00       & 0.01       & 0.01       & 0.00       \\
    $\calA^1_{\text{COPSS}}$(SN) & 0.01       & 0.01       & 0.03       & 0.02       & 0.00       \\
    $\calA^1_{\text{FDRseg}}$    & 0.00(3.00) & 0.00(3.00) & 0.00(3.00) & 0.00(3.00) & 0.00(3.00) \\
    $\calA^1_{\text{SMUCE}}$     & 0.00(3.00) & 0.00(3.00) & 0.00(3.00) & 0.00(3.00) & 0.00(3.00) \\
    COPSS(BS)                    & 0.00       & 0.00       & 0.00       & 0.00       & 0.00       \\
    COPSS(SN)                    & 0.00       & 0.00       & 0.01       & 0.00       & 0.00       \\
    \hline
  \end{tabular}
\end{table}
}

}

{
\subsection{Varying dependence structure}
\label{app:simu:vary_dependent}
In this subsection, we use simulations to investigate the effect of the dependence structure $m$. The simulation setting is the same as in Subsection \ref{app:simu:dependent}, except we fix the amplitude at $A=0.75$ and vary the dependence parameter $m \in \{1, 2, 3, 5, 8\}$ to observe its impact.

\begin{table}[h]
  \caption{Coverage rates in the mean–change model ($d=1$) with $m$–dependent errors.\label{tab:vary_m_optics}}
  \centering
  \begin{tabular}{cccccc}
    \hline
    $m$                          & 1          & 2          & 3          & 5          & 8          \\ \hline
    M-OPTICS(BS)                 & 0.84(2.97) & 0.86(2.52) & 0.79(2.47) & 0.81(2.59) & 0.78(2.86) \\
    M-OPTICS(SN)                 & 0.99(2.53) & 1.00(2.16) & 1.00(2.10) & 1.00(2.18) & 0.94(2.18) \\
    OPTICS(BS)                   & 0.83(2.96) & 0.08(2.87) & 0.13(3.04) & 0.01(1.51) & 0.00(1.06) \\
    OPTICS(SN)                   & 0.99(2.47) & 0.25(3.59) & 0.27(3.37) & 0.01(1.80) & 0.00(1.40) \\
    $\calA^1_{\text{COPSS}}$(BS) & 0.83       & 0.00       & 0.02       & 0.00       & 0.00       \\
    $\calA^1_{\text{COPSS}}$(SN) & 0.99       & 0.03       & 0.01       & 0.00       & 0.00       \\
    $\calA^1_{\text{FDRseg}}$    & 0.94       & 0.00       & 0.00       & 0.00       & 0.00       \\
    $\calA^1_{\text{SMUCE}}$     & 1.00       & 0.00       & 0.00       & 0.00       & 0.00       \\
    COPSS(BS)                    & 0.50       & 0.00       & 0.00       & 0.00       & 0.00       \\
    COPSS(SN)                    & 0.94       & 0.00       & 0.00       & 0.00       & 0.00       \\
    FDRseg                       & 0.78       & 0.00       & 0.00       & 0.00       & 0.00       \\
    SMUCE                        & 1.00       & 0.00       & 0.00       & 0.00       & 0.00       \\
    \hline
  \end{tabular}
\end{table}

From Table~\ref{tab:vary_m_optics}, we observe that while most methods perform adequately when $m=1$, their coverage rates collapse to near zero as the dependence increases ($m \ge 2$). In contrast, M-OPTICS consistently maintains high coverage rates across all tested values of $m$, demonstrating strong robustness to dependent errors. Furthermore, M-OPTICS(SN) achieves superior coverage (remaining near 1.00) while generally producing narrower confidence sets compared to M-OPTICS(BS).

\subsection{Comparison of OPTICS and Multiple-Splitting OPTICS}
\label{app:simu:ms-optics}
In this subsection, we evaluate the performance of two variants of our proposed methodology: the base OPTICS method and the multiple-splitting OPTICS (MS-OPTICS). The simulation setup mirrors the data generating process described in Subsection \ref{subsec:4.1}, focusing on a single change-point setting ($d=1$). To investigate the impact of sample size and error distribution on the confidence sets, we fix the signal amplitude at $A=0.75$ and vary the total sample size $N \in \{400, 800, 1000, 1600\}$. Furthermore, to assess the robustness of the methods, the error terms are generated from both a standard Normal distribution and a heavy-tailed $t_{10}$ distribution.

\begin{table}[h]
  \caption{Coverage rates (lengths) for amplitude $A=0.75$ across varying sample sizes $N$.\label{tab:vary_n_optics}}
  \centering
  \begin{tabular}{clcccc}
    \hline
    Distribution & Method & $N=400$ & $N=800$ & $N=1000$ & $N=1600$ \\ \hline
    \multirow{4}{*}{Normal} 
    & OPTICS(BS)    & 0.88(2.41) & 0.80(2.12) & 0.89(2.59) & 0.74(2.33) \\
    & MS-OPTICS(BS) & 0.88(3.59) & 0.92(2.47) & 0.90(2.67) & 0.87(2.56) \\
    & OPTICS(SN)    & 1.00(2.20) & 1.00(2.00) & 1.00(2.41) & 0.99(2.18) \\
    & MS-OPTICS(SN) & 0.98(3.45) & 0.99(2.09) & 1.00(2.36) & 1.00(1.99) \\ \hline
    \multirow{4}{*}{$t-{distribution}$} 
    & OPTICS(BS)    & 0.74(2.43) & 0.84(2.13) & 0.85(2.54) & 0.77(2.52) \\
    & MS-OPTICS(BS) & 0.88(4.10) & 0.83(2.53) & 0.91(2.67) & 0.90(2.62) \\
    & OPTICS(SN)    & 1.00(2.45) & 0.98(2.34) & 1.00(2.68) & 1.00(2.74) \\
    & MS-OPTICS(SN) & 0.93(4.08) & 0.99(2.42) & 0.99(2.53) & 0.99(2.40) \\ \hline
  \end{tabular}
\end{table}

Table~\ref{tab:vary_n_optics} summarizes the empirical coverage rates and average lengths of the estimated confidence sets. We observe that for moderate sample sizes (e.g., $N \le 1000$), both OPTICS and MS-OPTICS deliver comparable and highly satisfactory coverage rates across both distributions. As the sample size increases to $N=1600$, MS-OPTICS exhibits an added layer of stability, maintaining a coverages under the Normal distribution and the heavy-tailed $t_{10}$ distribution. This suggests that while the base OPTICS is highly effective and efficient for typical sample sizes, the multiple-splitting procedure can serve as a robust alternative when the sample size grows exceptionally large.
}

\subsection{Effectiveness of change-point detection algorithm}
\label{subsec:4.3}

As aforementioned in Section \ref{sec:intro}, the cardinality of OPTICS can evaluate the efficacy of various change-point detection algorithms when utilized as the base algorithms for constructing OPTICS. An efficient detection algorithm should demonstrate superior capability in distinguishing the true number of change-points from others. Therefore, OPTICS coupled with a powerful detection algorithm is expected to achieve the coverage rate with a smaller cardinality. In this subsection, we conduct a comparison between the effectiveness of SN and BS under the univariate mean change-point model setting in Section \ref{subsec:4.1} 

The left panel of Figure \ref{fig:car} depicts the boxplot of cardinalities of OPTICS using BS and SN as base algorithm across 100 simulation runs, where the error term follows $N(0,1)$. As the amplitude $A$ increases, the average cardinalities decrease, with OPTICS with SN exhibiting a notably faster decline. Additionally, the right panel of Figure \ref{fig:car} shows the coverage rates of OPTICS with the two base algorithms, where OPTICS with SN demonstrates a higher coverage rate compared to that with BS. Hence, under the univariate mean change-point model setting, the SN method proves to be more effective than the BS method. Figure \ref{fig:car_t} is devoted to the parallel results with $t(10)$ errors. Similar phenomena are observed. 

\begin{figure}[h]
	\centering
	\includegraphics[width=6cm, height=5cm]{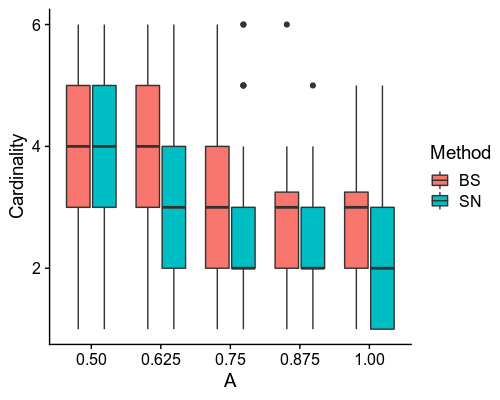}
	\includegraphics[width=6cm, height=5cm]{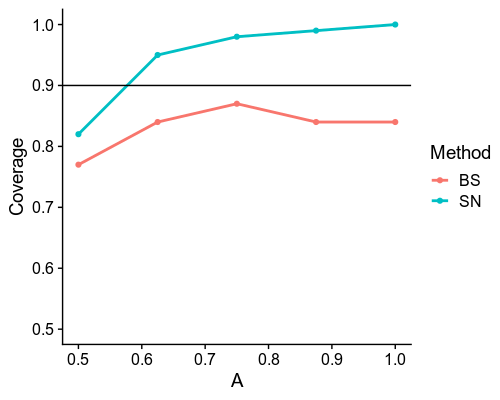}
	\caption{The boxplot of cardinalities of OPTICS and the line plot of coverage rates in univariate  mean change-point model: $N(0,1)$ error.}
	\label{fig:car}
\end{figure}

\begin{figure}[h]
	\centering
	\includegraphics[width=6cm, height=5cm]{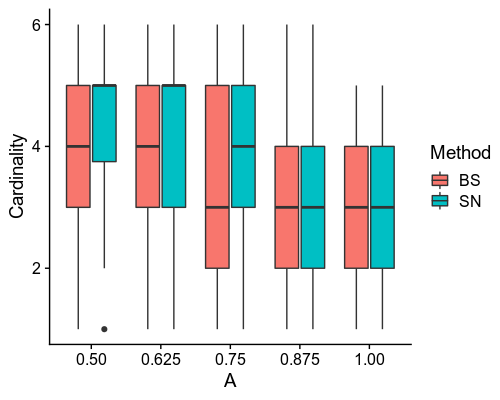}
	\includegraphics[width=6cm, height=5cm]{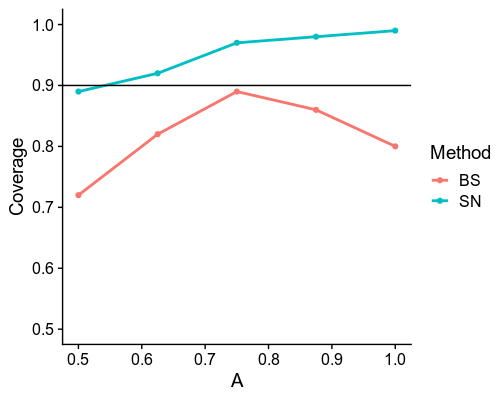}
	\caption{The boxplot of cardinalities of OPTICS and the line plot of coverage rates in univariate  mean change-point model: $t(10)$ error.}
	\label{fig:car_t}
\end{figure}

\section{More real data analysis results}\label{supp:real}
We provide the estimated change-position set $\calA_i, i=1,\ldots, 10$:
$$
\begin{aligned}
\mathcal{A}_1 & = \{13, 16\}, \\
\mathcal{A}_2 & = \{7, 10, 19, 22\}, \\
\mathcal{A}_3 & = \{4, 16, 19, 22\}, \\
\mathcal{A}_4 & = \{4, 7, 10, 13, 16, 19, 22\}, \\
\mathcal{A}_5 & = \{13, 22\}, \\
\mathcal{A}_6 & = \{4, 7, 10, 13, 16\}, \\
\mathcal{A}_7 & = \{4, 13, 16, 19, 22\}, \\
\mathcal{A}_8 & = \{22\}, \\
\mathcal{A}_9 & = \{19, 22\}, \\
\mathcal{A}_{10} & = \{22\}.
\end{aligned}
$$
The detected change-points for each individual are as follows:

$$
\begin{aligned}
\mathcal{S}_1 & = \{263, 341, 363, 388, 428, 449, 469, 1319, 1724, 1906, 2044, 2143, 2195\}, \\
\mathcal{S}_2 & = \{1642, 1663, 1771, 1795, 1816, 1965, 2195\}, \\
\mathcal{S}_3 & = \{540, 601, 2143, 2195\}, \\
\mathcal{S}_4 & = \{220, 2041, 2062, 2143\}, \\
\mathcal{S}_5 & = \{73, 174, 269, 1141, 1225, 1276, 1641, 1915, 1965, 1991, 2031, 2143, 2195\}, \\
\mathcal{S}_6 & = \{73, 105, 134, 2195\}, \\
\mathcal{S}_7 & = \{74, 134, 1572, 2195\}, \\
\mathcal{S}_8 & = \{177, 393, 521, 541, 601, 788, 811, 895, 932, 960, 1051, 1141, 1285, 1319, 1386, \\
             & \quad 1724, 1906, 1973, 1997, 2041, 2137, 2195\}, \\
\mathcal{S}_9 & = \{60, 221, 454, 521, 544, 581, 905, 925, 1029, 1054, 1141, 1225, 1249, 1378, 1522, \\
             & \quad 2047, 2071, 2141, 2195\}, \\
\mathcal{S}_{10} & = \{72, 134, 756, 1119, 1141, 1167, 1225, 1321, 1366, 1386, 1455, 1534, 1560, 1642, \\
                 & \quad 1685, 1726, 1818, 2044, 2091, 2143, 2166, 2195\}.
\end{aligned}
$$

% We provide the estimated change-position set when the number of change-points is taken as 5 in the real data analysis. The estimated set is 
% \begin{eqnarray*}
% 	\hat{\mathcal{S}}^5_{cpi}&=&\{05/2005, 08/2008, 11/2008, 01/2016, 01/2020\},\\
% 	\hat{\mathcal{S}}^5_{unemp}&=&\{09/2008, 11/2013, 03/2020,  09/2020, 08/2021\}.
% \end{eqnarray*} 
% And the changes are also described in Figure \ref{fig:cpi5}. 

% \begin{figure}[h]
% 	\centering
% 	\includegraphics[width=10cm, height=3.5cm]{.//figure//cpi_var5}
% 	\includegraphics[width=10cm, height=3.5cm]{.//figure//unem_mean5}
% 	\caption{Change-point detection for CPI and unemployment rate when $\hK=5$.{\label{fig:cpi5}}}
% \end{figure}

\bibliographystyle{Chicago}
\bibliography{wpref}
\end{document}